\documentclass{elsarticle}

\usepackage{lineno}
\modulolinenumbers[5]
\usepackage{xcolor}
\usepackage{hyperref}
\hypersetup{
    colorlinks,
    linkcolor={red!50!black},
    citecolor={blue!50!black},
    urlcolor={blue!80!black}
}
\usepackage{filecontents}
\usepackage{microtype}
\usepackage{pdfsync}
\usepackage[utf8]{inputenc}
\usepackage{mathtools}
\usepackage{amsmath}
\usepackage{amsthm}
\usepackage{thmtools}
\usepackage{thm-restate}
\usepackage{stackrel}
\usepackage{stmaryrd}
\usepackage{amssymb}
\usepackage{bbm}
\usepackage{mathabx}
\usepackage{tikz}
\usetikzlibrary{patterns,shapes,arrows,shadows,chains,matrix}
\usepackage[textwidth=2cm,textsize=small]{todonotes}
\usepackage{subcaption}
\usepackage{enumerate}
\usepackage{xspace}
\usepackage{paralist}
\usepackage[T1]{fontenc}
\usepackage{calligra}
\usepackage{cancel}
\usepackage[normalem]{ulem}
\usepackage{wrapfig}
\usepackage{multicol}
\usepackage{esvect}

\usepackage{cleveref}

\crefname{equation}{}{}

\newcommand{\CDR}{\mbox{\sc cdr}\xspace}
\newcommand{\ignore}[1]{}

\newcommand{\A}{{\cal A}}

\renewcommand{\iff}{\Leftrightarrow}

\newcommand{\condone}[1]{\mathbbm 1_{#1?}}

\newcommand{\LA}{{\sf LA}\xspace}
\newcommand{\PA}{{\sf PA}\xspace}

\newcommand{\Rgeq}{\R_{\geq 0}}
\newcommand{\Qgeq}{\Q_{\geq 0}}%{\Rgeq}%
\newcommand{\goesto}[1]{\xrightarrow{#1}}

\newcommand{\trule}[3]{\tuplesmall {#1 , #2, #3}}
\newcommand{\lang}[1]{L(#1)}
\newcommand{\languntimed}[1]{L^{\textrm{un}}(#1)}

\newcommand{\PDA}{\mbox{\sc pda}\xspace}
\newcommand{\TPDA}{\mbox{\sc tpda}\xspace}

\newcommand{\PDTA}{\mbox{\sc ptda}\xspace}
\newcommand{\dtPDA}{dt\PDA}
\newcommand{\TA}{\mbox{\sc ta}\xspace}
\newcommand{\RTA}{\mbox{\sc rta}\xspace}
\newcommand{\TRPDA}{\mbox{\sc trpda}\xspace}

\newcommand{\NFA}{\mbox{\sc nfa}\xspace}
\newcommand{\SRTA}{\mbox{\sc srta}\xspace}
\newcommand{\true}{\mathbf{true}}
\newcommand{\false}{\mathbf{false}}
\newcommand{\N}{\mathbb N}
\newcommand{\Z}{\mathbb Z}
\newcommand{\Q}{\mathbb Q}
\newcommand{\R}{\mathbb R}
\newcommand{\I}{\mathbb I}

\newcommand{\op}{\mathsf{op}}
\newcommand{\ops}{\mathsf{ops}}

\newcommand{\elapse}{\mathsf{elapse}}
\newcommand{\pushop}[2]{\push(#1 : #2)}
\newcommand{\popop}[2]{\pop(#1 : #2)}
\newcommand{\rewrite}{\mathsf{rewrite}}
\newcommand{\rewriteop}[3]{{\rewrite(#1 \to #2 : #3)}}
\newcommand{\fract}[1]{\left\{#1\right\}}
\renewcommand{\P}{\mathcal P}
\newcommand{\GG}{\mathcal G}
\newcommand{\QQ}{\mathcal Q}
\newcommand{\restrict}[2]{\left.#1\right|_{#2}}

\makeatletter
\renewcommand{\vec}[1]{\overline{#1}\@ifnextchar{^}{\,}{}\@ifnextchar{'}{\,}{}}
\makeatother

\newcommand{\sem}[1]{\llbracket#1\rrbracket}

\newcommand{\floor}[1]{\lfloor#1\rfloor}

\newcommand{\id}{\mathsf{id}}

\newcommand{\separator}{\;|\;}
\newcommand{\tuple}[1]{\left\langle#1\right\rangle}
\newcommand{\tuplesmall}[1]{\langle#1\rangle}
\newcommand{\eqv}[1]{\equiv_{#1}}
\newcommand{\eqvs}{(\eqv m)_{m\in\N}}

\newcommand{\reach}[2]{\stackrel{#1}{\leadsto}_{#2}} %{\leadsto_{#1}}%
\newcommand{\freach}[2]{\stackrel{#1}{\dashrightarrow}_{#2}} %{\leadsto_{#1}}%
\newcommand{\freachAE}[2]{\stackrel[\forall\exists]{#1}{\dashrightarrow}\!\!\mbox{}_{#2}\;}
\newcommand{\ReachSet}[1]{\mathsf{Reach}_{#1}}
\newcommand{\from}{\leftarrow}

\newcommand{\reset}{\mathsf{reset}}
\newcommand{\resetop}[1]{\reset(#1)}
\newcommand{\test}{\mathsf{test}}
\newcommand{\testop}[1]{\test(#1)}
\newcommand{\readletter}{\mathsf{read}}
\newcommand{\readop}[1]{\readletter(#1)}
\newcommand{\Resets}[1]{\mathsf{Reset}(#1)}

\newcommand{\push}{\mathsf{push}}
\newcommand{\pop}{\mathsf{pop}}
\newcommand{\cp}{\mathsf{copy}}
\newcommand{\psicopy}{\psi_\cp}

\newcommand{\psipop}{\psi_\pop}

\newcommand{\X}{\mathtt X}
\newcommand{\Y}{\mathtt Y}
\newcommand{\ZZ}{\mathtt Z}
\newcommand{\T}{\mathtt T}
\newcommand{\U}{\mathtt U}
\newcommand{\V}{\mathtt V}
\newcommand{\x}{x}%{\mathtt x}
\newcommand{\y}{y}%{\mathtt y}
\newcommand{\z}{z}%{\mathtt z}
\renewcommand{\L}{\mathtt L}

\newcommand{\set}[1]{\left\{ #1 \right\}}
\newcommand{\setof}[2]{\set{#1 \; \middle| \; #2}}

\newif\ifstartedinmathmode
\newcommand*{\st}{%s.t.~
  \relax\ifmmode\startedinmathmodetrue\else\startedinmathmodefalse\fi
  \ifstartedinmathmode{\;\cdot\;}\else{s.t.~}\fi%
}

\newcommand{\wlg}{w.l.o.g.~}

\newcommand{\wrt}{w.r.t.~}
\newcommand{\cf}{c.f.~}

\newcommand{\card}[1]{|{#1}|}

\newcommand{\PI}[2]{\text{\sc pi}(#1)_{#2}}
\newcommand{\project}[2]{\pi_{#1}(#2)}
\newcommand{\tick}[1]{\checkmark_{\!\!#1}}

\def\topbotatom#1{\hbox{\hbox to 0pt{$#1\bot$\hss}$#1\top$}}

\newcommand{\EXPTIME}{{\sf EXPTIME}\xspace}

\newcommand{\PTIME}{{\sf PTIME}\xspace}
\newcommand{\PSPACE}{{\sf PSPACE}\xspace}
\newcommand{\APSPACE}{{\sf APSPACE}\xspace}

\newcommand{\NLOGSPACE}{{\sf NL}\xspace}
\newcommand{\ALOGSPACE}{{\sf ALOGSPACE}\xspace}
\newcommand{\NP}{{\sf NP}\xspace}

\newcommand{\stkout}[1]{\ifmmode\text{\sout{\ensuremath{#1}}}\else\sout{#1}\fi}

\DeclareMathAlphabet{\mathcalligra}{T1}{calligra}{m}{n}
\Crefname{section}{Sec.}{Sec.}
\Crefname{figure}{Fig.}{Fig.}

\newtheorem{theorem}{Theorem}
\newtheorem{lemma}[theorem]{Lemma}
\newtheorem{fact}[theorem]{Fact}
\newtheorem{example}{Example}
\newtheorem{corollary}[theorem]{Corollary}
\newtheorem{remark}[theorem]{Remark}
\let\oldsection\section
\let\oldsubsection\subsection


\journal{Journal of Computer and System Sciences}

\begin{document}

\begin{frontmatter}

  \title{Reachability relations of timed pushdown automata\tnoteref{titlefn}}

  \tnotetext[titlefn]{This is an extended version of \cite{ClementeLasota:ICALP:2018}.}

  \author{Lorenzo Clemente\fnref{author1}}
  \fntext[author1]{Partially supported by Polish NCN grant 2016/21/B/ST6/01505.}
  \ead[author1]{clementelorenzo@gmail.com}
  \ead[author1]{https://mimuw.edu.pl/~lclemente/}

  %\orcid{0000-0003-0578-9103}

  \author{Sławomir Lasota\fnref{author2}}
  \fntext[author2]{Partially supported by the European Research Council (ERC) project Lipa under the EU’s Horizon 2020 research and innovation programme (grant agreement No. 683080).}
  \ead[author2]{sl@mimuw.edu.pl}
  \ead[author2]{https://www.mimuw.edu.pl/~sl/}

  \address{Wydział Matematyki, Informatyki i Mechaniki, University of Warsaw, Poland}

  \begin{abstract}
    Timed pushdown automata (\TPDA) are an expressive formalism
    combining recursion with a rich logic of timing constraints.
    We prove that reachability relations of \TPDA are expressible in linear arithmetic,
    a rich logic generalising Presburger arithmetic and rational arithmetic.
    The main technical ingredients are a novel quantifier elimination result for clock constraints (used to simplify the syntax of \TPDA transitions),
    the use of clock difference relations to express reachability relations of the fractional clock values,
    and an application of Parikh's theorem to reconstruct the integral clock values.
  \end{abstract}

  \begin{keyword}
    Timed automata \sep timed pushdown automata \sep reachability relation \sep clock difference relations \sep quantifier elimination
  \end{keyword}
  
\end{frontmatter}

\nolinenumbers

%\ccsdesc[500]{Theory of computation~Timed and hybrid models}
%\ccsdesc[500]{Theory of computation~Formal languages and automata theory}
%\ccsdesc[300]{Theory of computation~Logic}
%\ccsdesc[300]{Theory of computation~Logic and verification}
%\ccsdesc[100]{Theory of computation~Grammars and context-free languages}

% !TEX root = main.tex

\section{Introduction}

Timed automata (\TA) are one of the most studied models of reactive timed systems.
They extend classical finite automata with real-valued clocks which can be reset and compared by inequality constraints.
The fundamental algorithmic result in the field is Alur and Dill's proof of decidability (and in fact \PSPACE-completeness) of the reachability problem for \TA \cite{AD94},
for which they were awarded the Church Award in 2016 \cite{church:award}.
This result paved the way to the automatic verification of timed systems,
leading to industrial-strength tools such as UPPAAL \cite{Behrmann:2006:UPPAAL4} and KRONOS \cite{Yovine:KRONOS:1997}.
To this day, the reachability problem is a central algorithmic question which is the focus of intense research,
as testified by recent works such as \cite{HerbreteauSrivathsanWalukiewicz:IC:2016,AkshayGastinKrishna:LMCS:2018,GastinMukherjeeSrivathsan:CONCUR:2018,GastinMukherjeeSrivathsan:CAV:2019,GovindHerbreteauSrivathsanWalukiewicz:arXiv:2019}.

In certain applications, such as in parametric verification, deciding reachability between individual pairs of configurations is insufficient,
and one needs to construct the more general \emph{(binary) reachability relation},
i.e., the possibly infinite set of all pairs of configurations $(c, d)$
s.t.~there is an execution from $c$ to $d$.
The reachability relation for \TA has been shown to be effectively expressible in hybrid linear arithmetic with rational and integer sorts in a variety of works
\cite{ComonJurski:TA:1999,Dima:Reach:TA:LICS02,KrcalPelanek:TM:FSTTCS:2005,QuaasShirmohammadiWorrell:LICS:2017}.
This line of research very recently culminated in an extremely succinct and elegant proof \cite{FranzleQuaasShirmohammadiWorrell:2019}
expressing the \TA reachability relation as an existential formula of exponential size.
Since hybrid logic is decidable (and in fact its existential fragment is \NP-complete),
this yields an alternative proof of decidability of the reachability problem.
In the case of 1 clock \TA, one can even obtain a formula of polynomial size \cite{ClementeHofmanTotzke:CONCUR:2019},
yielding an optimal \NP algorithm for deciding simultaneous reachability in families of 1 clock \TA (which is \NP-hard).

In this paper, we compute the reachability relation for timed automata extended with a stack.
We propose to study the model of \emph{timed pushdown automata} (\TPDA),
which extends timed automata with a timed stack and rich set of clocks constraints.
The model features control clocks, as well as stack clocks.
As time elapses, all clocks both in the control and in the stack increase their values,
and they do so at the same rate.
Control clocks can be reset and compared against other control clocks using integral, fractional, and modular diagonal constraints.
At the time of a push operation, new stack clocks are created and pushed on the stack.
Their initial value is non-deterministically chosen as to satisfy a given push constraint between stack clocks and control clocks.
Push constraints are arbitrary Boolean combinations of integral, fractional, and modular diagonal constraints.
At the time of pop, stack clocks are compared to control clocks with analogous constraints.

It is important to remark that
the use of fractional constraints is crucial for the expressiveness of the model,
since \TPDA with just classical clock constraints
recognise the same class of timed languages
as \TPDA with untimed stack \cite[Theorem II.1]{ClementeLasota:LICS:2015}.
Uezato and Minamide have shown that this semantic collapse can be avoided by allowing fractional stack constraints \cite{UezatoMinamide:LPAR15};
%or, alternatively, by allowing modular stack constraints;
c.f.~\Cref{sec:discussion} for a detailed review of the literature on \TPDA and related models.
Since classical constraints can be expressed as combinations of integral and fractional clock constraints (c.f.~\Cref{rem:sugar}),
we consider integral and fractional constraints as the basic building blocks of \TPDA.
We also consider modular constraints since
\begin{inparaenum}[1)]
	\item they are not expressible as integral and fractional constraints (thus they increase the expressiveness of the model), and
	\item they can easily be handled by our solution technique
	with minimal overhead.
\end{inparaenum}

\oldsubsection*{Contributions}
Let $p,q$ be control locations, $\X = \set{\x_1, \dots, \x_n}$ the set of control clocks,
and $\Delta = \set{\delta_1, \dots, \delta_m}$ the set of transitions.
The \emph{reachability relation} of a \TPDA is the family of relations
\begin{align*}
  \reach {}  {pq} \ \subseteq\ \Qgeq^\X \times \Delta^* \times \Qgeq^\X,
\end{align*}
s.t.~from the initial clock valuation $\mu \in \Qgeq^\X$, control location $p$, and empty stack,
we can reach the final clock valuation $\nu \in \Qgeq^\X$, control location $q$, and empty stack,
by a sequence of transitions $w \in \Delta^*$, written $\mu \reach w {pq} \nu$.
The main contribution of the paper is a procedure for the effective description of the \TPDA reachability relation in the existential fragment of \emph{linear arithmetic},
i.e., first-order logic over the additive reals with integral ``$\floor{\_}$'' and fractional ``$\fract{\_}$'' operations.
Linear arithmetic is an expressive logic generalising Presburger $(\Z, \leq, \eqvs, +, 0)$ and rational arithmetic $(\R, \leq, +, 0)$,
and it is equi-expressive with the hybrid logic used in previous works on \TA reachability relations.
A formula of linear arithmetic $\varphi_{pq}(\x_1, \dots, \x_n, f_1, \dots, f_m, \x_1', \dots, \x_n')$
speaks about the initial values of clocks $\x_1, \dots, \x_n$,
their final values $\x_1', \dots, \x_n'$,
and the number of times $f_1, \dots, f_m$ that each transition $\delta_1, \dots, \delta_m$ is used in the run.
Formally, a run $\mu \reach w {pq} \nu$ \emph{satisfies} a linear arithmetic formula $\varphi_{pq}$ if
\begin{align*}
	\bar x: \mu, \bar f: \PI w, {}, \bar x': \nu \models \varphi_{pq}(\bar x, \bar f, \bar x'), 
\end{align*}
where $\PI w {} : \N^\Delta$ is the \emph{Parikh image} of $w \in \Delta^*$,
i.e., $\PI w {\delta_i}$ is the number of occurrences of $\delta_i$ in $w$.
%
%where $\mu$ models the initial clock values $\bar x$.
%$\PI w {}$ models the transition counts $\bar f$,
%and $\nu$ models the final clock values $\bar x'$.
%
The reachability relation \emph{is expressed} by a family of formulas $\set{\varphi_{pq}}_{p, q}$ if,
for all control locations $p, q$, clock valuations $\mu, \nu : \Rgeq^\X$, and sequence of transitions $w \in \Delta^*$,
$\mu \reach w {pq} \nu$ holds if, and only if, it satisfies $\varphi_{pq}$.

\paragraph{Main result}

The following is the main result of the paper.
\begin{theorem}
	\label{thm:TPDA}
	The reachability relation of a \TPDA is expressed by a family of formulas of linear arithmetic.
	%computable in double exponential time.
\end{theorem}

This properly generalises the previous expressibility results on \TA in several ways.
First, the kind of reachability relation that we consider is \emph{ternary} because it takes into account not only the initial and final clock values,
but also (the Parikh image of) the transitions executed during the run.
The binary reachability relation considered in previous works on \TA
takes into account only initial and final clock values
and can be obtained as a special case with the following formula of linear arithmetic:
\begin{align*}
	\exists \bar f \st \varphi_{pq}(\bar x, \bar f, \bar x').
\end{align*}
As an application of the ternary reachability relation we can count, for instance,
the number of symbols in the stack,
which is not possible with the binary reachability relation alone.
To this end, for a \TPDA transition $\delta_i$ and stack symbol
$\alpha_j \in \Gamma = \set{\alpha_1, \dots, \alpha_\ell}$ let
\begin{align*}
	c_{i, j} = \left\{
	\begin{array}{ll}
		1	& \text{ if $\delta_i$ pushes $\alpha_j$ on the stack,} \\
		-1	& \text{ if $\delta_i$ pops $\alpha_j$ from the stack, and} \\
		0	& \text{ otherwise.}
	\end{array}
	\right.
\end{align*}
Let $\bar \partial = \tuple{\partial_1, \dots, \partial_\ell}$
be a vector of integer variables
denoting the total variation of the number of occurrences of each stack symbol.
We can then model the variation of stack symbols
with the following formula of linear arithmetic
\begin{align*}
	\psi_{pq}(\bar x, \bar f, \bar \partial, \bar \x')
		\equiv
			\bigwedge_{j = 1}^\ell \partial_j = \sum_{i=1}^m c_{i, j} \cdot f_i \land
				\varphi_{pq}(\bar x, \bar f, \bar y).
\end{align*}
Second, \TPDA are more expressive than \TA because of the presence of the (timed) stack, which is an unbounded data structure.
Finally, the kind of clock constraints that we consider mixing integral, modular, and fractional constraints
is very expressive and was not previously considered for \TA.
%
%\footnote{While fractional/modular constraints do not increase the expressive power of \TA, this is not the case for \TPDA, as we discuss in \Cref{sec:discussion}.}.

\paragraph{Quantifier elimination for clock constraints}

The other contributions of the paper are of a more technical nature
and arise from the methodology used to prove \Cref{thm:TPDA}.
More precisely, the computation of the reachability relation is achieved by a sequence of translations
progressively simplifying the kind the clock constraints allowed in the automaton.
A crucial ingredient in our reductions is a novel \emph{quantifier elimination} result for the fragment of linear arithmetic corresponding to clock constraints (cf.~\Cref{lem:qe-clocks}),
which is another contribution of this work, possibly of independent interest.
While linear arithmetic is known to have elimination of quantifiers \cite{Weispfenning:1999:MRL,BoigelotJodogneWolper:TOCL:2005}
(and likewise for Presburger \cite{Presburger:1930} and rational arithmetic \cite{FerranteRackoff:QE:Reals}),
our result is stronger since we transform a quantifier clock constraints into a logically equivalent (quantifier-free) \emph{clock constraint},
instead of an arbitrary quantifier-free formula of linear arithmetic
(as a generic quantifier elimination procedure would do \cite{Weispfenning:1999:MRL,BoigelotJodogneWolper:TOCL:2005}).

The language of quantified clock constraints that we consider is very close to the so called \emph{difference logic}
(whose relevance in program verification was first noted by Pratt in 1977 \cite{Pratt1977}),
which is the first order theory of the reals with atomic formulas of the form $x_i \sim c$ and $x_i - x_j \sim c$
with $\sim$ a comparison operator in $\set{<, >, =}$ and $c \in \Q$ a rational constant.
There are two variants of difference logic,
depending on whether it is interpreted over the integers $\Z$ or over the reals $\R$;
both variants admit quantifier elimination \cite{Koubarakis:PKRR:1994}.
However, the integral variant of difference logic does not have modulo constraints,
and adding modulo constraints strictly increases its expressive power.
In fact, while for full Presburger arithmetic modulo constraints such as $x_i - x_j \eqv m k$ ($k \in \Z$) do not increase the expressive power,
since they can be expressed as $\exists z \st x_i - x_j = k + m \cdot z$,
the latter formula is not a formula of difference logic.
On the other hand, the rational variant of difference logic is more expressive than the fractional fragment of the quantified clock constraint that we consider,
since it allows arbitrary rational constants $c \in \Q$ to appear in the formula, while we allow only the constant $0$.
Thus, our quantified clock constraints are incomparable with difference logic,
and consequently our quantifier elimination result does not follow from the corresponding result for difference logic.

Our sequence of transformations produces a so-called \emph{fractional} \TPDA,
i.e., one which uses only fractional constraints.
In order to reconstruct the full reachability relation from a fractional \TPDA
we follow \cite{QuaasShirmohammadiWorrell:LICS:2017,FranzleQuaasShirmohammadiWorrell:2019} and encode integral clock values in the language of the automaton.
This is the technical reason why ternary reachability is more convenient than mere binary reachability in our setting.

\paragraph{Quantifier elimination for clock difference relations}

In the last step, we compute the reachability relation of a fractional \TPDA
by constructing a context-free grammar recognising precisely the sequence of transitions $w$ labelling its executions $\mu \reach {w} {pq} \nu$.
This step uses Parikh's theorem applied to the grammar
in order to compute a small existential formula of Presburger arithmetic
expressing the Parikh image of language it recognises.
We represent the reachability relation between the fractional values of clocks (which is the only relevant quantity for a fractional \TPDA)
by the so-called \emph{clock difference relations} (\CDR),
which are the fragment of quantifier-free rational arithmetic generated by atomic formulas of the form $u \leq v$,
where $u, v$ are terms of the form $\fract{x_0 - x_i}$, $\fract{x_0}$, $\fract{x_0' - x_i}$, and $\fract{x_0'}$.
Modulo some presentational details, \CDR were previously introduced to compute the binary reachability relation for \TA \cite{KrcalPelanek:TM:FSTTCS:2005},
and even in the analysis of communicating timed automata \cite{KrcalYi:CTA:CAV:2006}.
We show in \Cref{cor:CDR:composition} that \CDR are closed under relational composition,
which is a consequence of a novel quantifier elimination result for \CDR, another technical contribution of this work.
This allows us to build a \CDR expressing the fractional reachability relation by iteratively composing \CDR representing shorter runs, until no new \CDR are produced (up to logical equivalence).

\paragraph{Untiming of \TPDA languages are context-free}

Since all our transformation essentially preserve the untiming of the \TPDA language,
we obtain as a corollary that such languages are context-free.
This is similar as for the untiming of timed automata languages,
which are regular \cite{AD94}.

\begin{corollary}
  The untiming of a \TPDA language is context-free.
\end{corollary}

\oldsubsection*{Organisation}

This paper is an extended version of \cite{ClementeLasota:ICALP:2018}.
With respect to the conference version, we provide full proofs of all the formal constructions.
Moreover, the treatment of fractional \TPDA has been substantially simplified by the use of clock difference relations,
thus making the paper entirely self-contained and avoiding the introduction of register automata.

We start in Sec.~\ref{sec:logic} with basic notions on linear arithmetic,
clock constraints, clock difference relations, and fundamental quantifier elimination results for these logics. 
In Sec.~\ref{sec:TPDA} we introduce the model of timed pushdown automata (\TPDA),
and in Sec~\ref{sec:overview} we present an overview of the reductions leading to \Cref{thm:TPDA}.
%and state our main result that \TPDA reachability relations are expressible in linear arithmetic.
%
The reductions themselves are presented in \Cref{sec:simplify},
%\Cref{sec:simplify:atomic-pop,sec:simplify:fractional,sec:simplify:pop-integer-free,sec:simplify:push-copy,sec:simplify:stack-stack,sec:simplify:reset}
which progressively simplify the shape of control and stack constraints of the automaton until we eventually obtain a fractional \TPDA.
In Sec.~\ref{sec:fractional:TPDA}, the reachability relation of a fractional \TPDA is reduced to the Parikh image of a context-free grammar.
In Sec.~\ref{sec:complexity} we analyse the complexity of our construction,
and in Sec.~\ref{sec:discussion} we provide an extensive comparison between \TPDA and related models from the literature.
In Sec.~\ref{sec:conclusions} we conclude with some perspectives for further research.
Proofs of the more technical statements are provided \ref{app:proofs}
in order not to disrupt the flow of the presentation.

% !TEX root = main.tex

\section{Quantifier elimination}
\label{sec:logic}

\paragraph{Notation}
We denote by $\N$, $\Z$, $\Q$, and $\Qgeq$
the set of, resp., natural, integer, rational, and nonnegative rational numbers.
Let $\I = \Q \cap [0, 1)$ be the unit rational interval.
Let $\eqv m$ denote the congruence modulo $m \in \N\setminus\set 0$ in $\Z$.
For $a\in \Q$, let $\floor a \in \Z$ denote the largest integer $k$ \st $k \leq a$,
and let $\fract a = a - \floor a$ denote its fractional part.

Let $\condone C$, for a condition $C$, be $1$ if $C$ holds, and $0$ otherwise.
This will be used primarily in the following elementary properties of integral/fractional arithmetic:
For every $a, b \in \Rgeq$,
\begin{align}
	\label{eq:floor}
	\floor {a + b} &= \floor a + \floor b + \condone {\fract a + \fract b \geq 1}, 
		&\floor {a - b} &= \floor a - \floor b - \condone {\fract a < \fract b}, \\
	%
% not needed in the rest
%		\label{eq:floor:2}
%		\floor x - \floor y &= \left\{\begin{array}{ll}
%			\floor {x - y} 		& \textrm{ if } \fract x \geq \fract y \\
%			\floor {x - y} + 1	& \textrm{ otherwise}.
%		\end{array}\right. \\
	%
	\label{eq:fract}
	\fract {a + b} &= \fract a + \fract b - \condone {\fract a + \fract b \geq 1},
		&\fract {a - b} &= \fract a - \fract b + \condone {\fract a < \fract b}.
\end{align}

\paragraph{Linear arithmetic}

We call \emph{linear arithmetic} the first order language in the vocabulary of the structure (cf.~\cite{Weispfenning:1999:MRL})
$$\A = (\R, \leq, \eqvs, +, \floor{\_}, \fract{\_}, (k \cdot \_)_{k\in\Z}, 0, 1).$$
The symbol ``$+$'' is interpreted as the binary sum function,
and ``$\floor{\_}$'' is the unary integral part operator,
and ``$\fract{\_}$'' is the unary fractional part operator.
For every integer $k \in \Z$ we have a unary function ``$k \cdot \_$'' which multiplies its argument by $k$.
The formula $u \eqv m v$ means that $u - v$ is an integer multiple of $m$;
in particular, $\fract u = \fract v$.
While not formally part of the vocabulary, we allow to write expressions of the form $u - v$ as syntactic sugar.
We assume that constants are encoded in binary.

Linear arithmetic restricted to the integers $\Z$ (and without the fractional part operator)
is commonly known as \emph{Presburger arithmetic} \cite{Presburger:1930},
and removing the modulo and integral part operator yields \emph{rational arithmetic} \cite{FerranteRackoff:QE:Reals}.
Both the former sublogics admit elimination of quantifiers,
and the same holds for linear arithmetic \cite{Weispfenning:1999:MRL,BoigelotJodogneWolper:TOCL:2005}.
Regarding the complexity of the satisfiability problem,
it is \NP-complete for existential fragments of both Presburger \cite{Weispfenning:JSC:1988}
and rational arithmetic \cite{Sontag:IPL:1985}.
The same complexity can be shown for existential linear arithmetic.

\begin{theorem}[\protect{\cite[Theorem 3.1]{BoigelotJodogneWolper:TOCL:2005}, \cite[Theorem 1]{ClementeHofmanTotzke:CONCUR:2019}}]
	\label{thm:ELA:NP}
	The satisfiability problem for existential linear arithmetic is \NP-complete.
\end{theorem}
\noindent
The result above is not obvious: The satisfiability problem for linear arithmetic
does not immediately reduce (under polynomial time Turing reductions)
to the same problem for Presburger and rational arithmetic,
since atomic formulas can mix together integral ``$\floor{\_}$'' and fractional ``$\fract{\_}$'' operators.
However, it is possible to separate integral and fractional operators with a polynomial blow-up in the formula size,
and preserving the existential fragment.
For completeness, we present a particularly short proof of this fact from \cite{ClementeHofmanTotzke:CONCUR:2019}.

\begin{proof}[Proof (of \Cref{thm:ELA:NP})]
	By introducing linearly many new existentially quantified variables and suitable defining equalities,
	we assume w.l.o.g.~that there are no modulo constraints,
	and that terms are \emph{shallow}, in the sense that they are generated by the following restricted grammar
	(where $x, y$ denote variables):
	\begin{align*}
		s, t \; ::= \; x \separator k \separator \floor x \separator \fract x \separator {-x} \separator x + y \separator k \cdot x.
	\end{align*}
	Since now we have atomic propositions of the form $s \leq t$ %, s \equiv_k t$
	with $s, t$ shallow terms,
	we assume that we have no terms of the form ``$-x$''
	(by moving it to the other side of the relation
	and possibility introducing a new existential variable to make the term shallow again).
	Moreover, we can also eliminate terms of the form $k\cdot x$
	by introducing $O(\log k)$ new existential variables and using iterated doubling (based on the binary expansion of $k$);
	for instance, $5\cdot x$ is replaced by $x_0 + x_2$,
	by adding new variables $x_0, \dots, x_2$ and equalities
	$x_0 = x$, $x_1 = x_0 + x_0$, $x_2 = x_1 + x_1$.
	We end up with the following further restricted syntax of terms:
	\begin{align*}
		s, t \; ::= \; x \separator k \separator \floor x \separator \fract x \separator x + y.
	\end{align*}
	We replace atomic propositions of the form $s \leq t$ by the equivalent formula
	\begin{align}
		\label{eq:le:expand}
		\floor s < \floor t \vee (\floor s = \floor t \wedge \fract s \leq \fract t).
	\end{align}
	%
	%In this way, each term is guarded by the integral or fractional part.
	%
	We push the integral $\floor \_$ and fractional $\fract \_$ operations inside terms,
	according to \eqref{eq:floor}, \eqref{eq:fract}, and the following rules:
	\begin{align*}
%		\floor x
%			&\to \floor x &
%		\fract x
%			&\to \fract x
%		\\
		\floor k &\to k
			&\floor {\floor x} &\to \floor x
				&\floor {\fract x} &\to 0 \\
		\fract k &\to 0
			&\fract {\floor x} &\to 0
				&\fract {\fract x} &\to \fract x.
	\end{align*}
	%
	%It remains to consider sums $x + y$.
	%We perform a case analysis on $\fract x + \fract y$:
	%
	%\begin{align*}
	%	s \sim \floor {x + y}
	%		\to\ 	&(s \sim \floor x + \floor y \wedge \fract x + \fract y < 1)\ \vee \\
	%				&(s \sim \floor x + \floor y + 1 \wedge \fract x + \fract y \geq 1), \\
	%	s \sim \fract {x + y}
	%		\to\ 	&(s \sim \fract x + \fract y \wedge \fract x + \fract y < 1)\ \vee \\
	%				&(s \sim \fract x + \fract y - 1 \wedge \fract x + \fract y \geq 1).
	%\end{align*}
	%
	The terms of the form $\condone{\_}$ introduced by \eqref{eq:floor} and \eqref{eq:fract} are subsequently removed by expanding their definition.
	We thus obtain a logically equivalent \emph{separated formula},
	i.e., one where integral $\floor x$ and fractional $\fract y$ variables never appear together in the same atomic formula.
	Since we only added existentially quantified variables in the process,
	the resulting formula is still in the existential fragment,
	which can be decided in \NP by appealing to decision procedures for Presburger and rational arithmetic.
\end{proof}

%\emph{Difference arithmetic} (\DA) (also known as \emph{difference logic})
%is the fragment of \LA where atomic propositions are restricted to be of the form
%$x - y \leq k$,
%$\floor x - \floor y \bowtie k$ ,
%or $\fract x - \fract y \bowtie k$,
%where $\bowtie\ \in \set{\leq, \eqv 1, \eqv 2, \dots}$, $k \in \Z$.
%
%Two natural fragments of \DA are \emph{integral \DA} (\IDA, a fragment of Presburger arithmetic)
%and \emph{rational \DA} (\RDA, a fragment of rational arithmetic).

%Two formulas are \emph{equivalent} if they are satisfied by the same valuations.
%It is well-known that $\logicint$~\cite{Presburger:1930} and $\logicrat$~\cite{FerranteRackoff:QE:Reals}
%admit effective elimination of quantifiers:
%Every formula can effectively be transformed to an equivalent quantifier-free one.

\paragraph{Clock constraints}

Let $\X$ be a finite set of clocks.
%We denote a clock by $\x$, and the corresponding logical variable used in formulas by $x$.
%
We consider constraints which can separately speak about
the integer $\floor \x$ and fractional value $\fract \x$ of a clock $\x \in \X$.
A \emph{clock constraint} over $\X$ is a Boolean combination of
\emph{atomic clock constraints} of one of the forms
\begin{align*}
%	\begin{tabular}{}
  &
		&\textrm{(class}&\textrm{ical)}
    &\textrm{(inte}&\textrm{gral)}
    &\textrm{(modu}&\textrm{lar)}
    &\textrm{(fracti}&\textrm{onal)}
  \\[1ex]
  &\textrm{(non-diagonal)}	&\x		&\leq k	\ \
		&\floor \x & \leq k 	\ \  			
			&\floor \x & \eqv m k  	\ \   	 	
				&\fract \x & = 0   \\
  &\textrm{(diagonal)}	&\x - \y &\leq k
		&\floor \x - \floor \y & \leq k	
			&\floor \x - \floor \y & \eqv m k	
				&\fract \x & \leq \fract \y
%\end{tabular}
\end{align*}
where
$\x, \y \in \X$, $m \in \N$, and $k \in \Z$.
A clock constraint is \emph{conjunctive} if it is of the form $\bigwedge_{i=1}^n \varphi_i$,
where each $\varphi_i$'s is a (possibly negated) atomic clock constraint.
%
%Since we allow arbitrary boolean combinations,
%we consider also the constraint $\true$, which is always satisfied,
%and variants with any $\sim \ \in \{\leq, <, \geq, >\}$ in place of $\leq$.
%
A \emph{clock valuation} %over a set of clocks $X$
is a mapping $\mu \in \Qgeq^\X$ assigning a non-negative rational number to every clock in $\X$;
we write $\floor \mu$ for the valuation in $\N^\X$ \st $\floor \mu(\x) := \floor {\mu(\x)}$
and $\fract \mu$ for the valuation in $\I^\X$ \st $\fract \mu(\x) := \fract {\mu(\x)}$.
For a valuation $\mu$ and a clock constraint $\varphi$, %over clocks $x_1, \dots, x_k$, we write $\mu \models \varphi$
we say that $\mu$ \emph{satisfies} $\varphi$
if $\varphi$ is satisfied when integer clock values $\floor \x$ are evaluated according to $\floor \mu$
and fractional values $\fract \x$ according to $\fract \mu$.
For a clock valuation $\mu$ and a set of clocks $\Y \subseteq \X$,
let $\mu[\Y \mapsto 0]$ be the same as $\mu$ except that clocks in $\Y$ are mapped to $0$,
and let $\restrict \mu \Y \in \Qgeq^\Y$ be the restriction of $\mu$ to $\Y$.
For $\delta \in \Qgeq$, let $\mu + \delta$ be the clock valuation which adds $\delta$ to the value of every clock,
i.e., $(\mu + \delta)(\x) := \mu(\x) + \delta$ for every $\x \in \X$.

\begin{remark}[Classical clock constraints]
	\label{rem:sugar}
	%Integral and fractional constraints subsume classical ones.
	%Since $\x = \floor \x + \fract \x$ (and similarly for $\y$)%
	%
	%\footnote{We often identify a clock $\x$ with its value in a clock valuation for simplicity of notation.},
	%
	Classical constraints can be expressed in terms of integral and fractional ones (c.f.~\eqref{eq:floor}):
	\begin{align*}
		\x - \y \leq k
			\quad \textrm{if, and only if,} \quad
				(\floor \x - \floor \y \leq k \wedge \fract \x \leq \fract \y) \vee \floor \x - \floor \y \leq k - 1,
	\end{align*}
	%
%	\ilorenzo{added the reverse traslation, which is used in \Cref{sec:simplify:pop-integer-free}}%
	and, vice versa,
	integral constraints can be expressed in terms of classical and fractional ones:
	\begin{align*}
		\floor \x - \floor \y \leq k
			\quad \textrm{if, and only if,} \quad
			\begin{array}{c}
				\x - \y \leq k \vee
				(\x - \y  \leq k + 1 \wedge \fract \x > \fract \y).
			\end{array}
	\end{align*}
	%\islawek{corrected}
	%and similarly for the non-diagonal constraint $\x\leq k$.
	%
	On the other hand, the fractional constraint $\fract \x = 0$
	is expressible neither as a classical constraint
	nor as an integral one.
\end{remark}

\begin{remark}[$\floor \x - \floor \y$ versus $\floor {\x - \y}$]
	\label{rem:int-diff}
	In the presence of fractional constraints,
	the expressive power would not change if,
	instead of terms $\floor \x - \floor \y$
	speaking of the \emph{difference of the integer parts},
	we would instead choose terms $\floor{\x - \y}$
	speaking of the \emph{integer part of the difference},
	since the two are inter-expressible by \eqref{eq:floor}.
	%We will also use the following relationships:
	%\begin{align}
	%	\label{eq:sum}
	%	\floor{a + b} = \floor a + \floor b + \condone {\fract a + \fract b \geq 1}
	%	\quad \textrm{ and } \quad
	%	\fract{a + b} = \fract a + \fract b - \condone {\fract a + \fract b \geq 1}.
	%\end{align}
\end{remark}

\paragraph{Shifts}

While classical diagonal constraints $\varphi(\x, \y) \equiv \x - \y \leq k$ are invariant \wrt the elapse of time,
in the sense that $\varphi(\x + \delta, \y + \delta)$ is equivalent to $\varphi(\x, \y)$,
this is not the case for the other kind of constraints.
However, the \emph{class} of clock constraints is closed under such shifts.

\begin{lemma}
	\label{lem:shift}
	For any clock constraint $\varphi(x_1, \dots, x_n)$ and a fresh variable $x_0$,
	%we have that
	%
	$$\varphi(x_1 - x_0, \dots, x_n - x_0)$$
	is also expressible as a clock constraint, with a linear size blow-up.
\end{lemma}

\begin{proof}

	For classical constraints the claim is obvious.
	For a fractional constraint $\fract {x - x_0} \leq \fract {y - x_0}$,
	by applying \eqref{eq:fract} on both sides, we obtain
	$\fract {x} - \fract {x_0} + \condone {\fract {x} < \fract {x_0}} \leq
		\fract {y} - \fract {x_0} + \condone {\fract {y} < \fract {x_0}}$.
	By doing a case analysis on all possible total orderings of the fractional values,
	we have
	\begin{align*}
		\fract {x - x_0} \leq \fract {y - x_0} \textrm{ iff }
			\fract{x_0} \leq \fract{x} \leq \fract{y} \vee
				\fract{x} \leq \fract{y} < \fract{x_0} \vee
					\fract{y} < \fract{x_0} \leq \fract{x}.
	\end{align*}
	For integral constraints, by applying \eqref{eq:floor} twice, we have
	\begin{align*}
			\floor {x - x_0} - \floor {y - x_0} = \floor {x} - \floor{y} + \condone {\fract y < \fract{x_0}}
			-\condone {\fract x < \fract {x_0}}.
			\tag*{\qedhere}
	\end{align*}
	%\ilorenzo{corrected}
\end{proof}

\subsection{Quantifier elimination}
\label{sec:qe}

In this section we show that quantified clock constraints admit effective elimination of quantifiers.

\begin{lemma}
	\label{lem:qe-clocks}
	Quantified clock constraints admit effective elimination of quantifiers. % in \PSPACE.
	For conjunctive formulas, we can produce an equivalent quantifier-free clock constraint of exponential size in disjunctive normal form,
	where each disjunct is a conjunctive clock constraint of polynomial size.
\end{lemma}

\noindent
Clock constraints are a sublogic of quantifier-free \LA, where variables take only nonnegative values.
Since the latter logic admits elimination of quantifiers \cite{Weispfenning:1999:MRL,BoigelotJodogneWolper:TOCL:2005},
it follows that every quantified clock constraint $\varphi$
admits an equivalent quantifier-free \emph{\LA formula} $\psi$.
Our result above is stronger,
because we show that $\psi$ is a \emph{clock constraint},
instead of an arbitrary quantifier-free \LA formula.
%
%where fractional constraints can be only of the form $\fract \x - \fract \y \leq k$ with $k = 0$.
%It can also be seen as the logic of the two sorted structure
%
%\begin{align*}
%  (\N, \leq, \eqvs, +1, 0) \uplus (\unitint, \leq, 0).
%\end{align*}
%We denote by $\logicclock$ the sub-logic of \logic corresponding to clock constraints,
%and as before $\logicclockint$ and $\logicclockrat$ are the restrictions to the respective sorts.
%
%All the sub-logics above admit effective elimination of quantifiers,
%as shown below%
%
Lemma~\ref{lem:qe-clocks} follows directly from Lemmas~\ref{lem:qe-clocks-int-rat} and \ref{lem:qe-rat} below,
which take care of the integral, resp., fractional clock constraints.

\begin{lemma}
	\label{lem:qe-clocks-int-rat}
	The structure % $\AZsucczero$
	$(\N, \leq, \eqvs, +1, 0)$ admits effective elimination of quantifiers.
	The complexity is singly exponential for conjunctive formulas.
\end{lemma}
\noindent
An analogous result was obtained in \cite[Theorem 4.5]{Koubarakis:PKRR:1994} for \emph{integral difference logic},
i.e., where the atomic formulas are of the form $x - y \bowtie k$,
with $\bowtie \in \set{<, \leq, >, \geq, =}$ and $k \in \Z$.
Our context is slightly different because we also consider modular constraints,
strictly increasing the expressive power of difference logic.
(Interestingly, while quantifier elimination in \PA requires the introduction of modular constraints,
this is not the case for difference logic.)

\begin{proof}[Proof (of \Cref{lem:qe-clocks-int-rat})]
	We assume that all modulo statements are over the same modulus $m$.
	%which is easily achieved by taking the least common multiplier of all moduli and by introducing some extra disjunctions.
	%
	%In order to show quantifier elimination,
	It suffices to consider a formula $\exists y \cdot \varphi$
	where $\varphi$ is a conjunctive formula of the form %(we write $y, x_i$ instead of $\floor y, \floor {x_i}$ for simplicity)
	\begin{align}
		\label{eq:integer:qe}
		\bigwedge_i x_i + \alpha_i \leq y \leq x_i + \beta_i \ \wedge \ y \eqv m x_i + \gamma_i,
	\end{align}
	s.t., for every $i$,
	$\alpha_i, \beta_i \in \Z \cup \set {-\infty, +\infty}$ with $\alpha_i \leq \beta_i$, $\gamma_i \in \set {0, \dots, m - 1}$.
	For uniformity of notation we assume $x_0 = 0$ in order to model non-diagonal constraints on $y$.
	If not all $\alpha_i$'s are equal to $-\infty$,
	then a satisfying $y$ will be of the form $x_j + \alpha_j + \delta$ with $\delta \in \set {0, \dots, m - 1}$
	where $j$ maximises $x_j + \alpha_j$.
	The following quantifier free formula $\widetilde\varphi$ is equivalent to \eqref{eq:integer:qe}:
	\begin{align}
		\label{eq:quantifier_free}
		\!\!\!\!\!\bigvee_{\delta \in \set {0, \dots, m - 1}} \!\! \bigvee_j
		\bigwedge_i x_i + \alpha_i \leq x_j + \alpha_j + \delta \leq x_i + \beta_i \wedge x_j + \alpha_j + \delta \eqv m x_i + \gamma_i.\!\!
	\end{align}
	For the complexity claim, $\widetilde\varphi$ is exponentially bigger than \eqref{eq:integer:qe} when constants are encoded in binary.
	%
	%Clearly, the formula above is a constraint.
	%
	For the inclusion $\sem{\widetilde\varphi} \subseteq \sem{\exists y \cdot \varphi}$,
	let $(a_1, \dots, a_n) \in \sem{\widetilde\varphi}$.
	There exist $\delta$ and $j$ as per \eqref{eq:quantifier_free},
	and thus taking $a_0 := a_j + \alpha_j + \delta$ yields
	$(a_0, a_1, \dots, a_n) \in \sem{\exists y \cdot \varphi}$.
	For the other inclusion,
	let $(a_0, a_1, \dots, a_n) \in \sem{\varphi}$.
	Let $j \neq 0$ be \st $a_j + \alpha_j$ is maximised,
	and define $\delta := a_0 - (a_j + \alpha_j) \mod m$.
	Clearly $\delta \geq 0$ since $a_0$ satisfies all the lower bounds $a_i + \alpha_i$.
	Since $a_0$ satisfies all the upper bounds $a_i + \beta_i$ and $a_j + \alpha_j + \delta \leq a_0$,
	upper bounds are also satisfied.
	Finally, since $a_0 \eqv m a_i + \gamma_i$ and $a_0 \eqv m a_j + \alpha_j + \delta$,
	also the modular constraints $a_j + \alpha_j + \delta \eqv m a_i + \gamma_i$ are satisfied.
	Thus, $(a_1, \dots, a_n) \in \sem{\widetilde\varphi}$, as required.

	If all $\alpha_i$'s are equal to $-\infty$, then there are no lower bound constraints and only modulo constraints remain, hence
	%\slawek{this case assumed $\N$; I've adapted to $\Z$}
	%	the following quantifier free formula $\widetilde\varphi$ is equivalent to \eqref{eq:integer:qe}:
	%	\[
	%		\bigwedge_{i\neq 1} \ x_1 + \gamma_1 \eqv m x_i + \gamma_i.
	%	\]
	%	The complexity is polynomial in this case.
	%
	and a satisfying $y$ (if it exists) can be taken in the interval $\set{0, \dots, m - 1}$, yielding
	\begin{align*}
		\bigvee_{\delta \in \set {0, \dots, m - 1}} \bigwedge_i \ \delta \leq x_i + \beta_i \ \wedge \
		\delta \eqv m x_i + \gamma_i.
	\end{align*}
	The same complexity holds.
	The formula above is shown equivalent to \eqref{eq:integer:qe} by reasoning as in the previous paragraph.
\end{proof}

\begin{lemma}
  \label{lem:qe-rat}
	The structure $(\I, \leq, 0)$ admits effective elimination of quantifiers.
	The complexity is quadratic for conjunctive formulas.
%	For $\AZc$ the complexity is singly exponential for conjunctive formulas,
%while for $\AIc$ is quadratic.
\end{lemma}

\noindent
An analogous result was obtained in \cite[Theorem 5.5]{Koubarakis:PKRR:1994} for \emph{rational difference logic},
i.e., where the domain is $\Q$ and the atomic formulas are of the form $x - y \bowtie k$,
with $\bowtie \in \set{<, \leq, >, \geq, =}$ and $k \in \Q$.
The statement of Lemma~\ref{lem:qe-rat} concerns a more restrictive setting
where the domain is the unit rational interval $\I = \Q \cap [0, 1)$
and the constraints are only of the form $x = 0$ and $x \leq y$.
Since a first-order formula $\varphi$ of $(\I, \leq, 0)$ is also a formula of rational difference logic,
by \cite[Theorem 5.5]{Koubarakis:PKRR:1994} there exists an equivalent quantifier-free formula $\psi$ of difference logic.
However, we prove that $\psi$ is even a formula of the more restrictive structure $(\I, \leq, 0)$.

\begin{proof}[Proof (of \Cref{lem:qe-rat})]
	It suffices to consider a conjunctive formula of the form $\varphi \equiv \exists y \cdot \bigwedge_k \varphi_k$
	where $\varphi_k$ are atomic formulas.
	If any $\varphi_k$ is the constraint $y = 0$,
	then we obtain $\widetilde\varphi$ by replacing $y$ with $0$ everywhere.
	Otherwise, $\varphi$ is of the form
	\begin{gather*}
		%\label{eq:integer:qe}
		\exists y \cdot \bigwedge_{i \in I} {x_i} \leq y \wedge \bigwedge_{j \in J} y \leq {x_j},
	\end{gather*}
	and we can eliminate $y$ by writing the equivalent constraint $\widetilde\varphi$
	\begin{gather*}
		\bigwedge_{i \in I} \bigwedge_{j \in J} {x_i} \leq {x_j}.
	\end{gather*}
	The size of $\widetilde\varphi$ is quadratic in the size of $\varphi$.
\end{proof}

\subsection{Clock difference relations}
\label{sec:prelim:CDR}

Let $\X$ be a set of clocks containing a special clock $\x_0 \in \X$ which is never reset,
and, for every clock $\x$, let $\x'$ denote a copy thereof.
A \emph{clock difference relation} (\CDR) $\varphi(\bar \x, \bar x')$ over $\X$
is a Boolean combination of formulas of the form
\begin{align}
  \label{eq:CDR}
  u \leq t,
\end{align}
where $u, t$ are terms of one of the forms $\fract {\x_0 - \x}$, $\fract{\x_0}$, $\fract{\x_0' - \x'}$, $\fract{\x_0'}$.
As we will see in Sec.~\ref{sec:fractional:TPDA},
clock difference relations can express the one-step transition relation of timed pushdown automata,
restricted to fractional values.
As a basic building block,
the identity relation is expressible as the following \CDR
\begin{align}
	\label{eq:CDR:id}
	\varphi_{\id}(\bar x, \bar x') &\ \equiv\ 
	\fract{\x_0'} = \fract{\x_0} \wedge \bigwedge_{\x \in \X} \fract {\x_0' - \x'} = \fract {\x_0 - \x}.
\end{align}
Also fractional clock constraints are expressible as \CDR,
because of the following two equivalences:
\begin{align}
	\label{eq:CDR:zero}
	\fract \x = 0 \quad \textrm { iff } \quad
		&\fract{\x_0 - \x} = \fract {\x_0}, \textrm { and } \\
	\label{eq:CDR:leq}
	\fract \x \leq \fract \y \quad \textrm { iff } \quad
		&\fract{\x_0 - \x} \leq \fract{\x_0} \leq \fract{\x_0 - \y} \;\vee \\
		\nonumber
		&\fract{\x_0} \leq \fract{\x_0 - \y} \leq \fract{\x_0 - \x} \;\vee \\
		\nonumber
		&\fract{\x_0 - \y} \leq \fract{\x_0 - \x} \leq \fract{\x_0}.
\end{align}
%
%showing that the clock constraint on the right is equivalent to the \CDR on the left.
%
%More generally, we will see in Sec.~\ref{sec:fractional:TPDA}
%how \CDR can express the \emph{transition relation} of timed pushdown automata
%(more generally, we will prove that they can express the \emph{reachability relation}).

There are two additional facts that make \CDR particularly interesting.
First, for a fixed set of clocks $\X$, there are finitely many \CDR up to logical equivalence.
Second, \CDR are closed \wrt \emph{relational composition}:
Given two \CDR $\varphi(\bar \x, \bar \x'), \psi(\bar \x', \bar \x'')$
their composition is defined as
\begin{align}
	\label{eq:CDR:composition}
	(\varphi \circ \psi)(\bar \x, \bar \x'') \;\equiv\; \exists \bar \x' \cdot \varphi(\bar \x, \bar \x') \wedge \psi(\bar \x', \bar \x'').
\end{align}

\begin{lemma}
  \label{lem:CDR:qe}
  Clock difference relations admit effective elimination of quantifiers
  of the form $\exists y$ for $y$ different from the special variables $x_0, x_0'$.
\end{lemma}

\begin{proof}
	Let the set of clocks be $\X = \set{x_0, x_1, \dots, x_n}$
	and consider the quantified \CDR $\exists x_i \cdot \varphi(\bar x, \bar x')$ with $i \neq 0$,
	where $\varphi(\bar x, \bar x')$ is a \CDR (i.e., quantifier free).
	(The case $\exists x_i' \cdot \varphi(\bar x, \bar x')$ is analogous.)
	Consider the linear transformation $f : \R^{n+1} \to \R^{n+1}$ defined as:
	\begin{align*}
		f(x_0, x_1, \dots, x_n) = (x_0, {x_0 - x_1}, \dots, {x_0 - x_n}).
	\end{align*}
	Clearly, $f$ is a bijection,
	and $\varphi(f(\bar x), f(\bar x'))$ is a (fractional) clock constraint.
	Consider the quantified clock constraint
	\begin{align*}
		\xi \ \equiv\ \exists x_i \cdot \varphi(f(\bar x), f(\bar x')).
	\end{align*}
	By Lemma~\ref{lem:qe-rat}, $\xi$ is equivalent to a fractional clock constraint $\hat\xi(\bar x, \bar x')$ not containing $x_i$.
	The formula $\hat\xi'(\bar x, \bar x') \equiv \hat\xi(f^{-1}(\bar x), f^{-1}(\bar x'))$
	obtained by applying the inverse function
	$f^{-1}(\delta_0, \delta_1, \dots, \delta_n) = (\delta_0, {\delta_0 - \delta_1}, \dots, {\delta_0 - \delta_n})$
	is a \CDR logically equivalent to $\exists x_i \cdot \varphi(\bar x, \bar x')$.
	Since $x_i$ is not present in $\hat\xi$ and it is different from the special variable $x_0$,
	$x_i$ is also not present in $\hat\xi'$
	by the definition of $f^{-1}$.
\end{proof}

% we do not need this anymore
%For a clock difference relation $\varphi(\bar \x, \bar \x')$,
%let $\varphi^0 \equiv \varphi_\id$, $\varphi^{i+1} \equiv \varphi^i \circ \varphi$,
%and consider the infinite formula
%
%\begin{align*}
%  \varphi^* \equiv \bigvee_{i \geq 0} \varphi^i.
%\end{align*}
%
%Since, for a fixed set of clocks $\X$,
%there are only finitely many clock difference relations up to logical equivalence,
%$\varphi^*$ is in fact a clock difference relation equivalent to a finite union
%$\bigvee_{i = 0}^j \varphi^i$, and moreover $j$ can be taken to be the least index
%s.t.~$\varphi^j$ and $\varphi^{j+1}$ are logically equivalent.

%\begin{lemma}
%  Clock difference relations are closed under Kleene's star.
%\end{lemma}

\begin{corollary}
	\label{cor:CDR:composition}
	Clock difference relations are closed under composition.
\end{corollary}

\begin{remark}
	In the original definition of \cite{KrcalPelanek:TM:FSTTCS:2005},
	in a basic \CDR $u \leq t$ as above
	terms $u, v$ are of one of the forms
	$\fract {x} - \fract{y}$, $1 - (\fract {x} - \fract{y})$,
	$\fract {x'} - \fract{y'}$.
	Our presentation differs in two respects:
	1) We compare fractional parts of differences of clocks rather than differences of fractional parts;
	this has the advantage of being invariant under time elapse
	and thus we do not need expressions of the form $1 - (\fract {x} - \fract{y})$.
	2) Differences are taken only \wrt $x_0, x_0'$, instead of arbitrary clocks $x, x'$.
	  %
	%The two presentations are equi-expressive,
	%but we find it more convenient to have only one kind of constraints.%
\end{remark}
% !TEX root = main.tex

\section{Timed pushdown automata}
\label{sec:TPDA}

%i.e.,
%if $[\floor {x_1} \mapsto \floor {\mu(x_1)}, \cdots, \floor {x_k} \mapsto \floor {\mu(x_k)}, \fract {x_1} \mapsto \fract {\mu(x_1)}, \cdots, \fract {x_k} \mapsto \fract {\mu(x_k)} ] \models \varphi$,
%
%Two clock constraints $\varphi$ and $\psi$ are \emph{equivalent} if $\sem{\varphi} = \sem{\psi}$.
%

%\paragraph{The model.}

%\cite{AbdullaAtigStenman:DensePDA:12}
A \emph{timed pushdown automaton} (\TPDA) is a tuple
$\P = \tuple {\Sigma, \Gamma, L, \X, \ZZ, \Delta}$
where $\Sigma$ is a finite \emph{input alphabet},
$\Gamma$ is a finite \emph{stack alphabet},
$L$ is a finite set of \emph{control locations},
$\X$ is a finite set of \emph{control clocks}, 
$\ZZ$ is a finite set of \emph{stack clocks} disjoint from $\X$.
The last item $\Delta$ is a set of transition rules
of the form $\trule p \op q$ with $p, q \in L$ control locations,
where $\op$ determines the type of transition:
\begin{itemize}
	\item \emph{time elapse}%
	\footnote{Explicit time elapse transitions are non-standard in the literature on \TA.
	The standard semantics of timed automata where time can elapse freely in every control location
	is simulated by adding explicit time elapse transitions
	$\trule{p}{\elapse}{p}$ in every location $p$.
	Our explicit, more fine grained modelling of the elapse of time will simplify the constructions of the paper.}
	$\op = \elapse$;
	\item \emph{input} $\op = \readop a$ with $a \in \Sigma_\varepsilon := \Sigma \cup \{\varepsilon\}$ an input letter;
	\item \emph{test} $\op = \testop \varphi$ where $\varphi$ is a clock constraint over clocks in $\X$,
	called the \emph{transition constraint};
	\item \emph{reset} $\op = \resetop \Y$ with $\Y \subseteq \X$ a set of clocks to be reset
	(when $\Y$ is the singleton $\set \x$, sometimes we just write $\resetop \x$);
	\item \emph{push} $\op = \pushop{\alpha}{\psi}$
	with $\alpha \in \Gamma$ a stack symbol to be pushed on the stack
	under the clock constraint $\psi$ over clocks $\X \cup \ZZ$,
	called the \emph{stack constraint};
	\item \emph{pop} $\op = \popop{\alpha}{\psi}$ similarly as push.
\end{itemize}
We also allow transitions $\trule p {\op_1; \cdots; \op_n} q$
to carry a sequence of operations, to be executed in the given order,
as to avoid introducing intermediate control locations.
We assume that every atomic stack constraint contains some stack variable from $\ZZ$.
%
%otherwise they can be checked as a transition constraint.
%
A \TPDA has \emph{untimed stack} if the only stack constraint is $\true$.
%in which case we omit the constraint in stack operations and write simply $\pushop \alpha$ and $\popop \alpha$.
%
Without push/pop operations, we obtain nondeterministic timed automata (\TA).

\begin{example}
	\label{ex:TPDA}
	For illustration, consider a \TPDA with one control clock $\x$ and one stack clock $\z$ 
	that recognises the language of even-length palindromes over $\Sigma = \{a, b\}$ under certain timing constraints
	to be unravelled later.
	The stack alphabet of the \TPDA is $\Gamma = \Sigma$.
	In the initial control location $q_0$ the \TPDA just resets the control clock 
	and moves to the control location $q_1$:
	\begin{align*}
	\trule {q_0} {\resetop \x} {q_1}.
	\end{align*}
	In location $q_1$ the \TPDA reads $a$'s and $b$'s and keeps track of them on the stack, but the
	initial value of the stack clock depends on the input letter:
	%($c \in \{a, b\}$): 
	\begin{align*}
	&\trule {q_1} {\readop a; \pushop{a}{\z = 0}} {q_1}, \\
	&\trule {q_1} {\readop b; \pushop{b}{\z = 1}} {q_1}, \\
	&\trule {q_1} {\readop{\varepsilon}} {q_2}.
	\end{align*}
	The last transition silently moves to location $q_2$, in which
	the \TPDA pops the stack while checking a
	constraint on the time elapse since the corresponding push:
	\begin{align} \label{eq:popconstr}
	&\trule {q_2} {\readop c; \popop{c}{\floor \z \eqv 2 0 \land \fract{\z} \leq \fract{\x}}} {q_2} 
	\qquad (c \in \{a, b\}).
	\end{align}
	Finally, in locations $q_1$ and $q_2$ (but not in $q_0$) unrestricted time elapse is enabled:
	\begin{align*}
	\trule {q_1} {\elapse} {q_1}, \qquad \trule {q_2} {\elapse} {q_2}.
	\end{align*} 
	Once acceptance by empty stack is imposed, the language recognised by the \TPDA,
	when starting in control location $q_0$ with empty stack, %and ending in $q_2$
	are those even-length palindromes which
	satisfy the following timing constraints:
	\begin{enumerate}
		\item the integer part of time elapsed between every matching pair of $a$'s is even;
		\item the integer part of time elapsed between every marching pair of $b$'s is odd;
		\item for every matching pair of letters, the fractional part of time-stamp (time-stamp = time elapsed since the reset of $\x$) of the first letter is smaller or equal
		to the fractional part of time-stamp of the second letter.
		%\lorenzo{I don’t see how control clock $x$ is used to ensure this; maybe the fractional value of $\z$ pushed should be equal to that of $x$?}
	\end{enumerate}
	The condition 3.~is imposed by the pop constraint $\fract{\z} \leq \fract{\x}$.
	At the moment of pop, the fractional part of time-stamp is $\fract{\x}$.
	In terms of the values of clocks $\x$ and $\z$ at the moment of pop, 
	the fractional part of time-stamp at the moment of the corresponding push is $\fract{\x - \z}$, since 
	$\z$ (resp.~$\z - 1$) represents the amount of time elapsed between push and pop of $a$ (resp.~$b$).
	The condition 3.~follows, since $\fract{\z} \leq \fract{\x}$ is equivalent to $\fract{\x - \z} \leq \fract{\x}$.
\end{example}
Throughout the paper we use $\x, \x_i$ to denote control clocks, and $\z, \z_j$ (or $\y, \y_j$) to denote stack clocks;
we let $\x_0$ be a control clock that is never reset (and thus measures the total elapsed time),
$\x_1$ a control clock that is reset at every push (and thus is assumed to be 0 at the time of push),
and $\z_1$ (or $\y_1$) a stack clock that is $0$ when pushed.

For complexity estimations, we assume \wlg that all clock constraints of the \TPDA are presented as conjunctive clock constraints.
Non-conjunctive constraints can be converted to conjunctive ones by first transforming to disjunctive normal form $\bigvee_{i=1}^m \varphi_i$
(where each $\varphi_i$'s is conjunctive)
and distribute each disjunct $\varphi_i$ to a different transition using the automaton's nondeterminism.
We also assume that constants are encoded in binary,
that all modular constraints use the same modulus $M$ (also encoded in binary)
and that all other constants appearing in constraints are smaller than $M$.

\subsection{Semantics}
%
%The formal semantics of \TPDA follows~\cite{AbdullaAtigStenman:DensePDA:12}.
%
Every stack symbol is equipped with a fresh copy of clocks from $\ZZ$.
At the time of $\pushop{\alpha}{\psi}$,
the push constraint $\psi$ specifies possibly nondeterministically the initial value of all clocks in $\ZZ$ \wrt control clocks in $\X$.
Both global and stack clocks evolve at the same rate when a time elapse transition is executed.
At the time of $\popop{\alpha}{\psi}$,
the pop constraint $\psi$ specifies the final value of all clocks in $\ZZ$ \wrt control clocks in $\X$.
%
%Two \TPDA are \emph{equivalent} if they recognise the same timed language.
%
A \emph{timed stack} is a sequence $w \in (\Gamma \times \Qgeq^\ZZ)^*$ of pairs $(\gamma, \mu)$,
where $\gamma$ is a stack symbol and $\mu$ is a valuation for stack clocks in $\ZZ$.
%
%For a clock valuation $\mu$ and a set of clocks $Y$, let $\mu[Y \mapsto 0]$ be the same as $\mu$ except that clocks in $Y$ are mapped to $0$.
%
For $\delta \in \Qgeq$
and a timed stack $w = (\gamma_1, \mu_1) \cdots (\gamma_k, \mu_k)$,
let $w + \delta$ be $(\gamma_1, \mu_1 + \delta) \cdots (\gamma_k, \mu_k + \delta)$.
A (\TPDA) \emph{configuration} is a triple $(p, \mu, w) \in L \times \Qgeq^\X \times (\Gamma \times \Qgeq^\ZZ)^*$
where $p$ is a control location,
$\mu$ is a clock valuation over the control clocks $\X$,
and $w$ is a timed stack.
%
%Let $c = (p, \mu, u), d = (q, \nu, v)$ be two configurations.
%
For every rule $\delta = \trule{p}{\op}{q} \in \Delta$
%input symbol or time increment $a \in (\Sigma_\varepsilon \cup \Qgeq)$
we have a transition
%
%\begin{gather*}
$p, \mu, u \reach \delta {} q, \nu, v$
%\end{gather*}
%
whenever one of the following conditions holds:
%
%\begin{multicols}{2}
\begin{itemize}
	\item%[\bf (elapse)]
	$\op = \elapse$ and	there is some $t \in \Qgeq$ \st $\nu = \mu + t$ and $v = u + t$.
	\item%[\bf (input)]
	$\op = \readop a$, $\nu = \mu$, $u = v$.
	\item%[\bf (test)]
	$\op = \testop \varphi$, $a = \varepsilon$, $\mu \models \varphi$, $\nu = \mu$, $u = v$.
	\item%[\bf (reset)]
	$\op = \resetop Y$, $\nu = \mu[Y \mapsto 0]$, $v = u$.
	\item%[\bf (push)]
	$\op = \pushop{\alpha}{\psi}$, $\mu = \nu$,
	$v = u \cdot \tuple{\alpha, \mu_1}$
	if $\mu_1 \in \Qgeq^\ZZ$ satisfies $(\mu, \mu_1) \models \psi$,
	where $(\mu, \mu_1) \in \Qgeq^{\X \cup \ZZ}$ is the unique clock valuation
	that agrees with $\mu$ on $\X$ and with $\mu_1$ on $\ZZ$.
	\item%[\bf (pop)]
	$\op = \popop{\alpha}{\psi}$, $\mu = \nu$,
	$u = v \cdot \tuple{\alpha, \mu_1}$
	provided that $\mu_1 \in \Qgeq^\ZZ$ satisfies $(\mu, \mu_1) \models \psi$.
\end{itemize}
%\end{multicols}
%
%A \emph{timed word} is a sequence ${w = \t_1 \delta_1 \cdots \delta_n a_n \in (\Qgeq\Delta)^*}$
%of alternating time elapses and input symbols;
%
The one-step transition relation $c \reach \delta {} d$
is extended on sequences of transitions $w \in \Delta^*$ in the natural way.
\ignore{
The \emph{timed language} from location $\ell$ to $\arr$ is
$\lang {\ell, \arr} := \setof
	{\project {\varepsilon} w \in (\Qgeq\Sigma)^*}
	{\tuple{\ell, \mu_0, \varepsilon} \goesto w \tuple{\arr, \mu_0, \varepsilon}}$
where $\project {\varepsilon} w$ %$\cdot : (\Qgeq\Sigma_\varepsilon)^* \to (\Qgeq\Sigma)^*$
removes the $\varepsilon$'s from $w$ and
$\mu_0$ is the valuation that assigns $\mu_0(x) = 0$ to every clock $x$.
The corresponding \emph{untimed language} $\languntimed {\ell, \arr}$
is obtained by removing the time elapses from $\lang {\ell, \arr}$.
}%

\subsection{Reachability relation}

The \emph{reachability relation} $\reach {} {pq} \;\subseteq\; \Qgeq^\X \times \Delta^* \times \Qgeq^\X$ of the \TPDA $\P$
is obtained by requiring that the stack is empty at the beginning and at the end of the run.
Formally, we write
$$\mu \reach w {pq} \nu \qquad \textrm{(or also $p, \mu \reach w {} q, \nu$)}$$
if $p, \mu, \varepsilon \reach w {} q, \nu, \varepsilon$.
%
%We call the family of relations $\set{\reach {} {pq}}_{p, q \in \L}$, where
The reachability relation can be characterised in the following natural way.
%The following characterisation is used in the proof of Lemma~\ref{lemma:A}.

\begin{lemma}
	\label{lem:characterisation}
	%Let $\ell, \arr$ be control states of the \trPDA $\P$.
	%
	The reachability relation $\reach {} {pq}$ is the least relation satisfying the rules below,
	where $p, q, r, s \in Q$,
	$\mu, \nu \in \Q^\X$,
	and $u, v \in \Delta^*$:
	\begin{align}
		\label{eq:reachrel:A}
		&\textrm{\em (input)} &&
		  \frac{}{\mu \reach \delta {pq} \mu}
			 \qquad && \forall \delta = \trule p {\readop a} q \in \Delta \\[2ex]
		\label{eq:reachrel:B}
		&\textrm{\em (test)} &&
		\frac{}{\mu \reach \delta {pq} \mu}
		   \qquad && \forall  \delta = \trule p {\testop \varphi} q \in \Delta \st \mu \models \varphi \\[2ex]
		\label{eq:reachrel:C}
		&\textrm{\em (reset)} &&
		\frac{}{\mu \reach \delta {pq} \mu [\Y \mapsto 0]}
		   \qquad && \forall \delta = \trule p {\resetop \Y} q \in \Delta \\[2ex]
		\label{eq:reachrel:D}
		&\textrm{\em (elapse)} &&
		  \frac{}{\mu \reach \delta {pq} \mu + t}
			 \qquad && \forall \delta = \trule p \elapse q \in \Delta, t \geq 0 \\[2ex]
		\label{eq:reachrel:E}
	   	&\textrm{\em (transitivity)} &&
		  \frac{\mu \reach u {pr} \rho \quad \rho \reach v {rq} \nu}{\mu \reach {uv} {pq} \nu} \\[2ex]
		\label{eq:reachrel:F}
		&\textrm{\em (push-pop)} &&
			\frac{\mu \reach u {rs} \nu}{\mu \reach v {pq} \nu}
			&& 
			%\left\{
				\begin{array}{l}
					\forall \delta_\push = \trule p {\pushop \alpha {\psi_\push}} r, \\
					\delta_\pop = \trule s {\popop \alpha {\psi_\pop}} q \in \Delta, \\
					v = \delta_\push \cdot u \cdot \delta_\pop,
				\end{array}
			%\right.
	\end{align}
	whenever the following condition is satisfied:
	\begin{align}
		\label{eq:characterisation}
			\exists \mu_\ZZ \in \Qgeq^\ZZ
			\st
			%\left\{
			\begin{array}{l}
				(\mu, \mu_\ZZ) \models \psi_\push \textrm{ and } \\
				(\nu, \mu_\ZZ + \delta_{\mu\nu}) \models \psi_\pop,
			\end{array}
			%\right.
	\end{align}
	where $\delta_{\mu\nu} := \nu(\x_0) - \mu(\x_0)$ measures the total amount of time elapsed
	between a push and its corresponding pop
	(recall that $\x_0$ is never reset).
\end{lemma}
A special case of the reachability relation is the \emph{reachability set}, % $\ReachSet {pq} \subseteq \Delta^* \times \Rgeq^\X$,
where the initial clock values are not specified:
\begin{align*}
	\ReachSet {pq} = \setof{(w, \nu) \in \Delta^* \times \Rgeq^\X}{\exists \mu \in \Rgeq^\X \st \mu \reach w {pq} \nu}.
\end{align*}

We are mainly interested in the computational problem of building a finite representation for the reachability relation.
A related decision problem is the \emph{nonemptiness problem},
which amounts to deciding, given two control locations $p, q \in \L$, whether $\ReachSet {pq} \neq \emptyset$.

\subsection{Fractional reachability relation}

Our strategy to compute the reachability relation involves a series of transformations in \Cref{sec:simplify}
reducing to computing the \emph{fractional reachability relation}
$\freach {} {cd} \;\subseteq\; (\R \cap [0, 1))^\X \times \Delta^* \times (\R \cap [0, 1))^\X$, which is defined as follows.
For fractional valuations $\mu, \nu \in (\R \cap [0, 1))^\X$ and sequence of transitions $w \in \Delta^*$,
\begin{align}
	\label{eq:fractional:reachability}
	\mu \freach w {cd} \nu 
		\quad \textrm{ if } \quad
			\exists \tilde\mu, \tilde\nu \in \Rgeq^\X \st \fract {\tilde\mu} = \mu, \fract{\tilde\nu} = \nu, \textrm{ and } \tilde\mu \reach w {cd} \tilde\nu.
\end{align}

We say that a \TPDA is \emph{fractional} if the only clock constraints are the fractional ones.
We observe that fractional reachability is transitive for fractional \TPDA, which will be useful in \Cref{sec:fractional:TPDA}.
\begin{fact}
  \label{fact:freach:transitive}
  The fractional reachability relation is \emph{transitive} for fractional \TPDA, in the sense that
  \begin{align*}
    \mu \freach u {pr} \rho \freach v {rq} \nu
      \quad \textrm{ implies } \quad
        \mu \freach {uv} {pq} \nu.
  \end{align*}
\end{fact}

\begin{proof}
  Transitivity is not immediately clear from the definition of $\mu \freach w {pq} \nu$
  due to the existential quantification on $\mu'$.
  For $\mu, \nu \in (\R \cap [0, 1))^\X$, consider the following stronger notion: % $\mu \freachAE w {pq} \nu$, which holds if
  \begin{align*}
    \mu \freachAE w {pq} \nu \textrm{ if } \underbrace{\forall \mu' \in \Rgeq^\X}_{\textrm{vs.~} \exists \mu' \in \Rgeq^\X} \st \exists \nu' \in \Rgeq^\X \st \fract{\mu'} = \mu \textrm{ implies: } \fract{\nu'} = \nu \textrm{ and } \mu' \reach w {pq} \nu'.
  \end{align*}
  The relation $\freachAE {} {pq}$ is easily shown to be transitive. %, based on the transitivity of $\reach {} {pq}$.
  In fact, we show that $\freachAE {} {pq}$ and $\freach {} {pq}$ coincide, yielding the sought transitivity of $\freach {} {pq}$.
  One direction is immediate. For the other direction, assume $\mu \freach w {pq} \nu$ for some $\mu, \nu \in (\R \cap [0, 1))^\X$.
  By definition, there are $\tilde \mu, \tilde \nu \in \Rgeq^\X$
  \st $\tilde \mu \reach w {pq} \tilde \nu$, $\fract{\tilde\mu} = \mu$, and $\fract{\tilde\nu} = \nu$.
  In order to show $\mu \freachAE w {pq} \nu$,
  let $\mu' \in \Rgeq^\X$ be \st $\fract{\mu'} = \mu = \fract{\tilde\mu}$.
  The task is to find $\nu' \in \Rgeq^\X$ \st $\fract{\nu'} = \nu = \fract{\tilde \nu}$ and $\mu' \reach w {pq} \nu'$.
  If clock $\x_i$ is not reset in $w$, then let $\floor{\nu'(\x_i)} = \floor{\tilde\nu(\x_i)} + \floor{\mu'(\x_i)} - \floor{\tilde\mu(\x_i)}$.
  If clock $\x_i$ is reset in $w$, then let $\floor{\nu'(\x_i)} = \floor{\tilde\nu(\x_i)}$.
  This uniquely defines $\nu'$.
  Since the \TPDA is fractional, the very same run showing that $\tilde \mu \reach w {pq} \tilde \nu$ holds (including precise time elapses)
  also shows that $\mu' \reach w {pq} \nu'$ holds, as required.
\end{proof}

\section{Overview of the reductions}
\label{sec:overview}

\begin{figure}
	\begin{center}
		\begin{tikzpicture}[auto,>=latex', line/.style={draw, thick, -latex',shorten >=2pt},block/.style={rectangle, text width=7em, draw=blue, thick, fill=blue!20, align=center, text centered, rounded corners, minimum height=3.2em, drop shadow}]
			\matrix [column sep=5em, row sep=4em] {
				\node[block] (TPDA) {\TPDA}; &
					\node[block] (pushcopy) {push-copy \TPDA}; &
						\node[block] (popintegerfree) {no integral pop constraints}; \\
						&\node[block] (CFG) {context-free grammar}; &
							\node[block] (fractional) {fractional \TPDA}; \\
			};
			\begin{scope}[every path/.style=line]
				\path (TPDA) -- node [midway] {\Cref{sec:simplify:push-copy}} node [below] {} (pushcopy);
				\path (pushcopy) -- node [midway] {\Cref{sec:simplify:pop-integer-free}} node [below] {} (popintegerfree);
				\path (popintegerfree) -- node [midway] {\Cref{sec:simplify:fractional}} node [left] {} (fractional);
				\path (fractional) -- node [midway] {\Cref{sec:fractional:TPDA}} node [above] {} (CFG);
			\end{scope}
		\end{tikzpicture}
	\end{center}
	\caption{The main reductions to compute the reachability relation.}
	\label{fig:synopsis}
\end{figure}
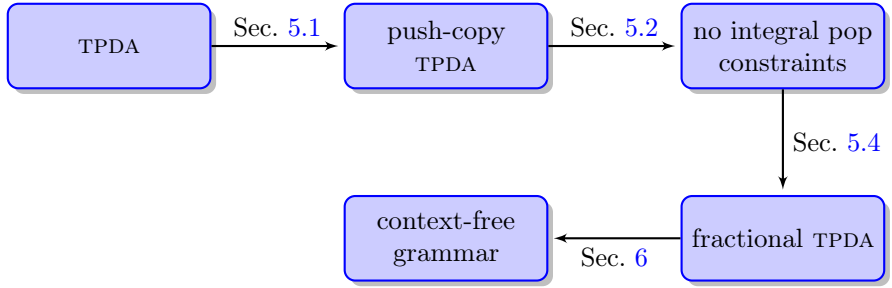

We prove our main result \Cref{thm:TPDA},
concerning the expressibility of the \TPDA reachability relation in linear arithmetic,
via a series of reductions, as outlined in \Cref{fig:synopsis}.
In the first part (\Cref{sec:simplify}) we reduce computing the reachability relation of a \TPDA
to the fractional reachability relation of a fractional \TPDA,
and in the second part (\Cref{sec:fractional:TPDA}) we compute the fractional reachability of a fractional \TPDA.
We briefly comment below on the most interesting steps.

\begin{enumerate}

	\item[\Cref{sec:simplify:push-copy}:]
		We transform a given \TPDA into a \TPDA where stack clocks are copies of control clocks.
		This step uses quantifier elimination for clock constraints (cf.~\Cref{sec:qe}),
		which introduces pop constraints of exponential size in disjunctive normal form, 
		which are in turn distributed to the transitions.
		The complexity of this step is an exponential blow-up in the number of pop transitions.
		This step essentially preserves the reachability relation.

	\item[\Cref{sec:simplify:pop-integer-free}:]
		We remove integral pop constraints, i.e., constraints of the form ${y_i - x_j \sim k}$.
		We simulate integral pop constraints by transition constraints on additional control clocks,
		generalising a previous result on \emph{untiming the stack} of a \TPDA (with a restricted syntax of constraints)
		{\cite[Theorem II.1]{ClementeLasota:LICS:2015}}.
		The complexity of this step is an exponential blow-up of the number of control locations and transitions
		in terms of the number of pop transitions.
		With the previous step, it combines to a double exponential blow-up.
		Also this step essentially preserves the reachability relation.

  \item[\Cref{sec:simplify:reset}:]
		We present a construction from \cite{FranzleQuaasShirmohammadiWorrell:2019}
		implying that all clocks can be assumed to be reset during the run, which simplifies our further development.
		This step only doubles the number of clocks.
		
		The same construction is used in \cite{FranzleQuaasShirmohammadiWorrell:2019} to show the much more interesting fact that,
		on timed automata (i.e., no stack),
		computing the reachability relation reduces to computing the \emph{reachability set},
		which is a much simpler object in principle.
		While the same also holds for \TPDA (as we show), in the presence of a stack
		the reachability set lacks the nice characterisation of \Cref{lem:characterisation},
		and thus it is not easier to compute than the reachability relation itself, as we detail in \Cref{sec:simplify:reset}.
		Therefore, we do not actually use in full generality the powerful reduction to the reachability set of \cite{FranzleQuaasShirmohammadiWorrell:2019}.

  \item[\Cref{sec:simplify:fractional}:]
		We remove all remaining non-fractional clock constraints,
		i.e., stack modulo constraints, and transition integral and modulo constraints,
		thus producing a fractional \TPDA.
		This step increases the number of control locations and transitions
		by an exponential multiplicative factor,
		thus leaving the combined complexity with the previous steps unchanged.
		This reduction preserves only the fractional reachability relation.
		The integral part of the reachability relation is reconstructed by encoding it in the (untimed) language,
		by using a technique inspired from \cite{QuaasShirmohammadiWorrell:LICS:2017}.

  \item[\Cref{sec:fractional:TPDA}:]
		In the last step, we compute the fractional reachability relation of a fractional \TPDA by encoding it into a context-free grammar
		and using Parikh's theorem to compute its commutative image.
		This is achieved by representing fractional reachability relations by clock difference relations (\CDR),
		which enjoy nice properties such as quantifier elimination (cf.~\Cref{lem:CDR:qe}),
		and thus closure under relational composition (cf.~\Cref{cor:CDR:composition}).
		The number of nonterminal symbols and productions of the context-free grammar
		is polynomial in the number of control locations and \CDR's (of which there are exponentially many in the number of clocks).
		Combined with the previous steps, this yields a grammar of doubly exponential size.
\end{enumerate}
% !TEX root = main.tex

\section{Reduction to fractional \TPDA}
\label{sec:simplify}

We show that computing the reachability relation
reduces to the same problem for fractional \TPDA.
%
%The reachability relation of a fractional \TPDA itself will be computed in Sec.~\ref{sec:cfg}
%by constructing a context-free grammar.
Our transformation is done in three major steps:
In the first step we restrict the form of push operations (Sec.~\ref{sec:simplify:push-copy}),
in the second step we restrict pop operations (Sec.~\ref{sec:simplify:pop-integer-free}),
and finally in the last step we eliminate all integral and modular constraints,
thus obtaining a fractional \TPDA (Sec.~\ref{sec:simplify:fractional}).
We summarise the combined complexity of this sequence of reductions.

\begin{lemma}
	%\label{thm:fractional:TPDA}
	%The reachability relation of \TPDA \slawek{shouldn't we explain what it means?}
	%effectively reduces to the reachability relation of fractional \TPDA
	%with a double exponential complexity blow-up.
	%
%	The number of control locations has an exponential blow-up,
%	and the number of clocks has linear blow-up.
	%
	A \TPDA $\P$ can be effectively transformed into a fractional \TPDA $\QQ$
	\st	
%	a semiatom describing the reachability relation of $\P$
%	can effectively be computed from a semiatom for the language reachability of $\QQ$.
	a linear-arithmetic description $\set{\varphi_{pq}}$ of the reachability relation of $\P$
	can effectively be computed from a linear arithmetic description $\set{\varphi_{p'q'}'}$
	of the reachability relation of $\QQ$. 
	%\slawek{the complexity is missing of transforming $\varphi'$ into $\varphi$}
%	family of \logic formulas $\set{\varphi_{pq}}$ expressing the reachability relation of $\P$
%	can effectively be computed from a family of \logic formulas $\set{\varphi_{p'q'}'}$
%	expressing the reachability relation of $\QQ$. %\slawek{the complexity is missing of transforming 
%	$\varphi'$ into $\varphi$}
	%
	The number of control locations and the size of the stack alphabet in $\QQ$ have a double exponential blow-up,
	and the number of clocks has an exponential blow-up.
\end{lemma}

\ignore{
%
%\noindent
%If there is no stack, then we do not need the first two steps, and we can do directly {\bf C}.
%
\begin{corollary}
	\label{cor:fractional:NTA}
	The reachability relation of push-copy \TPDA
	effectively reduces to the reachability relation of fractional \TPDA
	with an exponential blow-up in control locations.
%	The number of control locations has an exponential blow-up,
%	and the number of clocks has linear blow-up.
%\slawek{the statement is unclear}
%\slawek{in case of push-copy \TPDA, there is also exponential blow-up of the size of stack alphabet}
\end{corollary}
}

\subsection{Simplifying push operations: Push-copy}
\label{sec:simplify:push-copy}

A \TPDA is \emph{push-copy} if push operations can only copy control clocks into stack clocks.
More precisely, we have one stack clock $\y_i$ for each control clock $\x_i$, with $i \geq 0$;
by convention, $\x_0$ is never reset
and $\x_1$ is assumed to be $0$ at the time of push (as thus $y_1$).
The only push constraint is thus%
%\footnote{
%The constraint $\y_1 = 0$ can be put in the push-copy format by replacing it with $\y_1 = \x$,
%for a new control clock $\x$ which is ensured to be $0$ at the time of push.
%For simplicity, we allow to write $\y_1 = 0$ directly.
%}
%
\begin{gather}
	\label{eq:copy:push:constraint}
		\psicopy(\x_0, \dots, \x_n, \y_0, \dots, \y_n) \ \equiv\ \bigwedge_{i = 0}^n \y_i = \x_i.
		%\floor {\z_\x} = \floor \x \wedge \fract {\z_\x} = \fract \x.
\end{gather}
Thus the number of control clocks $\x_0, \x_1, \ldots, \x_n$ is the same as the number of stack clocks
$\y_0, \y_1, \ldots, \y_n$ in a push-copy \TPDA.

We transform a \TPDA into a push-copy \TPDA with essentially the same reachability relation.
By pushing copies of control clocks into the stack,
we postpone checking all non-trivial push stack constraints to the time of pop,
thus substantially simplifying the form of push constraints.
This step uses quantifier elimination to construct suitable pop constraints.
The blow-up in the size of pop constraints is exponential (due to quantifier elimination).

%Let $K_\leq$ be the non-strict variant of the ternary cyclic order $K$ from \eqref{eq:cyc},
%defined as $K_\leq(a,b,c) \equiv K(a,b,c) \vee a=b \vee b=c$
%for $a, b, c \in \unitint$.
%let the (strict) \emph{cyclic order} $K(a, b, c)$,
%and its non-strict variant $K_\leq(a,b,c)$ be:
%
%\begin{align}
%	\label{eq:cyc}
%	\begin{aligned}
%	K(a, b, c) \ & \equiv\ a < b < c \, \vee \, b < c < a \, \vee \, c < a < b
%	\\
%	K_\leq(a,b,c) \ & \equiv \ K(a,b,c) \, \vee \, a=b \, \vee \, b=c.
%	\end{aligned}
%\end{align}
%
%These relations will be useful in this section as well as in Sec.~\ref{sec:frac2reg}
%to describe the fractional parts of clocks.
%
Let $\psi_\push(\vec \x, \vec \z)$ be a push constraint,
and let $\psi_\pop(\vec \x', \vec \z')$ be the corresponding pop constraint.
Let $y_i$ be a new stack clock which is a copy of control clock $x_i$ at the time of push.
%
%for simplicity, we sometimes write $\z_0'$ for the vector of variables $(\z_0', \dots, \z_0')$.
%
%\vec z_{\vec x},
%
Its value $y_i'$ at the time of pop is $x_i$ plus the total time $y_1'$ that elapsed between push and pop.
% (SL camera ready)   (thus $z'_{x_0} = z_0'$).
%
%Let $\delta \geq 0$ be the time elapsed between push and pop,
%and let $\vec \delta = (\delta, \dots, \delta)$ (the length of which depends on the context).
%
Since all clocks evolve at the same rate,
for every control clock $\x_i$ and stack clock $\z_j$, we have
\begin{align}
	\label{equivalence}
	\x_i = \y'_i - \y_1' \qquad \textrm{ and } \qquad \z_j = \z_j' - \y_1'.
\end{align}
By applying the equations above,
we obtain the following new pop formula talking about the pop value of new stack clocks $\vec \y'$:
\begin{align}
	\label{eq:psi:pop'}
	\psi_\pop'(\vec \x', \vec \y') \ \equiv\ \exists \vec \z' \geq \vec 0 \st
		\psi_\push(\vec y' - \y_1', \vec \z' - \y_1') \wedge \psi_\pop(\vec \x', \vec \z').
\end{align}
Intuitively, $\psi_\pop'$ guesses the final value $\vec \z'$ of stack clocks 
as to satisfy push and pop constraints.
While ${\psi_\push(\vec y' - \y_1', \vec \z' - \y_1')}$ is not itself a clock constraint
(since variables are replaced by differences of variables),
by Lemma~\ref{lem:shift}
it is equivalent to some clock constraint $\psi_\push'(\vec y', \y_1', \vec \z')$
with a linear size blow-up.
By replacing $\psi_\push$ with $\psi_\push'$ in \eqref{eq:psi:pop'},
we can rewrite $\psi_\pop'$ as
\begin{align}
	\label{eq:psi:pop''}
	\psi_\pop'(\vec \x', \vec \y') \;\equiv\; \exists \vec \z' \geq \vec 0 \st
	\psi_\push'(\vec y', \y_1', \vec \z') \wedge
	\psi_\pop(\vec \x', \vec \z'),
\end{align}
By Lemma~\ref{lem:qe-clocks}, we can perform quantifier elimination,
obtaining a logically equivalent (quantifier free) clock constraint
$\xi_{\psi_\push, \psi_\pop}(\vec \x', \vec \y')$ of exponential size.

\paragraph{The construction}

We now present the formal construction.
Given a \TPDA $\P = (\Sigma, \Gamma, \L, \X, \ZZ, \Delta)$
we construct a push-copy \TPDA $\QQ = (\Sigma, \Gamma', \L, \X, \ZZ', \Delta')$,
%with essentially the same reachability relation,
where $\ZZ' = \setof {\y_i} {\x_i \in \X} \cup \set{\y_1}$
and the new stack alphabet $\Gamma'$ contains symbols of the form $\tuplesmall{\tilde\delta_\push, \tilde\delta_\pop} \in \Delta \times \Delta$, where
\begin{align*}
	\tilde\delta_\push = \trule p {\pushop {\gamma} {\psi_\push}} r
		\quad \textrm{ and } \quad
			\tilde\delta_\pop = \trule s {\popop{\gamma}{\psi_\pop}} q.
\end{align*}
The construction of $\QQ$ consists in checking $\xi_{\psi_\push, \psi_\pop}$
in place of $\psi_\pop$, assuming that the corresponding push transition was $\delta_\push$.
The latter is replaced by $\psicopy$.
Transitions in $\Delta'$ are determined as follows.
Input, test, time elapse, and clock reset transitions in $\P$
generate identical transitions in $\QQ$.
For every pair of push $\tilde \delta_\push$ and pop $\tilde \delta_\pop$ transitions in $\P$ as above,
we have a push $\delta_\push$ and pop $\delta_\pop$ transitions in $\QQ$ of the form
\begin{align}
	\label{eq:push:pop:new}
	\underbrace{\trule p {\pushop{\tuplesmall{\tilde\delta_\push, \tilde\delta_\pop}}{\psicopy}} r}_{\delta_\push}
		\ \ \textrm{and} \ \ 
			\underbrace{\trule s {\popop{\tuplesmall{\tilde\delta_\push, \tilde\delta_\pop}}{\xi_{\psi_\push, \psi_\pop}}} q}_{\delta_\pop},
\end{align}
where the constraint $\psicopy$ is defined in \eqref{eq:copy:push:constraint}.
This concludes the description of $\QQ$.
For a word $w \in (\Delta')^*$, let $\tilde w \in \Delta^*$ be obtained by replacing each
push $\delta_\push$ and pop $\delta_\pop$ transitions \eqref{eq:push:pop:new} by $\tilde \delta_\push$, resp., $\tilde \delta_\pop$.
The following lemma shows that $\P$ and $\QQ$ have the same reachability relation up to ``\,$\tilde {\ }$\,''.

\begin{example} 
\label{ex:pushcopy}
Consider a pair of corresponding push and pop transitions from Example~\ref{ex:TPDA}
($q'_1, q'_2$ are auxiliary control locations):
\begin{align*}
\tilde\delta_\push = \trule {q'_1} {\pushop{b}{\z = 1}} {q_1}, \qquad
\tilde\delta_\pop = \trule {q'_2} {\popop{b}{\floor \z \eqv 2 0 \land \fract{\z} \leq \fract{\x}}} {q_2}.
\end{align*}
%(We rename $x$ to $x_0$.)
According to our transformation, the new push transition is
\begin{align*}
&\delta_\push = \trule {q'_1} {\pushop{\tuplesmall{\tilde\delta_\push, \tilde\delta_\pop}}{y_1 = 0 \land y = x}} {q_1};
\end{align*}
if $x$ is identified with the special control clock $x_0$ (which needs not be copied to the stack in push transitions) 
we don't need the stack clock $y$ and the push-copy constraint simplifies to:
\begin{align*}
&\delta_\push = \trule {q'_1} {\pushop{\tuplesmall{\tilde\delta_\push, \tilde\delta_\pop}}{\y_1 = 0}} {q_1}.
\end{align*}
(Recall that $\y_1$ is a special stack clock set to $0$ at every push.)
The new pop transition
\begin{align*}
&\delta_\pop = \trule {q'_2} {\popop{\tuplesmall{\tilde\delta_\push, \tilde\delta_\pop}}{\xi}} {q_2}
\end{align*}
is derived by first instantiating the formula~\eqref{eq:psi:pop'}
which, in this case, is already in the form required by~\eqref{eq:psi:pop''}
(we use primed variables to indicate that they refer to the time of pop):
\begin{align*}
\exists \z' \geq 0 \st
		\z' - \y_1' = 1 \wedge \floor {\z'} \eqv 2 0 \land \fract{\z'} \leq \fract{\x'_0}
\end{align*}
and then by applying the quantifier elimination procedure of Lemma~\ref{lem:qe-clocks} to obtain an equivalent clock constraint:
\begin{align*}
		\xi(\x'_0, \y_1') \ \equiv \ \floor {\y_1'} \eqv 2 1 \land \fract{\y_1'} \leq \fract{\x'_0}.
\end{align*}
%Note that the special stack clock $\y_1$, even if it does not appear in constraints of the original \TPDA,  
%can appear in the pop constraint after the transformation step.
\end{example}

\begin{restatable}[Correctness]{lemma}{lemmaA}
	\label{lemma:A}
	%Let $\reach {} {pq}$ be the reachability relation of $\P$
	%and let $\reach {} {pq}'$ be that of $\QQ$.
	%
	For every $p, q \in \L$, $w \in (\Delta')^*$, and $\mu, \nu \in \Qgeq^\X$,
	%we have
	\begin{align*}
		%\label{eq:lemma:A}
		\mu \reach {\tilde w} {pq} \nu
			\textrm{ in $\P$} \quad \textrm {if, and only if,} \quad 
				\mu \reach w {pq} \nu  \textrm{ in $\QQ$}.
	\end{align*}
%	and $\QQ$ has stack alphabet exponential in the size of $\P$.
\end{restatable}

\begin{proof}
	We prove the ``only if'' direction by induction on the length of derivations,
	following the characterisation of Lemma~\ref{lem:characterisation}.
	(The other direction is proved analogously.)
	Let $\mu \reach {\tilde w} {pq} \nu$ in $\P$.
	Since all transitions are the same except push and pop transitions,
	it suffices to prove it for matching pairs of push-pop transitions.
	By~\eqref{eq:characterisation}, there exist transitions	$\tilde\delta_\push, \tilde\delta_\pop$ in $\Delta$ as above,
	and a stack clock valuation $\mu_\ZZ \in \Qgeq^\ZZ$,
	\st	$(\mu, \mu_\ZZ) \models \psi_\push(\vec \x, \vec \z)$,
	${(\nu, \mu_\ZZ + \delta_{\mu \nu}) \models \psi_\pop(\vec \x', \vec \z')}$,
	$\tilde w= \tilde\delta_\push \cdot \tilde v \cdot \tilde\delta_\pop$,
	and $\mu \reach {\tilde v} {rs} \nu$ in $\P$.
	(Recall that $\delta_{\mu \nu} = \nu(\x_0) - \mu(\x_0)$ is the time elapsed between push and pop.)
	By the inductive hypothesis, $\mu \reach v {rs} \nu$ in $\QQ$. % for some $v'$ with $\tilde v' = v$.
	By construction, $\QQ$ has matching transitions	$\delta_\push, \delta_\pop$ as in \eqref{eq:push:pop:new}.
	By definition of $\psicopy$, we have
	${(\vec \x: \mu, \vec \y: \mu, \y_1: 0) \models \psicopy(\vec \x, \vec \y)}$,
	where $\y_i$ is the stack clock copy of control clock $\x_i$.
	%
	%Since control clock $\x_0$ is never reset,
	%the quantity $\delta(\mu, \nu) = \nu(\x_0) - \mu(\x_0)$ is the time elapsed between push and pop,
	%and since since stack clock $\z_0$ was $0$ at the time of push,
	%this quantity is precisely $ $
	%
	We show that
	$$(\vec \x' :\nu, \vec \y': \mu + \delta_{\mu \nu}) \models \xi_{\psi_\push, \psi_\pop}(\vec \x', \vec \y'),$$
	thus showing $\mu \reach {w} {pq} \nu$ in $\QQ$ by \eqref{eq:characterisation}
	for $w = \delta_\push \cdot v \cdot \delta_\pop$. % satisfying $\tilde w' = w$.
	By definition,
	$\xi_{\psi_\push, \psi_\pop}(\vec \x', \vec \y')$
	is equivalent to $\psi_\pop'(\vec \x', \vec \y')$ from \eqref{eq:psi:pop'}.
	Take $\mu_\ZZ + \delta_{\mu \nu}$ as the valuation for $\vec \z'$,
	and we have
	$$(\vec \x': \nu, \vec \z' : \mu_\ZZ + \delta_{\mu \nu}, \vec \y' : \mu + \delta_{\mu \nu} ) \models
		\psi_\push(\vec \y' - \y_1', \vec \z' - \y_1') \wedge \psi_\pop(\vec \x', \vec \y')$$
	because $(\vec \x': \nu, \vec \z' : \mu_\ZZ + \delta_{\mu \nu}) \models \psi_\pop(\vec \x', \vec \z')$,
	$(\vec \y : \mu, \vec \z : \mu_\ZZ) \models \psi_\push(\vec \y, \vec \z)$,
	and $\y_1'$ has final value $\delta_{\mu \nu}$ since it was $0$ at the time of push by construction.
	This concludes the push-pop case, and the proof of the lemma.
\end{proof}

\paragraph{Reconstruction of the reachability relation}

Let $\varphi_{pq}(\bar x, \bar f, \bar x')$ be a family of linear arithmetic formulas
expressing the reachability relation of $\QQ$.
We index the list of variables $\bar f$ by writing $f_\delta$ with $\delta \in \Delta'$ a transition of $\QQ$.
Let $\Delta_\push$ and $\Delta_\pop$ be the set of transitions in $\Delta$ of the form $\tilde\delta_\push$, resp., $\tilde\delta_\pop$. % as in \eqref{eq:push:pop:new}.
For a transition $\tilde\delta_\push \in \Delta_\push$,
let $f_{\delta_\push, 1}, \dots, f_{\delta_\push, m}$ be all the $f_{\delta_\push}$'s of the form
$\trule p {\pushop{\tuplesmall{\delta, \tilde\delta_\pop}}{\psicopy}} r$ with $\delta = \tilde\delta_\push$,
and similarly for $f_{\delta_\pop, 1}, \dots, f_{\delta_\pop, n}$.
(The indices $m$ and $n$ depend on $\tilde\delta_\push$, resp., $\tilde\delta_\pop$, but for simplicity we omit this dependence.)
The reachability relation of $\P$ can be expressed as 
\begin{align*}
	\tilde\varphi_{pq}(\bar x, \bar g, \bar x') \;\equiv\;
		&\exists \bar f \st \varphi_{pq}(\bar x, \bar f, \bar x') \wedge
			\bigwedge_{\tilde\delta_\push \in \Delta_\push} g_{\tilde\delta_\push} = f_{\delta_\push, 1} + \cdots + f_{\delta_\push, m} \wedge \\
				&\bigwedge_{\tilde\delta_\pop \in \Delta_\pop} g_{\tilde\delta_\pop} = f_{\delta_\pop, 1} + \cdots + f_{\delta_\pop, n} \wedge
					\bigwedge_{\delta \in \Delta \setminus (\Delta_\push \cup \Delta_\pop)} g_\delta = f_\delta.
\end{align*}

The theorem below summarises the complexity of the push-copy reduction.

\begin{theorem}
	\label{thm:push-copy}
	%A linear-arithmetic description $\set{\varphi_{pq}}$ of the reachability relation of $\P$
	%can effectively be computed from a linear arithmetic description $\set{\varphi_{pq}'}$
	%of the reachability relation of $\QQ$. 
	Computing the reachability relation of a \TPDA $\P$
	reduces to computing the reachability relation of a push-copy \TPDA $\QQ$
	with an exponential blow-up in the number of transitions.
\end{theorem}

The exponential blow-up in the number of transitions is justified as follows.
Since the $\xi_{\psi_\push, \psi_\pop}$'s are obtained applying Lemma~\ref{lem:qe-clocks} to a conjunctive formula,
it is in fact a disjunction of exponentially many conjunctive clock constraints $\xi_{\psi_\push, \psi_\pop}^j$ of polynomial size:
\begin{align*}
	\xi_{\psi_\push, \psi_\pop} \;\equiv\;
		\xi_{\psi_\push, \psi_\pop}^1 \vee \cdots \vee \xi_{\psi_\push, \psi_\pop}^m.
\end{align*}
By using nondeterminism, we split a pop transition
$\trule s {\popop{\tuplesmall{\tilde\delta_\push, \tilde\delta_\pop}}{\xi_{\psi_\push, \psi_\pop}}} q$
into transitions
$$\trule s {\popop{\tuplesmall{\tilde\delta_\push, \tilde\delta_\pop}}{\xi_{\psi_\push, \psi_\pop}^1}} q, \dots,
\trule s {\popop{\tuplesmall{\tilde\delta_\push, \tilde\delta_\pop}}{\xi_{\psi_\push, \psi_\pop}^m}} q,$$
thus arriving at a transition relation of exponential size, as claimed above.

From now on we consider push-copy \TPDA only, and consequently we use stack clocks $\y_1, \y_1, \dots$,
as copies of the corresponding control clocks $\x_0, \x_1, \dots$.
We continue with two easy preprocessing steps. % that involve no blow-up.

\oldsubsection*{Simplifying pop constraints I: No stack/stack-stack pop constraints}
\label{sec:simplify:stack-stack}

Thanks to \Cref{thm:push-copy} we can assume that the automaton is push-copy,
i.e., stack clocks are copies of control clocks.
An immediate consequence is that diagonal pop constraints involving only stack clocks
can be replaced by checking \emph{at the time of push} a transition constraint between the corresponding control clocks.
For instance, the stack-stack pop operation $\popop \alpha {y_i - y_j \sim k}$
can be replaced by $\popop \alpha \true$
provided that $\x_i - \x_j \sim k$ holds at the time of push.
In this way we can remove diagonal stack-stack pop constraints.
Non-diagonal stack constraints like $y_i \sim k$ are converted to the diagonal form $y_i - x \sim k$
where $x$ is an auxiliary control clock which is assured to have value 0 at every pop.
%
%This transformation is linear.
%
We henceforth assume that there are no stack/stack-stack pop constraints.

\oldsubsection*{Simplifying pop constraints II: (Possibly negated) atomic pop constraints}
\label{sec:simplify:atomic-pop}

\begin{figure}%
	\centering%
	%\!\!\!\!\!\!\!
	\includegraphics[scale=1]{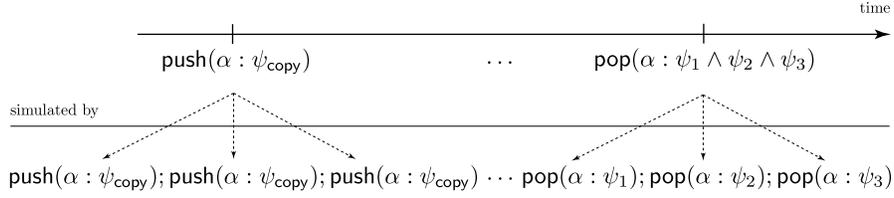}%
	\caption{Atomic pop constraints.}%
	\label{fig:atomic:pop}%
	%\mathsf{push}(\alpha : \psi_\mathsf{copy})
	%	\mathsf{pop}(\alpha : \psi_1 \wedge \psi_2 \wedge \psi_3)
	%	\mathsf{push}(\alpha : \psi_\mathsf{copy}); \mathsf{push}(\alpha : \psi_\mathsf{copy}); \mathsf{push}(\alpha : \psi_\mathsf{copy})
	%	\mathsf{pop}(\alpha : \psi_1); \mathsf{pop}(\alpha : \psi_2); \mathsf{pop}(\alpha : \psi_3)
\end{figure}

Pop operations $\popop \alpha {\psi_\pop}$
can be simplified in order for $\psi_\pop$ to be a (possibly negated) atomic clock constraint.
The idea is to push many copies of the same symbol; \cf Fig.~\ref{fig:atomic:pop}.
Formally, let $n$ be the maximum number of atomic constraints in any pop constraint $\psi_\pop$.
A push operation $\trule p {\pushop \alpha {\psicopy}} q$ is replaced by pushing $n$ copies
$\alpha_1, \dots, \alpha_n$ of $\alpha$:
\begin{align*}
	\trule p {\pushop {\alpha_1} {\psicopy}; \pushop {\alpha_2} {\psicopy}; \cdots; \pushop {\alpha_n} {\psicopy}} q.
\end{align*}
A pop operation $\trule p {\popop \alpha {\psi_1 \wedge \cdots \wedge \psi_n}} q$
(where we allow the same $\psi_i$ to appear many times) is replaced by
\begin{align*}
	\trule p {\popop {\alpha_1} {\psi_1}; \cdots; \popop {\alpha_n} {\psi_n}} q.
\end{align*}
The correctness of this transformation uses the fact that the \TPDA is push-copy.
(More generally, it suffices that there is no guessing of stack clocks at the time of push,
because we cannot enforce that the same guesses is made $n$ times.)
The complexity of this step is a linear blow-up in the number of transitions.

%\begin{remark}
%	Incidentally, this transformation associates to every stack symbol $\alpha$
%	an atomic pop constraint $\psi_\alpha$, which will be convenient in the following.
%\end{remark}

%\subsection{Simplifying push operations II: Push-zero}
%\ilorenzo{NOT CLEAR HOW}

%The aim of this section is to show that push operations can be further simplified
%as to ensure that a single stack clock is sent with initial value $0$:
%
%\begin{align}
%	\psizero (\y_1) \;\equiv\; \y_1 = 0.
%\end{align}
%
%The intuition is that, every time a clock $x_i$ is reset,
%we push on the stack a corresponding stack symbol $\tilde x_i$ with initial age $0$.

\subsection{Simplifying pop constraints III: No integral pop constraints}
\label{sec:simplify:pop-integer-free}

\newcommand{\alphapsi}{\alpha}

A \TPDA is \emph{pop-integer-free} if pop operations do not have integral constraints $\floor \x - \floor \y \sim k$,
i.e., they only have modular and fractional constraints.
The aim of this section is to remove such integral constraints from pop transitions,
while being able to reconstruct the reachability relation.
In a preliminary step,
we convert $\floor \x - \floor \y \sim k$ into fractional and classical constraints $\x - \y \sim k$; cf.~Remark~\ref{rem:sugar}.
The reason for doing this is that the semantics of classical diagonal constraints
is invariant under time elapse (which is not true for integral constraints)
and this will simply the proof of \Cref{thm:pop-integer-free}.

It thus remains to remove classical pop constraints of the form 
\begin{align}
	\label{eq:psi}
	\psi \equiv y_i - x_j \sim k.
\end{align}
Let $\P = (\Sigma, \Gamma, \L, \X, \ZZ, \Delta)$ be a push-copy \TPDA
and fix a pop constraint $\psi$ as above
occurring in a pop operation of the form $\popop \alpha \psi$.
For convenience we assume, w.l.o.g., that every stack symbol $\alpha$ appears   
with a unique pop constraint, i.e., there are no two pop operations
$\popop \alpha {\psi_1}$ and $\popop \alpha {\psi_2}$ with the same 
stack symbol but different clock constraints.
The idea is to introduce few extra \emph{control clocks} of the form $x_\psi$
and replace every occurrence of pop operation $\popop \alpha \psi$ with constraint $\psi$
with a sequence of two operations of the form $\popop \alpha \true; \testop{x_\psi - x_j \sim k}$.
In other words, we simulate a stack constraint with a \emph{transition constraint}.
The cost of removing one such $\psi$ is to add a constant number of clocks and stack symbols,
and multiply the number of control locations by a constant.
By iterating the construction we can remove all such pop constraints;
%finally, classical control constraints in $\widetilde \psi$
%can be converted back into integral and fractional control constraints according to Remark~\ref{rem:sugar}.
%
the construction preserves all the other constraints present in the automaton.
The combined complexity stated below follows from the fact that there are at most $\card \Delta$ pop constraints $\psi$'s.

\begin{theorem}
	\label{thm:pop-integer-free}
	For every \TPDA $\P = (\Sigma, \Gamma, \L, \X, \ZZ, \Delta)$
	we can produce a \TPDA $\QQ = (\Sigma, \Gamma', \L', \X', \ZZ, \Delta')$
	not containing pop integral constraints $\floor \x - \floor \y \sim k$
	\st the reachability relation of $\P$ is efficiently computable from that of $\QQ$.
	The complexity of the construction is
	$\card{\Gamma'} = O(\card \Gamma \cdot \card \Delta)$,
	$\card{\X'} = O(\card \X \cdot \card \Delta)$,
	$\card{\L'} = \card \L \cdot 2^{O(\card \Delta)}$, and
	$\card{\Delta'} = \card \Delta \cdot 2^{O(\card \Delta)}$.
\end{theorem}

It is remarkable that such a simulation is at all possible.
A priori, each new push-pop pair creates a novel timing constraint on the run.
This should be contrasted with \emph{fractional stack constraints},
which cannot be removed by adding control clocks \cite{UezatoMinamide:LPAR15}.

If the original \TPDA $\P$ contained only classical stack constraints
(i.e., neither modular $\floor \x - \floor \y \eqv m k$ nor fractional $\fract x \leq \fract y$ stack constraints, as it is the case with \dtPDA),
then $\QQ$ won't have any stack constraints \emph{at all}.
Therefore, the stack of $\QQ$ is essentially \emph{untimed},
which allows us to recover (and actually generalise) the following result previously announced in \cite[Theorem II.1]{ClementeLasota:LICS:2015}.

\begin{corollary}[Stack untiming \cite{ClementeLasota:LICS:2015}]
	\label{cor:untiming}
	\TPDA without modular nor fractional stack constraints (such as \dtPDA)
	effectively recognise the same class of timed languages as \TPDA with untimed stack.
\end{corollary}

\paragraph{Intuition}

\begin{figure}
	\centering
	\includegraphics[scale=1]{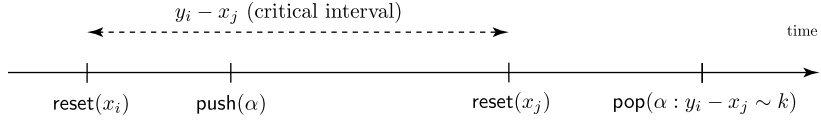}
	\caption{The critical interval of a pop constraint $y_i - x_j \sim k$.}
	\label{fig:critical}
\end{figure}

In this section we provide some intuition behind the formal construction leading to \Cref{thm:pop-integer-free},
which will be presented in \Cref{sec:simplify:pop-integer-free:A,sec:simplify:pop-integer-free:B}.

Fix a pop constraint $\psi$ as in \eqref{eq:psi}.
The \emph{critical interval} of a matching push and pop operation
is the interval between the last reset of $x_i$ before push
and the last reset of $x_j$ before pop;
cf.~Fig.~\ref{fig:critical}.
(In the picture $\resetop{x_i}$ happens before $\resetop{x_j}$, but it could also happen the other way around.)
The position in time of the two resets above is the only information necessary to determine
whether the constraint $\psi$ holds,
because the stack clock $y_i$ is a copy of the control clock $x_i$ at the time of push
(thanks to Sec.~\ref{sec:simplify:push-copy}).
We would like that critical intervals are either nested within each other,
or disjoint.

\begin{figure}
	\centering
	\includegraphics[scale=1]{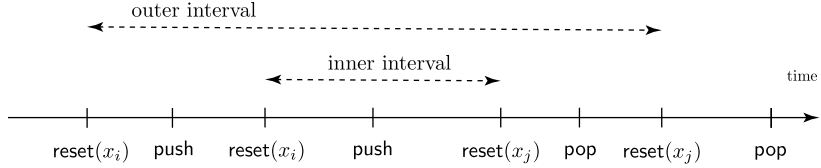}
	\caption{Nested critical intervals.}
	\label{fig:critical_nested}
\end{figure}

If we have nested push and pop operations,
then also the corresponding critical intervals are nested; cf.~Fig.~\ref{fig:critical_nested}.
In nested push and pop operations,
the left endpoint $\resetop {x_i}$ of the inner critical interval can only (possibly) move further to the left:
This happens when there is an earlier reset of $x_i$ before the outer push.
Symmetrically, the right endpoint $\resetop{x_j}$ can only (possibly) move further to the right,
which happens when there is a later reset of $x_j$ before the outer pop.
Thus inner critical intervals are shorter than the outer ones,
with useful consequences regarding which pop constraints we really need to check.

\begin{figure}
	\centering
	\includegraphics[scale=1]{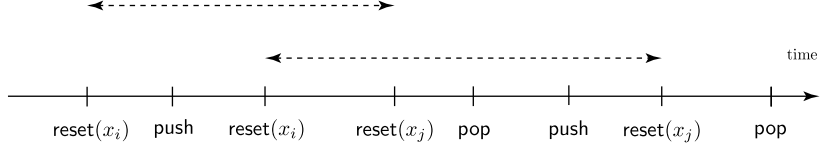}
	\caption{Overlapping critical intervals.}
	\label{fig:critical_overlapping}
\end{figure}

Sometimes the critical intervals are neither nested nor disjoint.
One such example is shown in Fig.~\ref{fig:critical_overlapping},
where in a push-pop-push-pop sequence (of the same symbol)
the second reset of $x_i$ happens before the first reset of $x_j$.
The crucial observation in such cases is that
at most \emph{two} critical intervals can overlap at any given moment.
For this reason, in \Cref{sec:simplify:pop-integer-free:A,sec:simplify:pop-integer-free:B}
we actually use two copies $x_i^0, x_i^1$ of $x_i$, instead of only one.

We say that a pop constraint $\psi$ is \emph{active} in the interval
between $\pushop{\alpha}{\psicopy}$ and its matching $\popop{\alpha}{\psi}$.
%provided that $\psi$ is a conjunct of $\psi_\pop$,
%which we denote by $\psi \in \psi_\pop$.
%
The way in which we simulate an active constraint $\psi$ depends on whether it is of type A or B.

\begin{figure}
	\centering
	\includegraphics[scale=1]{./imgs/fig_critical_nested_A.ai}
	\caption{A type A active constraint $\psi \equiv y_i - x_j \precsim k$.}
	\label{fig:critical_nested_A}
\end{figure}

Type A constraints are of the form
\begin{align}
	\tag{\textbf{A}}
	\psi \equiv y_i - x_j \precsim k, \qquad \textrm{ where } \precsim \;\in\! \set{\leq, <}.
\end{align}
Outer type A constraints subsume the inner ones,
thus it suffices to verify the outermost critical interval;
\cf~Fig.~\ref{fig:critical_nested_A}.
%
%This is achieved as follows.
%
The picture also shows a simplified scheme to remove type A constraints without overlapping intervals,
which works as follows.
The automaton makes $\psi$ active the first time $\alphapsi$ is pushed on the stack and $\psi$ was inactive.
We introduce an additional control clock $x_\psi$,
which is reset whenever $x_i$ is reset and $\psi$ is not active.
When $\psi$ is active, we never reset $x_\psi$.
A $\popop {\alphapsi} {\psi}$ %with $\psi \in \psi_\pop$
is simulated by $\popop {\alphapsi} \true; \testop{x_\psi - x_j \precsim k}$.
(The actual construction is slightly more complicated in order to deal with overlapping intervals; \cf~Sec.~\ref{sec:simplify:pop-integer-free:A}.)
%\end{align*}
%For simplicity, we test not only the outermost critical interval,
%but also some additional intervals nested within it.
%
%The test above involves only the control clocks $x_\psi$ and $x_j$.
%We have thus removed the pop constraint $\psi$.

\begin{figure}
	\centering
	\includegraphics[scale=1]{./imgs/fig_critical_nested_B.ai}
	\caption{A type B active constraint $\psi \equiv y_i - x_j \succsim k$.}
	\label{fig:critical_nested_B}
\end{figure}

Type B constraints are of the form 
\begin{align}
	\tag{\textbf{B}}
	\psi \equiv y_i - x_j \succsim k, \qquad \textrm{ where } \succsim \;\in\! \set{\geq, >}.
\end{align}
Since inner type B constraints subsume the outer ones,
it suffices to verify only the innermost active type B critical interval;
\cf~Fig.~\ref{fig:critical_nested_B}.
A simplified scheme to remove type B constraints in the absence of overlapping intervals is as follows.
The automaton guesses that the type B pop constraint $\psi \equiv y_i - x_j \succsim k$ becomes active.
The additional control clock $x_\psi$ is reset every time $x_i$ is reset.
%
%(Thanks to our preprocessing,
%if $\psi$ is active, then every reset of $x_i$ is followed by a push,
%and thus by a nested critical interval.)
%
A $\popop {\alphapsi} \psi$ is simulated by $\popop {\alphapsi} \true; \testop{x_\psi - x_j \succsim k}$.
(The actual construction is slightly more complicated in order to deal with overlapping intervals; \cf~Sec.~\ref{sec:simplify:pop-integer-free:B}.)
%\end{align*}
%
%While it suffices to check only the innermost interval,
%for simplicity we additionally check some more larger intervals.
%
%	$\mathsf{push} \quad \mathsf{push} \quad \mathsf{push} \quad \mathsf{pop}(\psi) \quad \mathsf{pop}(\psi) \quad \mathsf{pop}(\psi)$
%	$\mathsf{push} \quad \mathsf{push} \quad \mathsf{reset}(x_i) \quad \mathsf{push} \quad \mathsf{pop}(\psi) \quad \mathsf{pop}(\psi) \quad \mathsf{pop}(\psi)$
%	$\mathsf{test}(x_\psi - x_j \succsim k)$
%
We have two constructions, depending on whether $\psi$ is of type A or of type B.

\subsubsection{The construction---Type A}
\label{sec:simplify:pop-integer-free:A}

%\ilorenzo{$\L' \mapsto \L_A$}

For $\psi$ of type $A$, we construct a \TPDA
$\P_A = (\Sigma, \Gamma', \L_A, \X', \ZZ, \Delta_A)$ as follows.
We add two extra copies $x_i^0$ and $x_i^1$ of clock $x_i$;
$\X' = \X \cup \set{x_i^0, x_i^1}$.
%
%The constructions of $\P_A$ and $\P_B$ are very similar;
%the only difference is when $x_\psi$ is reset.
%
A control location is either of the form $(p, d)$ or $(p, d, e)$,
where $p \in \L$ is the current control location.
The index $d \in \set{0, 1}$ indicates that $x_i^d$ is the \emph{tracking copy} of $x_i$
i.e., whenever $x_i$ is reset, so is $x_i^d$.
The index $e \in \set{0, 1}$ in $(p, d, e)$
indicates that $x_i^e$ is the \emph{frozen copy} of $x_i$,
which is not reset anymore and used to store a previous value of $x_i$.
Thus, $L_A = L \times (\set {0, 1} \cup \set {0, 1}^2)$.
We introduce a new stack symbol $\hat\alphapsi$ which denotes that $\psi$ becomes active when pushed on the stack;
thus, $\Gamma' = \Gamma \cup \set{\hat\alphapsi}$.

Transitions in $\Delta_A$ are as follows.
Let $\delta = \trule p \op q \in \Delta$ be a transition in $\P$.
If it is either an input $\op = \readop a$,
test $\op = \testop \varphi$,
time elapse $\op = \elapse$ transition,
reset $\op = \resetop \Y$ not resetting $x_i \not\in \Y$,
push $\op = \pushop \gamma {\psicopy}$ with $\gamma \neq \alphapsi$
(where $\psicopy$ is defined in \eqref{eq:copy:push:constraint}),
or pop $\op = \popop \gamma {\psipop}$ with $\gamma \neq \alphapsi$,
then it generates corresponding transitions in $\P_A$
\begin{align}
	\label{eq:A:generic}
	&\delta_1  = \trule {(p, d)} \op {(q, d)} \textrm{ and } \delta_2 = \trule {(p, d, e)} \op {(q, d, e)} \in \Delta_A,
\end{align}
for every choice of $d, e$.
A reset transition $\op = \resetop {\Y \cup \set{x_i}}$ resetting $x_i$
generates transitions in $\P_A$ of the form
\begin{align}
	\label{eq:A:reset:1}	
	&\trule {(p, d)} {\resetop {\Y \cup \set{x_i, x_i^d}}} {(q, d)}, \\
	\label{eq:A:reset:2}
	&\trule {(p, d, e)} {\resetop {\Y \cup \set{x_i, x_i^{1-d}}}} {(q, d, 1 - d)} \in \Delta_A.
	%\in \Delta_A
\end{align}
A push transition $\op = \pushop {\alphapsi} {\psicopy}$
generates transitions
\begin{align}
	\label{eq:A:push:1}
	&\trule {(p, d)} {\pushop {\hat\alphapsi} {\psicopy}} {(q, d, d)}, \\
	\label{eq:A:push:2}
	&\trule {(p, d, e)} {\pushop {\alphapsi} {\psicopy}} {(q, d, e)} \in \Delta_A.
\end{align}
A pop transition $\op = \popop{\alphapsi} {\psi}$ generates transitions%
%\lorenzo{having $d$ in $(q, 0, e, d)$ is convenient as bookkeeping for the proof}
%
\begin{align}
	\label{eq:A:pop:1}
	&\trule{(p, d, e)} {\popop{\hat\alphapsi}{\true};\; \testop{x_i^d - x_j \precsim k}} {(q, e)}, \\
	\label{eq:A:pop:2}
	&\trule{(p, d, e)} {\popop{\alphapsi}{\true}} {(q, d, e)} \in \Delta_A.
\end{align}
This concludes the description of $\Delta_A$, and of $\P_A$.
For each new transition $\delta_A$ added in the equations \Cref{eq:A:generic,eq:A:reset:1,eq:A:reset:2,eq:A:push:1,eq:A:push:2,eq:A:pop:1,eq:A:pop:2} above,
let $\tilde \delta_A = \delta \in \Delta$ be the originating transition%
\footnote{
	Formally speaking, in \eqref{eq:A:pop:2} there may be different $\delta$'s inducing the same $\delta_A$,
	and thus $\tilde \delta_A$ would not be well defined.
	This can be avoided by recording in the stack symbol $\alpha$ the pop transition $\delta$,
	but we avoid it for simplicity.
} of $\P$.
The mapping ``$\;\tilde{}\;$'' is extended pointwise to a mapping $\Delta_A^* \to \Delta^*$ and it will be used in the correctness statements below.
Recall that ${\delta_{\mu\nu} := \nu(x_0) - \mu(x_0)}$ is the total time elapsed in the run (since $x_0$ is never reset).

%\ilorenzo{New \Cref{ex:type_A} and \Cref{ex:type_A:critical} showing the construction:}

\begin{example}
	\label{ex:type_A}
	We illustrate the construction on an example.
	Consider the timed language $L$ over the ternary alphabet $\Sigma = \set{a, b, c}$
	consisting of all timed words $w$
	whose untiming is of the form $a^n (c^* b)^n$ with $n \in \N$
	\st the amount of time between an $a$ and the last $c$ before the matching $b$ is at most $5$:
	\begin{align*}
		w = (a, t_n) \cdots (a, t_1) \cdots (c, t_{i_1}) (b, t_{i_1+1}) \cdots (c, t_{i_2}) (b, t_{i_2+1}) \cdots (c, t_{i_n}) (b, t_{i_n+1})
	\end{align*}
	%\islawek{small correction: $t_n \ldots t_1$}
	%
	and, for every $1 \leq j \leq n$, $t_{i_j} - t_j \leq 5$.
	A direct translation of $L$ yields the following \TPDA.
	There are two control locations $p, q$,
	two control clocks $x_1, x_2$, one stack clock $y_1$,
	and one stack symbol $\Gamma = \set \alpha$.
	In $p$ the automaton can read $a$ and pushes it on the stack together with $x_1$,
	which is reset for the occasion.
	In $p$ the automaton can also nondeterministically move to $q$, without performing any action.
	In $q$ the automaton can read a letter $c$ and reset $x_2$,
	or it can read a letter $b$ and pop the stack
	provided the timing constraint holds:
	\begin{align*}
		&\tuple{p, \elapse; \resetop {x_1}; \readop a; \pushop {\alpha} {y_1 = x_1}, p}, \\
		&\tuple{p, \readop \varepsilon, q}, \\
		&\tuple{q, \elapse; \resetop {x_2}; \readop c; q}, \\
		&\tuple{q, \elapse; \readop b; \popop {\alpha} {y_1 - x_2 \leq 5}; q}.
	\end{align*}
	The language $L$ is recognised by $A$
	by considering $p$ the initial control location,
	$q$ the final one,
	and the stack is empty at the beginning and at the end of the run.
	This concludes the description of the automaton.
	The automaton obtained by applying the transformation
	to remove the type A classical pop constraint $y_1 - x_2 \leq 5$
	yields a \TPDA $A'$ with new stack alphabet $\Gamma' = \set {\alpha, \hat \alpha}$,
	a new set of control clocks $\X = \set{x_1, x_2, x_1^0, x_1^1}$,
	%\islawek{I guess you want to use here clocks $\X = \set{x_1, x_2, x_1^0, x_1^1}$ for agreement	with the old automaton}
	%
	control locations of the form 
	$(p, d), (p, d, e), (q, d), (q, d, e)$
	for every $d, e \in \set{0, 1}$,
	and transitions
	\begin{align*}
		&\tuple{(p, d), \elapse; \resetop {x_1, x_1^d}; \readop a; \pushop {\hat\alpha} {y_1 = x_1}, (p, d, d)}, \\
		&\tuple{(p, d, e), \elapse; \resetop {x_1, x_1^{1-d}}; \readop a; \pushop \alpha {y_1 = x_1}, (p, d, 1-d)}, \\
		&\tuple{(p, d, e), \readop \varepsilon, (q, d, e)}, \\
		&\tuple{(q, d, e), \elapse; \resetop {x_2}; \readop c, (q, d, e)}, \\
		&\tuple{(q, d, e), \elapse; \readop b; \popop \alpha \true, (q, d, e)}, \\
		&\tuple{(q, d, e), \elapse; \readop b; \popop {\hat \alpha} \true; \testop {x_1^d - x_2 \leq 5}, (q, e)},
	\end{align*}
	for every $d, e \in \set{0, 1}$.
	The new \TPDA $A'$ does not contain integral/classical pop constraints
	and recognises the same language $L$
	if we consider $(p, 0)$ the initial control location
	and $(q, 0), (q, 1)$ the final ones.
	In this particular example, there are only critical nested intervals (c.f.~\Cref{fig:critical_nested})
	and no critical overlapping intervals (c.f.~\Cref{fig:critical_overlapping}),
	and thus the automaton be simplified to use only one copy $x_1^0$ of $x_1$
	instead of two copies $x_1^0, x_1^1$.
\end{example}

\begin{example}
	\label{ex:type_A:critical}
	In this example we show that two copies $x_c^0, x_c^1$ of a control clock $x_c$
	are required in the presence of overlapping critical intervals (c.f.~\Cref{fig:critical_overlapping}).
	Consider the timed language $L$ over the four-letter alphabet
	$\Sigma = \set {a, b, c, f}$
	consisting of all timed words $w$ s.t.~\begin{inparaenum}[1)]
		\item for every prefix $u$ of $w$
		the number of $b$'s is less than or equal to the number of $a$'s,
		\item for every occurrence of $a$ and its matching occurrence of $b$ (if any)\footnote{
			We say that an occurrence $a_j = b$ of $b$
			matches an occurrence $a_i = a$ of $a$
			in a word $w = a_1 \cdots a_n$
			if $i < j$ and $j$ is the smallest index $> i$ \st the number of occurrences of $a$ in $u = a_i \cdots a_j$ is the same as the number of occurrences of $b$ in $u$.
		}
		the amount of time between the last occurrence of $c$ before this occurrence of $a$
		and the last occurrence of $f$ before this occurrence of $b$ (if any) is at most $5$,
		\item if there is no such $c$ or $f$ in the previous point,
		then the measurement is conventionally done since the beginning of the word
		(i.e., at time zero).
	\end{inparaenum}
	This language is recognised by the \TPDA $A$
	containing one control location $p$ (which is regarded as both initial and final),
	together with transitions
	\begin{align*}
		&\delta_\tau = \tuple{p, \elapse, p},\\
		&\delta_a = \tuple{p, \readop a; \pushop {\alpha} {y = x_c}, p}, \\
		&\delta_b = \tuple{p, \readop b; \popop {\alpha} {y - x_d \leq 5}, p}, \\
		&\delta_c = \tuple{p, \readop c; \resetop {x_c}, p}, \\
		&\delta_f = \tuple{p, \readop f; \resetop {x_f}, p}.
	\end{align*}
	The \TPDA $A'$ obtained by removing the classical stack constraint $y - x_d \leq 5$
	according to the reduction of this section
	contains control locations $(p, d), (p, d, e)$ for every $d, e \in \set{0, 1}$
	and all transitions of the form
	\begin{align*}
		&\delta_\tau^d = \tuple{(p, d), \elapse, (p, d)},
			\delta_\tau^{de} = \tuple{(p, d, e), \elapse, (p, d, e)}, \\
		&\delta_a^d = \tuple{(p, d), \readop a; \pushop {\hat\alpha} {y = x_c}, (p, d, d)}, \\
		&\delta_a^{de} = \tuple{(p, d, e), \readop a; \pushop \alpha {y = x_c}, (p, d, 1-d)}, \\
		&\delta_b^{de} = \tuple{(p, d, e), \readop b; \popop \alpha \true, (p, d, e)}, \\
		&\delta_{b'}^{de} = \tuple{(p, d, e), \readop b; \popop {\hat \alpha} \true; \testop {x_c^d - x_f \leq 5}, (p, e)}, \\
		&\delta_c^d = \tuple{(p, d),  \readop c; \resetop {x, x_c^d}; (p, d)}, \\
		&\delta_c^{de} = \tuple{(p, d, e), \readop c; \resetop {x, x_c^{1-d}}; (p, d, 1-d)}, \\
		&\delta_f^d = \tuple{(p, d), \readop f; \resetop {x_f}; (p, d)}, \\
		&\delta_f^{de} = \tuple{(p, d, e), \readop f; \resetop {x_f}, (p, d, e)}.
	\end{align*}
	The new \TPDA $A'$ does not contain classical pop constraints,
	which have been replaced with the local constraint $x_c^d - x_f \leq 5$
	between control clocks $x_c^d, x_f$.
	We claim that both copies $x_c^0, x_c^1$ of $x_c$ are required for correctness.
	For instance, consider the following timed word
	\begin{align*}
		w = (c, 0) (a, 1) (c, 3) (f, 5) (b, 6) (a, 7) (f, 8) (b, 10).
	\end{align*}
	The timed word $w$ contains two critical overlapping intervals as in \Cref{fig:critical_overlapping}:
	The first one is between $(c, 0)$ and $(f, 5)$
	and the second one is between $(c, 3)$ and $(f, 8)$.
	The word $w$ is accepted by the original automaton $A$
	since $5 - 0 \leq 5$ and $8 - 3 \leq 5$.
	Indeed, it induces the following run in $A$:
	Let $(p, t_c, t_f, w)$ be the configuration where clock $x_c$ has value $t_c$,
	clock $x_f$ has value $t_f$,
	and $w$ is the content of the stack.
	\begin{align*}
		(p, 0, 0, \varepsilon)
			&\reach {\delta_c} {} (p, 0, 0, \varepsilon) 
				&\text{} \\
			&\reach {\delta_\tau} {} (p, 1, 1, \varepsilon)
				\reach {\delta_a} {} (p, 1, 1, (\alpha, 1))
				&\text{(push)}\\
			&\reach {\delta_\tau} {} (p, 3, 3, (\alpha, 3))
				\reach {\delta_c} {} (p, 0, 3, (\alpha, 3))
				&\text{} \\
			&\reach {\delta_\tau} {} (p, 2, 5, (\alpha, 5))
				\reach {\delta_f} {} (p, 2, 0, (\alpha, 5))
				&\text{}\\
			&\reach {\delta_\tau} {} (p, 3, 1, (\alpha, 6))
				\reach {\delta_b} {} (p, 3, 1, \varepsilon)
				&\text{(pop)} \\
			&\reach {\delta_\tau} {} (p, 4, 2, \varepsilon)
				\reach {\delta_a} {} (p, 4, 2, (\alpha, 4))
				&\text{(push)}\\
			&\reach {\delta_\tau} {} (p, 5, 3, (\alpha, 5))
				\reach {\delta_f} {} (p, 5, 0, (\alpha, 5))
				&\text{}\\
			&\reach {\delta_\tau} {} (p, 7, 2, (\alpha, 7))
				\reach {\delta_b} {} (p, 7, 2, \varepsilon).
				&\text{(pop)}
	\end{align*}
	Configurations in $A'$ are of the form
	$((p, e), t_c^0, t_c^1, t_f, w)$ or
	$((p, d, e), t_c^0, t_c^1, t_f, w)$.
	(We avoid the valuation of clock $x_c$ for simplicity
	since it always equals the tracking copy $x_c^e$.)
	The of $A'$ over $w$ above is
	\begin{align*}
		((p, 0), 0, 0, 0, \varepsilon)
			&\reach {\delta_c} {} ((p, 0), 0, 0, 0, \varepsilon) 
				&\text{} \\
			&\reach {\delta_\tau} {} ((p, 0), 1, 1, 1, \varepsilon)
				\reach {\delta_a} {} ((p, 0, 0), 1, 1, 1, (\hat \alpha, 1))
				&\text{(push)}\\
			&\reach {\delta_\tau} {} ((p, 0, 0), 3, 3, 3, (\hat \alpha, 3))
				\reach {\delta_c} {} ((p, 0, 1), 3, 0, 3, (\hat \alpha, 3))
				&\text{} \\
			&\reach {\delta_\tau} {} ((p, 0, 1), 5, 2, 5, (\hat \alpha, 5))
				\reach {\delta_f} {} ((p, 0, 1), 5, 2, 0, (\hat \alpha, 5))
				&\text{}\\
			&\reach {\delta_\tau} {} ((p, 0, 1), 6, 3, 1, (\hat \alpha, 6))
				\reach {\delta_b} {} ((p, 1), 6, 3, 1, \varepsilon)
				&\text{(pop)} \\
			&\reach {\delta_\tau} {} ((p, 1), 7, 4, 2, \varepsilon)
				\reach {\delta_a} {} ((p, 1, 1), 7, 4, 2, (\hat \alpha, 4))
				&\text{(push)}\\
			&\reach {\delta_\tau} {} ((p, 1, 1), 8, 5, 3, (\hat \alpha, 5))
				\reach {\delta_f} {} ((p, 1, 1), 8, 5, 0, (\hat \alpha, 5))
				&\text{}\\
			&\reach {\delta_\tau} {} ((p, 1, 1), 10, 7, 2, (\hat \alpha, 7))
				\reach {\delta_b} {} ((p, 1), 10, 7, 2, \varepsilon).
				&\text{(pop)}
	\end{align*}
	Crucially, in control location $(p, 1)$ after the first pop transition $\delta_b$
	the automaton remembers that clock $x_c^1$ (which has value $3$)
	is the tracking one
	and it should be copied to the stack in the next push
	(and not $x_c^0$, which has value $6$).
	In this way, in the last pop transition $\delta_b$
	the automaton correctly checks
	the local clock constraint $x_c^1 - x_f \leq 5$ (which holds since $x_c^1 = 7$ and $x_f = 2$)
	and not the incorrect one $x_c^0 - x_f \leq 5$.
\end{example}

% NOTE: The two lemmas below are not used anywhere
% The following two \Cref{lem:A:prop:a,lem:A:prop:b} follow directly from the construction of $\P_A$.
% %
% \begin{lemma}
% 	\label{lem:A:prop:a}
% 	If $(p, d, e), \mu \reach {} {} (q, d, f), \nu$ and $\mu(x_i) = \mu(x_i^e)$,
% 	then $\nu(x_i) = \nu(x_i^f)$ and $\nu(x_i^d) = \mu(x_i^d) + \delta_{\mu\nu}$.
% \end{lemma}

% \begin{lemma}
% 	\label{lem:A:prop:b}
% 	If $(p, d), \mu \reach {} {} (q, e), \nu$ and $\mu(x_i) = \mu(x_i^d)$,
% 	then $\nu(x_i) = \nu(x_i^e)$.
% \end{lemma}

The following lemma states the correctness of the construction.
%
% two \Cref{lem:soundness:A,lem:completeness:A} below state the correctness of the construction.
%Combining \Cref{lem:soundness:A,lem:completeness:A}
%we obtain the correctness of the construction.

\begin{restatable}[Correctness \protect{[A]}]{lemma}{lemCorrectnessA}
	\label{lem:correctness:A}
	For control locations $p, q \in L$, clock valuations $\mu, \nu \in \Rgeq^\X$, sequence of transitions $w \in \Delta_A^*$, and flag $e \in \set{0, 1}$,
	let $\mu' = \mu[x_i^0 \mapsto \mu(x_i)], \nu'_e = \nu[x_i^e \mapsto \nu(x_i)]$.
	Then,
	\begin{align*}
		p, \mu \reach {\tilde w} {} q, \nu
			\quad \textrm{ iff } \quad
				\exists e \in \set{0, 1} \st (p, 0), \mu' \reach w {} (q, e), \nu'_e.
	\end{align*}
\end{restatable}

\paragraph{Reconstruction of the reachability relation \protect{[A]}}

\Cref{lem:correctness:A} immediately allows us to reconstruct the reachability relation of $\P$ from that of $\P_A$:
Let $\varphi_{(p, d)(q, e)}(\bar x, \bar f, \bar x')$ express the reachability relation $(p, d), \_ \reach {} {} (q, e), \_$ of $\P_A$.
For simplicity, we index the variables in $\bar f$ as $f_\delta$, where $\delta \in \Delta_A$.
For a transition $\tilde\delta \in \Delta$,
let $f_{\delta, 1}, \dots, f_{\delta, m}$ be all the $f_\delta$'s \st $\delta = \tilde\delta$.
(Formally, the index $m$ depends on $\tilde\delta$, but for simplicity we omit this dependence.)
Then the reachability relation $p, \_ \reach {} {} q, \_$ of $\P$ can be expressed as
(we write $\bar y$ instead of $\bar x'$ for readability)
\begin{align*}
	\varphi_{pq}(\bar x, \bar g, \bar y) \;\equiv\;
		\exists \bar f \st \bigvee_{e \in \set{0, 1}} \varphi_{(p, 0)(q, e)}(\bar x, x_i^0, x_i^1, \bar y, y_i^0, y_i^1) \wedge x_i^0 = x_i \wedge y_i^e = y_i \wedge \\
		\wedge\bigwedge_{\tilde\delta \in \Delta} g_{\tilde\delta} = f_{\delta, 1} + \cdots + f_{\delta, m}.
\end{align*}

\subsubsection{The construction---Type B}
\label{sec:simplify:pop-integer-free:B}

The construction of $\P_B$ is similar to $\P_A$,
except for the set of control locations,
which are now of the form
$L_B = L \times \set {0, 1, 2} \times \set {0, 1}^2$,
and for the set of transitions, which we describe next.
Let $\trule p \op q \in \Delta$ be a transition in $\P$.
If it is either an input $\op = \readop a$,
test $\op = \testop \varphi$,
time elapse $\op = \elapse$ transition,
reset $\op = \resetop \Y$ not resetting $x_i \not\in \Y$,
push $\op = \pushop \gamma {\psicopy}$ with $\gamma \neq \alphapsi$
(where $\psicopy$ is defined in \eqref{eq:copy:push:constraint}),
or pop $\op = \popop \gamma {\psipop}$ with $\gamma \neq \alphapsi$,
then it generates corresponding transitions in $\P_B$
\begin{align}
	\label{eq:B:generic}
	&\trule {(p, b, d, d)} \op {(q, b, d, d)}, b \in \set{0, 2}. \\
	&\trule {(p, 1, d, e)} \op {(q, 1, d, e)}.
\end{align}
A reset transition $\op = \resetop {\Y \cup \set{x_i}}$ resetting $x_i$ generates transitions
\begin{align}
	\label{eq:B:reset:02}
	&\trule {(p, b, d, d)} {\resetop {\Y \cup \set{x_i, x_i^d}}} {(q, b, d, d)}, b \in \set{0, 2}, \\
	\label{eq:B:reset:1}
	&\trule {(p, 1, d, e)} {\resetop {\Y \cup \set{x_i, x_i^{1-d}}}} {(q, 1, d, 1 - d)}.
	%\in \Delta_A
\end{align}
A push transition $\op = \pushop {\alphapsi} {\psicopy}$ generates transitions
\begin{align}
	\label{eq:B:push:0}
	&\trule {(p, b, d, d)} {\pushop {\alphapsi} {\psicopy}} {(q, 0, d, d)}, b \in \set{0, 2}, \\
	\label{eq:B:push:1}
	&\trule {(p, b, d, d)} {\pushop {\hat\alphapsi} {\psicopy}} {(q, 1, d, d)}, b \in \set{0, 2}.
\end{align}
A pop transition $\op = \popop{\alphapsi} {\psi}$ generates transitions%
\begin{align}
	\label{eq:B:pop:1:0}
	&\trule{(p, 1, d, e)} {\popop{\hat\alphapsi}{\true};\; \testop{x_i^d - x_j \succsim k}} {(q, 2, e, e)}, \\
	\label{eq:B:pop:1:1}
	&\trule{(p, 2, d, d)} {\popop{\alphapsi}{\true}} {(q, 2, d, d)}.
\end{align}
The flag $b = 2$ ensures that an outer push-pop pair can be performed \eqref{eq:B:pop:1:1}
only if it contains a nested push-pop pair for which $\psi$ has been checked \eqref{eq:B:pop:1:0}.
For each new transition $\delta_B$ added in the equations \Cref{eq:B:generic,eq:B:reset:02,eq:B:reset:1,eq:B:push:0,eq:B:push:1,eq:B:pop:1:0,eq:B:pop:1:1} above,
let $\tilde \delta_B = \delta \in \Delta$ be the originating transition%
\footnote{
	Formally speaking, in \eqref{eq:B:pop:1:1} there may be different $\delta$'s inducing the same $\delta_B$,
	and thus $\tilde \delta_B$ would not be well defined.
	This can be avoided by recording in the stack symbol $\alpha$ the pop transition $\delta$,
	but we avoid it for simplicity.
} of $\P$.
The mapping ``$\;\tilde{}\;$'' is extended pointwise to a mapping $\Delta_B^* \to \Delta^*$ and it will be used in the correctness statements below.
%
%Combining \Cref{lem:soundness:B,lem:completeness:B} we obtain the correctness of the construction.
%
The following lemma states the correctness of the construction.

\begin{restatable}[Correctness \protect{[B]}]{lemma}{lemCorrectnessB}
	\label{lem:correctness:B}
	For control locations $p, q \in L$, clock valuations $\mu, \nu \in \Rgeq^\X$, sequence of transitions $w \in \Delta_B^*$, and flag $e \in \set{0, 1}$,
	let $\mu' = \mu[x_i^0 \mapsto \mu(x_i)], \nu'_e = \nu[x_i^e \mapsto \nu(x_i)]$.
	Then,
	\begin{align*}
		p, \mu \reach {\tilde w} {} q, \nu
			\quad \textrm{ iff } \quad
				\exists e \in \set{0, 1} \st (p, 2, 0, 0), \mu' \reach w {} (q, 2, e, e), \nu'_e.
	\end{align*}
\end{restatable}

\paragraph{Reconstruction of the reachability relation \protect{[B]}}

We reconstruct the reachability relation of $\P$ from that of $\P_B$
by applying \Cref{lem:correctness:B}.
Let $\varphi_{pqe}(\bar x, \bar f, \bar y)$ express the reachability relation
${(p, 2, 0, 0), \_ \reach {} {} (q, 2, e, e), \_}$ of $\P_B$ (we write $\bar y$ instead of $\bar x'$ for readability).
For simplicity, we index the variables in $\bar f$ as $f_\delta$, where $\delta \in \Delta_B$.
For a transition $\tilde\delta \in \Delta$,
let $f_{\delta, 1}, \dots, f_{\delta, m}$ be all the $f_\delta$'s \st $\delta = \tilde\delta$.
(Formally, the index $m$ depends on $\tilde\delta$, but for simplicity we omit this dependence.)
Then the reachability relation $p, \_ \reach {} {} q, \_$ of $\P$ can be expressed as
\begin{align*}
	\varphi_{pq}(\bar x, \bar g, \bar y) \;\equiv\;
		\exists \bar f \st \bigvee_{e \in \set{0, 1}} \varphi_{pqe}(\bar x, x_i^0, x_i^1, \bar f, \bar y, y_i^0, y_i^1) \wedge x_i^0 = x_i \wedge y_i^e = y_i \\
		\wedge\bigwedge_{\tilde\delta \in \Delta} g_{\tilde\delta} = f_{\delta, 1} + \cdots + f_{\delta, m}.
\end{align*}

\subsection{Clocks are reset at least once}
\label{sec:simplify:reset}

In this section we show that all clocks can be assumed to be reset at least once during the run.
In fact, this follows from a more powerful observation showing that computing the reachability relation of timed automata
reduces to computing the reachability set \cite[proof of Theorem 1]{FranzleQuaasShirmohammadiWorrell:2019}.
This is achieved by what is called \emph{clock memorisation} in \cite{FranzleQuaasShirmohammadiWorrell:2019}.
We observe that clock memorisation can be applied as-is to \TPDA.
\begin{lemma}
	\label{lem:clock:memorisation}
	Computing the \TPDA reachability relation reduces to computing the \TPDA reachability set.
\end{lemma}

\noindent
One may wonder why we do not apply the lemma above at the outset
and compute the allegedly simpler \TPDA reachability set.
The reason is that, as it will become clear in Sec.~\ref{sec:fractional:TPDA},
for \TPDA the reachability relation is more fundamental and actually easier to compute that the reachability set.
This is due to the transitivity \eqref{eq:reachrel:E} and push-pop \eqref{eq:reachrel:F} rules in the characterisation of the reachability relation from \Cref{lem:characterisation},
which have no counterpart for timed automata.

Nonetheless, clock memorisation allows us to assume that clocks are reset at least once in the run,
which makes some formal constructions in the following sections easier to present.
For completeness and given its wide applicability, we present below the clock memorisation technique of \cite{FranzleQuaasShirmohammadiWorrell:2019}.

\begin{proof}[Proof (of \Cref{lem:clock:memorisation})]
	Let $\P$ be the original \TPDA.
	Fix an initial control location $p \in \L$
	and assume that $\x_0$ is a distinguished reference clock which is $0$ at the beginning of the run.
	The idea is to add a copy $\y_i$ of every control clock $\x_i$ of $\P$.
	Then the execution of $\P$ is preceded by a preprocessing phase starting at a new control location $\tilde p$,
	where arbitrary time elapses $\elapse$ alternate with resets of the form $\resetop {\set{\x_i, \y_i}}$.
	In this way, at the end of preprocessing $\x_i = \y_i$ can be arbitrary,
	and moreover each clock of the new automaton is reset at least once.
	%the time elapses are chosen in such a way as to reach any initial clock valuation $\mu \in \Rgeq^{\X'}$
	%
	After preprocessing, the automaton nondeterministically starts simulating an execution of $\P$ from $p$,
	with the proviso that $\y_i$ is not reset anymore. % (while the original $\x_i$ is reset according to $\P$).
	In this way, $\y_i - \x_0$ is constant during the simulation of $\P$
	and equal to the initial value of $\x_i$.
	Formally, if $\psi_{\tilde pq}(\bar f, \bar x', \bar y')$ is a linear arithmetic formula expressing the reachability set $\ReachSet {\tilde pq}$ of the newly constructed automaton,
	then the following linear arithmetic formula expresses the reachability relation $\reach {} {pq}$ of $\P$:
	\begin{align*}
		\varphi_{pq}(\bar x, \bar f, \bar x') \equiv
			\psi_{\tilde pq}(\bar f, \bar x', \x_1 + \x_0', \dots, \x_n + \x_0'). \tag* \qedhere
	\end{align*}
\end{proof}

\subsection{Fractional \TPDA}
\label{sec:simplify:fractional}

A \TPDA is \emph{fractional} if it has only fractional constraints.
We show in this section how to eliminate all non-fractional constraints.
Thanks to the previous
\Cref{sec:simplify:push-copy,sec:simplify:stack-stack,sec:simplify:pop-integer-free,sec:simplify:atomic-pop,sec:simplify:pop-integer-free},
we assume that the \TPDA $\P$ is push-copy, pop-integer-free, there are no stack-stack pop constraints, and that pop constraints are atomic.
Diagonal control-control integral $\floor {\x_i} - \floor {\x_j} \sim k$
and modulo $\floor {\x_i} - \floor {\x_j} \eqv m k$ constraints %on control clocks $\x, \y$
are removed by a standard construction \cite{BerardDiekertGastinPetit:1998:Epsilon},
incurring a multiplicative blow-up in the number of control locations
exponential in the number of diagonal constraints.
Consequently, transition and pop constraints are (possibly negated) atomic constraints of the form
\begin{align}
	\label{eq:transition:constraints:C}
	&\textrm{(transition)}	&\floor {\x_i} &\leq k, &\floor {\x_i} &\eqv M k, &\fract {\x_i} &= 0, &\fract {\x_i} &\leq \fract {\x_j}, \\%[1ex]
	\label{eq:stack:constraints:C}
	&\textrm{(pop)}		& & &\floor {\y_i} - \floor {x_j} &\eqv M k, &\fract {\y_i} &= 0, &\fract {\y_i} &\leq \fract {\x_j}.
\end{align}
where we assume w.l.o.g.~that $M$ is the maximal constant appearing in any constraint.
The idea is to replace the integral value of clocks by their \emph{unary abstraction} $\lambda$,
which remembers only the modulo class of each clock and the exact value up to $M$.
Fractional constraints are unchanged.
To reconstruct the reachability relation,
the new automaton additionally outputs special symbols for each integral time unit that elapses
for those clocks which are not reset anymore until the end of the run.

\paragraph{Preliminaries}
Let $M \in \N$.
Valuations $\mu, \nu \in \Qgeq^\X$ are \emph{$M$-unary equivalent}
%written $\mu \approx_M \nu$,
if, for every clock $\x \in \X$,
$\floor {\mu(\x)} \eqv M \floor{\nu(\x)}$ and $\floor{\mu(\x)} < M \iff \floor{\nu(\x)} < M$.
Let $\Lambda_M$ be the (finite) set of $M$-unary equivalence classes of clock valuations.
For $\lambda \in \Lambda_M$ and a clock $\x$,
we write $\lambda(\x)$ for $\mu(\x)$, where $\mu$ is any clock valuation in $\lambda$, provided its 
choice does not matter.
We write $\lambda[\Y \mapsto 0]$ for the equivalence class of $\nu[\Y \mapsto 0]$ for some $\nu\in\lambda$,
and, for a clock $\x \in \X$,
$\lambda[\x \mapsto \x + 1]$ for the equivalence class of $\nu[\x \mapsto \nu(\x)+1]$ for some $\nu\in\lambda$;
in both cases, the choice of $\nu$ is irrelevant.
Let $\varphi_\lambda$ be the \emph{characteristic constraint} of the unary class $\lambda$:
\begin{align}
 \varphi_\lambda(\vec \x) \equiv \bigwedge_{\x \in \X} \floor \x \eqv M \lambda(\x) \wedge (\floor \x < M \iff \lambda(\x) < M),
\end{align}
where $\lambda(\x)$ denotes $\nu(\x)$ for some $\nu\in\lambda$
(whose choice is irrelevant).
%say that clocks belong to $\lambda$.
%
For a control constraint $\varphi$ (of the form \eqref{eq:transition:constraints:C}),
let $\restrict \varphi \lambda$ be $\varphi$
where every non-diagonal integer $\floor \x \leq k$ or modulo constraint $\floor \x \eqv M k$
is uniquely resolved to be $\true$ or $\false$ by looking at $\lambda$;
thus, $\restrict \varphi \lambda$ contains only fractional constraints.
The following observation formalises that it suffices to know the unary class of $\mu$ and its fractional part in order to know whether it satisfies $\varphi$.

\begin{fact}
	\label{fact:unary}
	For every valuation $\mu \in \lambda$ %\Rgeq^\X$ 
	and constraint $\varphi$,
	$\mu \models \varphi$ iff $\fract \mu \models \restrict \varphi {\lambda}$.
\end{fact}

We now explain how to remove modular constraints from a pop constraint $\psi$. %(of the form \eqref{eq:stack:constraints:C}).
Let $\mu \in \Rgeq^\X$ be the control clock valuation at the time of push,
let $\nu \in \Rgeq^\X$ be the control clock valuation at the time of pop,
and let $\rho \in \Rgeq^\ZZ$ be the stack clock valuation at the time of pop. 
Recall that $\x_0$ is a control clock which is never reset
and $\y_1$ a stack clock which is $0$ upon push.
Thanks to push-copy,
\begin{align*}
	\rho(\y_1) = \nu(\x_0) - \mu(\x_0)
		\quad \textrm{and} \quad
			\rho(\y_i) = \mu(\x_i) + \rho(\y_1) = \mu(\x_i) + \nu(\x_0) - \mu(\x_0).
\end{align*}
A modulo constraint $\floor {\y_i} - \floor {\x_j} \eqv M k$
is satisfied at the time of pop if $\floor {\rho(\y_i)} - \floor {\nu(\x_j)} \eqv M k$.
It suffices to compute the modulo class of $\floor {\nu(\x_j)}$ (which we will store in the finite control) and of $\floor {\rho(\y_i)}$.
For the latter, we make use of the following observation.

\begin{fact}
	\label{fact:fractional}
	$\floor {\rho(\y_i)} =
			\floor {\mu(\x_i)} + \floor {\nu(\x_0)} -  \floor {\mu(\x_0)} +
				\condone {\fract{\rho(\y_i)} < \fract {\rho(\y_1)}} -
					\condone {\fract {\nu(\x_0)} < \fract {\rho(\y_1)}}$.
\end{fact}
\begin{proof}
	We compute $\floor {\rho(\y_i)}$ and $\floor {\rho(\y_1)}$ directly:
	\begin{align*}
		\floor {\rho(\y_i)}
			&= \floor {\mu(\x_i) + \rho(\y_1)} =
				&& \textrm{(def.~$\rho$)} \\
			&= \floor {\mu(\x_i)} + \floor {\rho(\y_1)} + \condone {\fract {\mu(\x_i)} + \fract {\rho(\y_1)} \geq 1} = 
				&& \textrm{(by \eqref{eq:floor})} \\
			&= \floor {\mu(\x_i)} + \floor {\rho(\y_1)} + \condone {\fract {\rho(\y_i) - \rho(\y_1)} + \fract {\rho(\y_1)} \geq 1} = 
				&& \textrm{($\mu(\x_i) = \rho(\y_i) - \rho(\y_1)$)} \\
			&= \floor {\mu(\x_i)} + \floor {\rho(\y_1)} + \\
			&\ \ + \condone {\left(\fract {\rho(\y_i)} - \fract {\rho(\y_1)} + \condone {\fract {\rho(\y_i)} < \fract {\rho(\y_1)}}\right) + \fract {\rho(\y_1)} \geq 1} = 
				&& \textrm{(by \eqref{eq:fract})} \\
			&= \floor {\mu(\x_i)} + \floor {\rho(\y_1)} + \condone {\fract {\rho(\y_i)} + \condone {\fract {\rho(\y_i)} < \fract {\rho(\y_1)}} \geq 1} = 
				&& \textrm{(simpl.)} \\
			&= \floor {\mu(\x_i)} + \floor {\rho(\y_1)} + \condone {\fract {\rho(\y_i)} < \fract {\rho(\y_1)}}.
				&& \textrm{(def.~$\condone {}$)}, \\[2ex]
		\floor {\rho(\y_1)}
			&= \floor {\nu(\x_0) - \mu(\x_0)} =
				&& \textrm{(def.~$\rho$)} \\
			&= \floor {\nu(\x_0)} - \floor {\mu(\x_0)} - \condone {\fract {\nu(\x_0)} < \fract {\mu(\x_0)}}
				&& \textrm{(by \eqref{eq:floor})} \\
			&= \floor {\nu(\x_0)} - \floor {\mu(\x_0)} - \condone {\fract {\nu(\x_0)} < \fract {\nu(\x_0) - \rho(\y_1)}}
				&& \textrm{($\mu(\x_0) = \nu(\x_0) - \rho(\y_1)$)} \\
			&= \floor {\nu(\x_0)} - \floor {\mu(\x_0)} + \\
			&\ \ - \condone {\fract {\nu(\x_0)} < \fract {\nu(\x_0)} - \fract {\rho(\y_1)} + \condone {\fract {\nu(\x_0)} < \fract {\rho(\y_1)}}}
				&& \textrm{(by \eqref{eq:fract})} \\
			&= \floor {\nu(\x_0)} - \floor {\mu(\x_0)} - \condone {\fract {\rho(\y_1)} < \condone {\fract {\nu(\x_0)} < \fract {\rho(\y_1)}}}
				&& \textrm{(simpl.)} \\
			&= \floor {\nu(\x_0)} - \floor {\mu(\x_0)} - \condone {\fract {\nu(\x_0)} < \fract {\rho(\y_1)}}
				&& \textrm{(def.~$\condone {}$)}.
	\end{align*}
	%
	%Equation \eqref{eq:fractional:fact}
	The claim follows from the two equations above.
\end{proof}
\noindent
Thanks to \Cref{fact:fractional}, we reconstruct the modulo class of $\floor {\rho(\y_i)}$
from those of $\floor {\mu(\x_i)}$, $\floor {\nu(\x_0)}$, and $\floor {\mu(\x_0)}$, which will all be stored in the control,
provided we can determine the value of the correction terms
$\condone {\fract{\rho(\y_i)} < \fract {\rho(\y_1)}}$ and $\condone {\fract {\nu(\x_0)} < \fract {\rho(\y_1)}}$.
The latter is easily achieved by looking at the fractional values at the time of pop of the stack clocks $\y_i, \y_1$ and of the control clock $\x_0$.

Formally, for a pop constraint $\psi$,
a unary abstraction $\lambda_\push$ at the time of push
and one $\lambda_\pop$ at the time of pop,
let $\restrict \psi {\lambda_\push, \lambda_\pop}$ be $\psi$
where every modulo constraint $\floor {\y_i} - \floor {\x_j} \eqv M k$
is resolved to be $\true$ or $\false$
by replacing $\floor {x_j}$ by $\lambda_\pop (\x_j)$,
and $\floor {\y_i}$ ($i \geq 0$) by
\begin{align} \label{eq:floory}
	\underbrace{\;\lambda_\push (\x_i)\;}_{\text{initial value of $\x_i$}} + \quad
		\underbrace{\;\lambda_\pop (\x_0) - \lambda_\push (\x_0)
			+ \condone {\fract {\y_i} < \fract {\y_1}}
				- \condone {\fract {\x_0} < \fract {\y_1}}\;}
					_{\text{time between push and pop}}.
\end{align}
%
%and $\floor{\y_1}$ by
%
%\begin{align} \label{eq:floory0}
%		\underbrace{\;\lambda_\pop (\x_0) - \lambda_\push (\x_0)
%				- \condone {\fract {\x_0} < \fract {\y_1}}\;}
%					_{\text{time between push and pop}}.
%\end{align}
%
While the notation above is not a constraint,
it can be converted to a constraint by expanding the definition of the correction terms of the form $\condone {}$.
%
%Intuitively, $\lambda_\push (\x_i)$ represents the modulo class of the initial integral value $\floor {\x_i}$,
%and the second expression represents the modulo class of the integral time that elapsed between push and pop $\floor {}$.
%
%This is formally justified by the following lemma.

\paragraph{The construction}
Let $\P = (\Sigma, \Gamma, L, \X, \ZZ, \Delta)$ be a \TPDA
satisfying the assumptions from the beginning of this section.
We build a fractional \TPDA $\QQ = (\Sigma', \Gamma', L', \X, \ZZ, \Delta')$
in such a way that we can express the reachability relation of $\P$ in terms of the one of $\QQ$.
The new input alphabet is $\Sigma'= \Sigma \cup T$,
where $$T = \setof {\tick \x} {\x \in \X}$$
is a set of tick symbols $\tick \x$ for every control clock $\x \in \X$;
those are used in order to reconstruct the integral value of clocks in the reachability relation of $\P$.
The new stack alphabet $\Gamma' = \Gamma \times \Lambda_M$ extends $\Gamma$
by recording the $M$-unary equivalence class of clocks which are pushed on the stack.
Control locations $L'$ of $\QQ$ are either new intermediate control locations used to perform the simulation
(which we do not describe explicitly for simplicity),
or of the form $\tuple {p, \lambda, \T}$
where $p \in L$ is a control location of $\P$,
$\lambda \in \Lambda_M$ abstracts the integral value of clocks, and
$\T \in 2^\X$ is the set of clocks which are \emph{not} allowed to be reset anymore in the future;
the last component is used in order to emit ticks $\tick \x$ when one unit of time elapses for clocks which are not reset anymore.

Every transition $\trule p \op q \in \Delta$ generates one or more transitions in $\QQ$
according to $\op$.
If ${\op = \readop a}$ is an input transition,
then $\Delta'$ contains a corresponding transition
$\trule {\tuple {p, \lambda, \T}} {\readop a} {\tuple{q, \lambda, \T}} \in \Delta'$,
for every choice of $\lambda$ and $\T$;
for conciseness, we will henceforth implicitly assume that we take all possible choices of free parameters (such as $\lambda$ and $\T$ above) for transitions in $\Delta'$.
If $\op = \testop \varphi$ is a test transition,
then $\QQ$ contains a corresponding test transition
$$\trule {\tuple {p, \lambda, \T}} {\testop {\restrict \varphi \lambda}} {\tuple{q, \lambda, \T}} \in \Delta',$$
where $\restrict \varphi \lambda$ contains only fractional constraints.
If $\op = \resetop \Y$ is a reset transition,
then $\QQ$ contains a reset transition of the form
\begin{align}
	\label{eq:fractional:reset}
	\trule {\tuple {p, \lambda, \T}} {\resetop \Y} {\tuple{q, \lambda[\Y \mapsto 0], \T \cup \Y'}} \in \Delta'
\end{align}
whenever $\Y \subseteq \X \setminus \T$, i.e., no forbidden clocks are reset,
and $\Y' \subseteq \Y$ are new clocks which are declared to be reset now for the last time.
If $\op = \elapse$ is a time elapse transition,
then we need to update the unary abstraction of control clocks
and also emit tick symbols $\tick \x$'s for the integral elapse of time of clocks in $\T$.
We simulate only time elapses of length $<1$; longer elapses can be obtained by repeating many small elapses.
This is achieved with the the following three groups of transitions in $\QQ$:
\begin{enumerate}[a)]
%\noindent
%\textbf{(1)}
\item
	First, we silently go to $\tuple {p, \lambda, \T, q}$ to start the simulation:
	\begin{align*}
	\langle \tuple {p, \lambda, \T},
		\readop \varepsilon, \tuple {p, \lambda, \T, q}\rangle \in \Delta'.
	\end{align*}

%\noindent
%\textbf{(2)}
\item
	The following formula says that clocks $\Y \subseteq \X$ have maximal fractional value
	(and thus will overflow first when time elapses):
	\begin{align*}
		%\varphi^{>0}_\Y(\bar x) \; \equiv \; \bigwedge_{\x \in \Y} \fract \x > 0 \wedge \bigwedge_{\x \in \X \setminus \Y} \fract \x = 0,
		\varphi^{\max}_\Y (\bar x) \; \equiv \; \bigwedge_{\x_i \in \Y} \bigwedge_{\x_j \in \X} \fract {x_j} \leq \fract {\x_i}.
	\end{align*}
	The automaton guesses such a set of clocks $\Y$,
	checks that their fractional value is $0$ after time elapse,
	reads corresponding ticks,
	and updates the unary abstraction accordingly:
	$\langle \tuple {p, \lambda, \T, q}, \ops_{\Y,\T}, \tuple {p, \lambda[\Y \mapsto \Y + 1], \T, q} \rangle \in \Delta'$,
	where the overflown clocks which will not be reset in the future are $\Y \cap \T = \set{\x_{i_1}, \dots, \x_{i_m}}$ and
	\begin{align}
		\label{eq:fractional:ops}
		\ops_{\Y,\T} :=
			\testop {\varphi^{\max}_\Y};
				\elapse;
					\testop{\bigwedge_{\x_i \in \Y} \fract {\x_i} = 0};
						\readop{\tick {i_1} \cdots \tick {i_m}}.
	\end{align}

\item
	When enough time has elapsed, we quit the simulation
	($\xi$ is in general different from the starting $\lambda$):
	$$\trule {\tuple {p, \xi, \T, q}} {\readop{\varepsilon}} {\tuple {q, \xi, \T}} \in \Delta'.$$
	This concludes the simulation of time elapse.
\end{enumerate}
%
%\noindent
If $\op = \pushop{\gamma}{\psicopy}$ is a push-copy transition,
then $\QQ$ contains a fractional push-copy transition,
additionally recording the current unary class in the stack:
\begin{align}
	\nonumber
	&\trule
	{\tuple {p, \lambda, \T}}
	{\pushop{\tuple {\gamma, \lambda}} {\psi_\push}}
	{\tuple{q, \lambda, \T}} \in \Delta', \textrm { where } \\[1ex]
	\label{eq:fractional:psipush}
	&\psi_\push \equiv \fract {\y_1} = 0 \wedge \bigwedge_{\x_i \in \X} \fract {\y_i} = \fract {\x_i}.
\end{align}
Finally, if $\op = \popop{\gamma}{\psi}$ is a pop transition without integral constraints,
then $\QQ$ contains fractional pop transitions of the form
$$
	\trule
	{\tuple {p, \lambda_\pop, \T}}
	{\popop{\tuple {\gamma, \lambda_\push}}{\restrict \psi {\lambda_\push,\lambda_\pop}}}
	{\tuple{q, \lambda_\pop, \T}} \in \Delta'.
$$
This concludes the description of $\QQ$.
We eliminated all occurrences of $\floor \x$ both from transition and push/pop stack constraints.
Thus, all transition and stack constraints of $\QQ$ are fractional.

\begin{example} \label{ex:integralpop}
Continuing with the \TPDA from Example~\ref{ex:TPDA}, consider the pair of push and pop transitions derived in Example~\ref{ex:pushcopy} ($\alpha = \tuplesmall{\tilde\delta_\push, \tilde\delta_\pop}$):
\begin{align*}
\delta_\push = \trule {q'_1} {\pushop{\alpha}{\y_1 = 0}} {q_1}, \ 
\delta_\pop = \trule {q'_2} {\popop{\alpha}{\floor {\y_1} \eqv 2 1 \land \fract{\y_1} \leq \fract{\x_0}}} {q_2}.
\end{align*}
Our construction, in order to eliminate the modular pop constraint $\floor{\y_1} \eqv 2 1$,
enriches the locations with $2$-unary abstraction $\lambda$ (we ignore 
here for simplicity the further component $T$). The $2$-abstraction 
amounts to the remainder of $\floor{\x_0}$ modulo 2, thus $\lambda \in \{0, 1\}$ and 
the transformation yields two push transitions:
\begin{align*}
\trule {\langle q'_1, \lambda \rangle} {\pushop{\langle \alpha, \lambda \rangle}{\fract{\y_1} = 0}} {\langle q_1, \lambda \rangle} 
\qquad (\lambda \in \{0,1\}).
\end{align*}
To derive the corresponding pop transitions, we use the formula~\eqref{eq:floory} and substitute
\begin{align} \label{eq:x0mod2}
\lambda_\pop - \lambda_\push - \condone {\fract {\x_0} < \fract {\y_1}} \eqv 2 1
\end{align}
in place of the pop constraint $\floor{\x_0} \eqv 2 1$,
where $\lambda_\pop$ and $\lambda_\push$ are the reminders of $\floor{\x_0}$ at the time of pop and push, respectively. In presence of the other pop constraint $\fract {\y_1} \leq \fract {\x_0}$, the formula~\eqref{eq:x0mod2} 
is equivalently expressed by $\lambda_\pop \neq \lambda_\push$. The transformation thus yields altogether two pop transitions:
\begin{align*}
&\trule {\langle q'_2, 0 \rangle} {\popop{\langle \alpha, 1 \rangle}{\fract{\y_1} \leq \fract{\x_0}}} {\langle q_2, 0 \rangle} \\
&\trule {\langle q'_2, 1 \rangle} {\popop{\langle \alpha, 0 \rangle}{\fract{\y_1} \leq \fract{\x_0}}} {\langle q_2, 1 \rangle}.
%\qquad (\lambda_\pop \neq \lambda_\push)
\end{align*}
\end{example}
The following two lemmas state the correctness of the construction.
\begin{restatable}[Soundness]{lemma}{lemFractionalSoundness}
	\label{lem:fractional:soundness}
	For control locations $p, q \in \L$,
	clock valuations ${\mu, \nu \in \Rgeq^\X}$,
	and a sequence of operations $w \in (\Delta')^*$,
	let $\nu' \in \Rgeq^\X$ be the unique clock valuation \st 
	$\forall \x_i \in \X \st \nu'(\x_i) = \fract {\nu(\x_i)} + \card {w}_{\checkmark_i}$.
	Then
	\begin{align*}
		\tuple {p, \lambda(\mu), \emptyset}, \mu \reach w {} \tuple {q, \lambda(\nu), \X}, \nu
			\quad \textrm{ implies } \quad
			 p, \mu \reach {} {} q, \nu'.
	\end{align*}
\end{restatable}
\begin{restatable}[Completeness]{lemma}{lemFractionalCompleteness}
	\label{lem:fractional:completeness}
	If $p, \mu \reach \pi {} q, \nu$ and all clocks are reset in $\pi$,
	then there is $w \in (\Delta')^*$ \st 
	\begin{align*}
		\tuple {p, \lambda(\mu), \emptyset}, \mu \reach w {} \tuple {q, \lambda(\nu), \X}, \nu
			\textrm{ and }
				\forall \x_i \in \X \st \floor{\nu(\x_i)} = \card w_{\checkmark_i}.
	\end{align*}
\end{restatable}

%\paragraph{Correctness of the construction}

%The following lemma formalises the fact that $\QQ$ looks only at the fractional value of clocks.
%\begin{lemma}
%	For every control locations $c, d \in \L'$ of $\QQ$, clock valuations ${\mu, \mu', \nu \in \Rgeq^\X}$, and operations $w \in (\Delta')^*$,
	%
%	\begin{align*}
%		c, \mu \reach w {} d, \nu \textrm { and } \fract{\mu'} = \fract \mu
%			\ \textrm { implies } \ 
%				\exists \nu' \in \Rgeq^\X \st c, \mu' \reach w {} d, \nu' \textrm{ and } \fract{\nu'} = \fract \nu.
%	\end{align*}
%\end{lemma}

\paragraph{Reconstruction of the reachability relation}

The correctness provided by \Cref{lem:fractional:soundness,lem:fractional:completeness} above
allows us to reconstruct the linear arithmetic description of the reachability relation $\reach {} {}$ of the \TPDA $\P$
by looking at reachability relation of the fractional \TPDA $\QQ$.
In fact, we never need to look at the integral values of clocks in the reachability relation of $\QQ$,
since those are reconstructed on the sole basis of the number of ticks $\tick i$ read by $\QQ$.
For this reason, it suffices to know the \emph{fractional reachability relation} $\freach {} {cd}$ of $\QQ$.
We assume that the formula expressing $\freach {} {cd}$ does not contain occurrences of integral clock variables $\floor {\x_i}, \floor {\x_i'}$,
i.e., it is of the form $\varphi_{cd}(\fract{\bar x}, \bar f, \fract{\bar x'})$,
where $\bar f = (f_1, \dots, f_{\card {\Delta'}})$ counts the number of occurrences of transitions in $\Delta'$ of $\QQ$.

With these ingredients we can express the reachability relation of $\P$ in terms of the fractional reachability relation of the fractional \TPDA $\QQ$.
Recall that transitions of $\P$ are from $\Delta = \set{\delta_1, \dots, \delta_m}$.
%
%Recall that the set of clocks (in both automata) is $\X = \set{\x_1, \dots, \x_n}$,
%
Transitions in $\QQ$ can be organised in three groups:
\begin{itemize}

	\item In the first group we have variables $f_1, \dots, f_n$ counting transitions $\readop {\tick i}$'s,
	which will be used to reconstruct the integral values of clocks of $\P$.

	\item In the second group we have variables $f_{n+1}, \dots, f_{n+m}$
	counting transitions $\delta_i$'s originally from $\Delta$,
	which we preserve.

	\item Finally, in the last group we have variables $f_{n+m+1}, \dots, f_{\card{\Delta'}}$
	counting the remaining (new) transitions from $\Delta' \setminus \Delta$,
	which we project away.

\end{itemize}
The last point is implemented by defining
$$\tilde \varphi_{cd} (\fract{\bar x}, f_1, \dots, f_{n+m}, \fract{\bar x'}) \equiv \exists f_{n+m+1}, \dots, f_{\card{\Delta'}} \st \varphi_{cd}(\fract{\bar x}, \bar f, \fract{\bar x'}).$$
The first two points allow us to express the reachability relation of $\P$ %between control locations $p, q \in \L$
as
\begin{align*}
	&\varphi_{pq}(\fract{\bar x}, f_{n+1}, \dots, f_{n+m}, \fract{\bar x'}) \;\equiv\;
		\bigvee_{\lambda, \xi \in \Lambda_M} \varphi_{\lambda}(\floor{\bar x}) \wedge \varphi_{\xi}(\floor{\bar x'}) \wedge \\
	&		\!\!\wedge \exists f_1, \dots, f_n \st \tilde \varphi_{\tuple{p,\lambda, \emptyset}, \tuple{q, \xi, \X}}(\fract{\bar x}, f_1, \dots, f_{n+m} , \fract{\bar x'}) \wedge
				\floor{x_1} = f_1 \wedge \cdots \wedge \floor{x_n} = f_n.
\end{align*}
% !TEX root = main.tex

\section{Fractional reachability relations of fractional \TPDA}
\label{sec:fractional:TPDA}

In this section we compute the fractional reachability relation (defined in \eqref{eq:fractional:reachability}) for a fractional \TPDA~\(\P\).
%introduced at the very end of \Cref{sec:simplify:fractional},
%
We assume that there is a control clock $\x_0$ which is never reset.
%and a distinguished stack clock $\z_0$ which equals $\x_0$ upon push and pop%
%---i.e., push and pop constraints are always of the form $\z_0 = \x_0 \wedge \varphi$.

We begin by showing how to express the one-step fractional transition relation of a fractional \TPDA as a \CDR.
Recall that $\varphi_{\id}$ from \eqref{eq:CDR:id} implements the identity relation.
The one-step transition relations of a fractional \TPDA is expressible as the following \CDR:
\begin{align}
  \label{eq:CDR:read}
  \varphi_{\readop a} (\bar x, \bar x') &\;\equiv\; \varphi_{\id}(\bar x, \bar x'), \\
  \varphi_{\testop \psi} (\bar x, \bar x') &\;\equiv\; \varphi_{\id}(\bar x, \bar x') \wedge \psi(\bar x), \\
  \varphi_{\resetop \Y}(\bar x, \bar x') &\;\equiv\;
    \bigwedge_{\x \in \Y} \fract {\x_0' - \x'} = \fract {\x_0'} \wedge
    \bigwedge_{\x \in \X \setminus \Y} \fract{\x_0' - \x'} = \fract {\x_0 - \x} \wedge \\
    &\; \wedge \fract{\x_0'} = \fract{\x_0}, \\ %\textrm{ and } \\
  \label{eq:CDR:elapse}
  \varphi_{\elapse}(\bar x, \bar x') &\;\equiv\;
    \bigwedge_{\x \in \X} \fract{\x'_0 - \x'} = \fract{\x_0 - \x}.
\end{align}
We assume that a clock constraint $\psi$ is in the \CDR form thanks to the equivalences \eqref{eq:CDR:zero}--\eqref{eq:CDR:leq}.
The following lemma states that the basic \CDR above
capture the one-step fractional reachability relation.
\begin{fact}
  \label{fact:onestep:CDR}
  For all valuations $\mu, \nu \in \Rgeq^\X$ and a transition $\delta = \tuple{p, \op, q} \in \Delta$ of the form $\op = \readop a, \testop \psi, \resetop \Y, \elapse$,
  \begin{align*}
    (\mu, \nu) \models \varphi_\op
      \quad \textrm{iff} \quad
        \mu \freach \delta {pq} \nu.
  \end{align*}
\end{fact}

We construct a context-free grammar $\GG$ with terminal symbols from $\Delta$
and a nonterminal of the form $\tuple{p, \varphi, q}$
for every control locations $p, q \in \L$ and a \CDR $\varphi$.
For every transition  $\delta = \trule p \op q \in \Delta$ of the form
$\op = {\readop a, \testop \psi, \resetop \Y, \elapse}$, % is not a stack operation,
we have a production
\begin{align}
  \label{eq:CFG:base}
  \tuple {p, \varphi_\op, q} \from \delta,
\end{align}
where the basic clock relations $\varphi_{\readop a}, \varphi_{\testop \psi}, \varphi_{\resetop \Y}, \varphi_{\elapse}$
are defined in \eqref{eq:CDR:read}--\eqref{eq:CDR:elapse} above.
These rules mimic cases \eqref{eq:reachrel:A}--\eqref{eq:reachrel:D}
in the characterisation of the reachability relation of Lemma~\ref{lem:characterisation}.
Transitivity rules \eqref{eq:reachrel:E} are mimicked by productions of the form
\begin{align}
  \label{eq:CFG:transitivity}
  \tuple {p, \varphi \circ \psi, r} \from \tuple {p, \varphi, q} \cdot \tuple {q , \psi, r}.
\end{align}
Push-pop rules \eqref{eq:reachrel:F} are simulated as follows:
For every pair of matching push
$\delta_\push = \trule p {\pushop \alpha {\psi_\push}} r$
and pop
$\delta_\pop = \trule s {\popop \alpha {\psi_\pop}} q$ transitions,
we have a production
\begin{align}
  \label{eq:CFG:push-pop}
  &\tuple {p, \varphi, q} \ \from\ \delta_\push \cdot \tuple {r, \psi, s} \cdot \delta_\pop, \textrm{ where } \\ 
%  \label{eq:CFG:varphi}
  \nonumber
  &\varphi(\bar x, \bar x') \equiv
    \exists \bar z, \bar z' \!\st\! \psi(\bar x, \bar x') \wedge
      \psi_\push(\bar x, \bar z) \wedge
        \psi_\pop(\bar x', \bar z') \wedge
          \!\underbrace{\bigwedge_i \fract{\z_i'} = \fract{\z_i + x_0' - x_0}}_{\textrm{(A)}}.
\end{align}
The part (A) above ensures that the final fractional value $\fract{\z_i'}$ of stack clocks
is obtained from its initial value by elapsing the same amount of time $x_0' - x_0$.
While (A) is syntactically not a \CDR,
it is in fact equivalent to the \CDR
$\bigwedge_i \fract {x_0 - z_i} = \fract {x_0' - z_i'}$%
\footnote{This follows from the identity
$\forall a, b, c \in \R \st \fract a = \fract b \iff \fract{c-a} = \fract{c-b}$.}.
%the latter holds since $\fract{c-a}$ depends only on $\fract{c}$ and $\fract{a}$ by \eqref{eq:fact:minus},
%and similarly for $\fract {c-b}$.}.
%
The quantified variables $\bar z, \bar z'$ can be eliminated by Lemma~\ref{lem:CDR:qe},
since they do not involve the reference variables $x_0, x_0'$.
Formally, we assume that $\varphi$ is presented as an equivalent \CDR,
uniquely determined by $\delta_\push$, $\delta_\pop$, and $\psi$.
%This expression above will make the correctness proof easier.

%To formally allow the substitution between equivalent formulas, we add the following \emph{implication rule}:
%
%\begin{align}
%  \label{eq:CFG:implies}
%  \tuple {p, \varphi, q} \from \tuple {p, \psi, q}, \textrm{ whenever } \models \varphi \to \psi.
%\end{align}
%
%With the implication rule in the set of productions, we can assume that in \eqref{eq:CFG:push-pop}
%$\varphi$ is replaced by any \CDR $\tilde\varphi$ logically equivalent to it.

\begin{example}
For illustration of the latter rules~\eqref{eq:CFG:push-pop},
consider one of matching pairs of push and pop transitions from Example~\ref{ex:integralpop}:
\begin{align*}
\delta_\push & = \trule {\langle q'_1, 0 \rangle} {\pushop{\langle \alpha, 0 \rangle}{\fract{\y_1} = 0}} {\langle q_1, 0 \rangle}  \\
\delta_\pop & = \trule {\langle q'_2, 1 \rangle} {\popop{\langle \alpha, 0 \rangle}{\fract{\y_1} \leq \fract{\x_0}}} {\langle q_2, 1 \rangle}.
%\qquad (\lambda_\pop \neq \lambda_\push)
\end{align*}
Instantiating $\psi$ with $\true$, we obtain a rule
\begin{align*}
  &\tuple {\langle q'_1, 0 \rangle, \varphi, \langle q_2, 1 \rangle} \ \from\ \delta_\push \cdot 
    \tuple {\langle q_1, 0 \rangle, \true, \langle q'_2, 1 \rangle} \cdot \delta_\pop
\end{align*}
where $\varphi \equiv \fract{\x_0} \leq \fract{\x'_0}$ is a \CDR equivalent to the following formula: 
\begin{align*}
\exists \y_1, \y_1' \! \st \! \fract{\y_1} = 0 \wedge \fract{\y_1'} \leq \fract{\x'_0} \wedge
\fract{\x_0 - \y_1} = \fract{\x'_0 - \y_1'}.
\end{align*}
\end{example}

The following two lemmas show that $\GG$ correctly encodes the fractional reachability relation of $\P$.
%
%These two claims are a direct consequence of the characterisation of the reachability relation $\reach {} {pq}$
%from Lemma~\ref{lem:characterisation},
%where the basic cases have been established in \Cref{fact:onestep:CDR}.
%
\begin{restatable}[Soundness]{lemma}{lemCFGsoundness}
  \label{lem:CFG:soundess}
%  For every nonterminal $\tuple{p, \varphi, q}$ of $\GG$,
%  clock valuations $\mu, \nu \in \Rgeq^\X$ \st $(\mu, \nu) \models \varphi$,
%  and a sequence of transitions $w \in \Delta^*$,
  %
  %\begin{align*}
   $w \in L(p, \varphi, q) \textrm{ and } (\mu, \nu) \models \varphi \textrm{ implies } \mu \freach w {pq} \nu.$
  %\end{align*}
\end{restatable}

For the completeness proof, it is more convenient to work with the actual reachability relation of $\P$, yielding a stronger statement.
\begin{restatable}[Completeness]{lemma}{lemCFGcompleteness}
  \label{lem:CFG:completeness}
  %For control locations $p, q \in L$ of $\GG$,
  %clock valuations $\mu, \nu \in \Rgeq^\X$,
  %and a word $w \in \Sigma^*$,
  %
  %\begin{align*}
   $\mu \reach w {pq} \nu \textrm{ implies }
    \exists \varphi \st w \in L(p, \varphi, q) \textrm{, } (\mu, \nu) \models \varphi$.
  %\end{align*}
\end{restatable}

\begin{corollary}%[Completeness]
  \label{cor:CFG:completeness}
  $\mu \freach w {pq} \nu \textrm{ implies }
  \exists \varphi \st w \in L(p, \varphi, q) \textrm{, } (\mu, \nu) \models \varphi$.
\end{corollary}

\paragraph{Construction of the fractional reachability relation}

We now show how to build the fractional reachability relation of the fractional \TPDA $\P$
by looking at the family of context-free languages $L(p, \varphi, q)$ generated by the context-free grammar $G$ constructed above.
By Parikh's theorem \cite{Parikh:JACM:1966}, the Parikh image $\PI {L(p, \varphi, q)} {}$ is a semilinear set, and thus it is expressible in Presburger arithmetic \cite{GinsburgSpanier:Semilinear:Presburger:PJM:1966}.
While the naive translation from semilinear sets to Presburger arithmetic is exponential,
in the special case of Parikh images of context-free languages recognised by grammars
we can use the following more efficient direct translation.
\begin{theorem}[\protect{\cite[Theorem 4]{VSS05}}]
  \label{thm:Parikh}
  The Parikh image of $L(p, \varphi, q)$ is expressible by an existential Presburger formula $\psi_{p, \varphi, q}(\bar f)$ (and thus in linear arithmetic)
  computable in linear time in the size of the grammar $G$.
\end{theorem}

\noindent
From the above result,
%Let  be a formula of linear arithmetic (in fact, of existential Presburger arithmetic, which is a fragment thereof)
%expressing the Parikh image of $L(p, \varphi, q)$, the language generated by the nonterminal $\tuple {p, \varphi, q}$.
%
we express the fractional reachability relation of $\P$ as
\begin{align}
  \varphi_{pq}(\bar \x, \bar f, \bar \x') \;\equiv\;
    \bigvee_{\textrm{\CDR }\ \varphi} \varphi(\bar \x, \bar \x') \wedge \psi_{p, \varphi, q}(\bar f).
\end{align}

The following lemma states that $\varphi_{pq}$ above correctly expresses the fractional reachability relation of the fractional \TPDA $\P$.
\begin{lemma}
  For every fractional clock valuations $\mu, \nu : (\R \cap [0, 1))^\X$,
  transition count valuation $\eta : \N^\Delta$,
  and control locations $p, q \in \L$,
  \begin{align*}
    \mu, \eta, \nu \models \varphi_{pq}
      \quad \textrm{if, and only if,} \quad
        \exists w \in \Delta^* \st \PI w {} = \eta \textrm{ and } \mu \freach w {pq} \nu.
  \end{align*}
\end{lemma}

\begin{proof}
  For the ``only if'' direction, assume $\mu, \eta, \nu \models \varphi_{pq}(\bar x, \bar f, \bar x')$.
  There is a \CDR $\varphi$ \st $\mu, \nu \models \varphi(\bar \x, \bar \x')$ and $\eta \models \psi_{p, \varphi, q}(\bar f)$.
  By the definition of $\psi_{p, \varphi, q}$ there is a sequence of transitions $w \in L(p, \varphi, q)$ with Parikh image $\PI w {} = \eta$.
  By \Cref{lem:CFG:soundess}, $\mu \freach w {pq} \nu$, as required.
  For the ``if'' direction, assume $w$ is a sequence of transitions \st $\mu \freach w {pq} \nu$.
  %By definition of fractional reachability, $\tilde \mu \reach w {pq} \tilde \nu$
  %for some clock valuations $\tilde \mu, \tilde \nu \in \Rgeq^\X$ \st $\fract {\tilde\mu} = \fract {\mu}$ and $\fract {\tilde\nu} = \fract {\nu}$.
  By \Cref{cor:CFG:completeness}, there is a \CDR $\varphi$
  \st $w \in L(p, \varphi, q)$, and thus its Parikh image $\PI w {}$ satisfies $\PI w {} \models \psi_{p, \varphi, q}(\bar f)$,
  and $(\mu, \nu) \models \varphi$.
  %
  %Since $\varphi$ only has fractional constraints, $(\mu, \nu) \models \varphi(\bar x, \bar x')$ as well.
  Putting the pieces together, $(\mu, \PI w{}, \nu) \models \varphi(\bar x, \bar x') \wedge \psi_{p, \varphi, q}(\bar f)$, as required.
\end{proof}

\begin{corollary}
  The fractional reachability relation of a fractional \TPDA is expressible in linear arithmetic.
\end{corollary}

%\begin{proof}
%  By Lemmas~\ref{lem:CFG:completeness},
%  the $pq$-reachability relation of $\P$
%  is included in the sum of all atoms of the form $(L(p, \varphi, q), \varphi)$   for a \CDR $\varphi$ ,
%  and by Lemma~\ref{lem:CFG:soundess} it in fact equals this sum.
%\end{proof}

% !TEX root = main.tex                                                                                                                                                        

\section{Complexity}  \label{sec:complexity}
\label{sec:complexity}

In this section we comment on the complexity of our procedure when applied to \TPDA and certain subclasses thereof.
By combining the blow-up of the different constructions in \Cref{sec:simplify,sec:fractional:TPDA},
we can express the reachability relation of a \TPDA with an existential formula of linear arithmetic of doubly exponential size.
If we assume that we start from a push-copy \TPDA, we can avoid the singly exponential blow-up of \Cref{sec:simplify:push-copy},
and obtain a formula of singly exponential size.
We obtain the following refinement of our main result \Cref{thm:TPDA}:

\begin{theorem}
	\label{thm:TPDA:complexity}
	The reachability relation of a \TPDA is expressible as an existential formula of linear arithmetic of doubly exponential size.
	For push-copy \TPDA, the complexity reduces to singly exponential.
\end{theorem}

The singly exponential complexity for push-copy \TPDA is particularly interesting,
since already for \TA the best currently known procedure builds a formula of exponential complexity \cite{FranzleQuaasShirmohammadiWorrell:2019}.
In fact, it is not possible to compute such a formula in polynomial time for push-copy \TPDA unless $\NP = \EXPTIME$,
and for \TA unless $\NP = \PSPACE$:
Assuming the contrary, since the satisfiability problem for the existential fragment is in \NP (cf.~\Cref{thm:ELA:NP}),
we would be able to solve the nonemptiness problem in \NP.
However, nonemptiness for \TA is \PSPACE-hard and, as we show below, it is \EXPTIME-hard for push-copy \TPDA.

The nonemptiness problem for several subclasses of push-copy \TPDA is \EXPTIME-complete.
The upper-bound follows from the push-copy assumption,
since we avoid the exponential blow-up of quantifier elimination of \Cref{sec:simplify:push-copy},
thus obtaining an equi-nonempty context-free grammar of singly exponential size.
This gives an \EXPTIME procedure to decide nonemptiness of push-copy \TPDA.
This subsumes previous results on \TPDA with timeless stack \cite[Theorem 4.1]{bouajjani:timed:PDA:94}
%discrete-time \TPDA \cite{Dang:PTA:2003},
and dense-time \TPDA without diagonal constraints \cite[Theorem 6]{AbdullaAtigStenman:DensePDA:12}.

The \EXPTIME-hardness results are obtained in two steps:
In the first step (\Cref{sec:complexity:TA}) we recall known \PSPACE-hardness results for the underlying \TA,
and in the second step (\Cref{sec:complexity:TPDA}) we present a generic reduction
showing that adding an untimed stack (which can model alternation) causes the complexity to jump to \APSPACE = \EXPTIME \cite{ChandraKozenStockmeyere:JACM:1981}
(\APSPACE is alternating \PSPACE).
\Cref{sec:complexity:1clockTPDA} contains additional complexity results for \TPDA with 1 clock.

\subsection{Complexity of \TA nonemptiness}
\label{sec:complexity:TA}

The \emph{regular intersection-nonemptiness problem} asks whether
$$\lang {A_1} \cap \cdots \cap \lang {A_n} \neq \emptyset$$
for given \NFA's $A_1, \dots, A_n$.
It is well-known that the regular intersection-nonemptiness problem is \PSPACE-complete \cite[Lemma 3.2.3]{Kozen:FOCS:1977}.
For the purpose of the reductions below, we assume \wlg that \NFA states belong to the set $\set{0, \dots, n - 1}$.
The following three lemmas establish \PSPACE-hardness (and hence also \PSPACE-completeness) of 
1) \TA with only fractional constraints, 2) discrete-time \TA with only equational non-diagonal integral constraints,
and 3) discrete-time TA with only non-diagonal modular constraints.

\begin{lemma}
	\label{lem:fractional:TA:PSPACE}
	The reachability problem for fractional \TA is \PSPACE-complete.
\end{lemma}

\begin{proof}
	We reduce from the regular intersection-nonemptiness problem.
	We build a fractional \TA $B$ which simulates the \NFA's $A_1, \dots, A_n$ by storing their control locations in the fractional clocks.
	For every $0 \leq i < n$, we have $n+1$ clocks $x^i_0, x^i_1, \dots, x^i_n$
	which during the simulation satisfy the invariant:
	\begin{align*}
		0 \leq \fract {x^i_0} < \fract {x^i_2} < \cdots < \fract {x^i_{n-1}} < 1
	\end{align*}
	and $\fract{x^i_n}$ equals $\fract {x^i_k}$ precisely when automaton $A_i$ is in state $q_i = k$.
	It is clear that an \NFA $A_i$ transition going from state $q_i = k$ to state $q_i' = k'$
	can be simulated by elapsing an integral amount of time (which can be checked by adding an extra control clock)
	and resetting $x^i_n$ precisely when $\fract {x^i_{k'}} = 0$.
\end{proof}

\begin{lemma}[\protect{\cite[Theorem 4.17]{AD94}}]
	\label{lem:integral:TA:PSPACE}
	The reachability problem for discrete-time \TA with constraints of the form $\x_i = k$, equivalently 
	$\floor{\x_i} = k$, ($k \in \N$) is \PSPACE-complete.
\end{lemma}

\begin{proof}
	The hardness argument given in \cite{AD94} consists in a direct simulation from the membership problem of linearly bounded Turing machines.
	An alternative reduction from the intersection-nonemptiness problem can be given.
	We have $n+1$ clocks $x, \y_1, y_1, \dots, y_n$ and we maintain the invariant that $y_i = k$ precisely when automaton $A_i$ is in state $k$.
	Setting $y_i$ to $k$ is performed by letting a total of $2n$ time units elapse (achieved using clock $x$),
	ensuring that whenever any clock $y_j = n$ it is immediately reset (including $y_i$),
	and additionally resetting $y_i$ when it equals $n - k$.
	At the end of the $2n$ time units elapse $y_i = k$,
	and all the other $y_j$ will have retained their original value.
\end{proof}

\begin{lemma}
	\label{lem:modular:TA:PSPACE}
	The reachability problem for discrete-time \TA with constraints of the form $\floor {\x_i} \eqv m k$ ($k, m \in \N$) is \PSPACE-complete.
\end{lemma}

\begin{proof}
	The reduction is similar to \Cref{lem:integral:TA:PSPACE},
	by maintaining the invariant $y_i \eqv n k$ whenever automaton $A_i$ is in state $k$.
\end{proof}

\subsection{Complexity of \TPDA nonemptiness}
\label{sec:complexity:TPDA}

Our complexity results for push-copy \TPDA are summarised below.

\begin{theorem}
	\label{thm:TPDA:nonemptiness:EXPTIME}
	The nonemptiness problem for push-copy \TPDA is \EXPTIME-complete, and \EXPTIME-hard already for
	1) fractional \TPDA with untimed stack,
	2) discrete-time \TPDA with untimed stack with only constraints of the form $\floor {x_i} = k$ ($k \in \N$),
	3) the same but with only constraints of the form $\floor {x_i} \eqv m k$ ($k,m \in \N$).
\end{theorem}

\noindent
Similar \EXPTIME-hardness results about extensions of pushdown automata exist in the literature,
such as pushdown automata extended with regular stack valuations \cite[Theorem 7]{EsparzaKuceraSchoon:IC:2003},
networks of communicating pushdown systems \cite[Proposition 2.4]{HeusnerLerouxMuschollSutre:LMCS:2012},
pushdown timed automata with nondeterministic clock assignments \cite[Lemma 7]{AbdullaAtigStenman:DensePDA:12},
and pushdown register automata over equality \cite[Theorems 10 and 11]{MurawskiTzevelekos:JCSS:2017} and more general homogeneous data \cite[Corollary 10]{ClementeLasota:CSL:2015}.

\begin{proof}[\protect{Proof (of \Cref{thm:TPDA:nonemptiness:EXPTIME})}]
	The \EXPTIME upper-bound for push-copy \TPDA follows from the observation that 1) we can avoid the exponential construction of \Cref{sec:simplify:push-copy},
	and 2) the other reductions to an equi-nonempty context-free grammar have a combined singly exponential complexity, and
	3) nonemptiness of context-free grammars is in \PTIME (and in fact complete for this class).

	For the \EXPTIME-hardness, we follow the approach of \cite[Lemma 7]{AbdullaAtigStenman:DensePDA:12}
	and reduce from the \emph{regular-context-free intersection nonemptiness problem},
	which amounts to deciding whether
	$$\lang A \cap \lang {A_1} \cap \cdots \cap \lang {A_n} \neq \emptyset$$
	for a \PDA $A$ and \NFA's $A_1, \dots, A_n$.
	The latter problem is \EXPTIME-complete (cf.~\cite[Proposition 2.4]{HeusnerLerouxMuschollSutre:LMCS:2012} and \cite[Theorem 7]{EsparzaKuceraSchoon:IC:2003}),
	which follows from the fact that \APSPACE = \EXPTIME and that a stack can be used to model alternation
	(similarly as in the proof that \PDA nonemptiness is \PTIME-complete, since \ALOGSPACE = \PTIME).
	The statement of \Cref{thm:TPDA:nonemptiness:EXPTIME} follows from reductions analogous to \Cref{lem:integral:TA:PSPACE,lem:modular:TA:PSPACE,lem:fractional:TA:PSPACE}.
\end{proof}

\subsection{Complexity of 1-clock \TPDA nonemptiness}
\label{sec:complexity:1clockTPDA}

Since for 1-clock \TA the nonemptiness problem can be solved in \NLOGSPACE \cite[Proposition 5.1]{LaroussinieMarkeySchnoebelen:CONCUR:2004},
one may wonder whether an analogous complexity collapse happens for 1-clock \TPDA
(\TPDA with 1 control clock and 1 stack clock).
With a timed stack, we have the following \PSPACE-hardness result.

\begin{lemma}
	The nonemptiness problem for 1-clock \TPDA is \PSPACE-hard.
\end{lemma}

\begin{proof}
	1-clock \TPDA can simulate classical 2-clock \TA whose nonemptiness problem is \PSPACE-complete~\cite{FJ15}, as follows.
	\TA clock $x_1$ becomes \TPDA control clock $x$,
	and the other \TPDA clock $x_2$ is stored as a stack clock $z$,
	in as many copies as there are upcoming tests involving it.
	For simulation of a reset of $x_2$, the \TPDA checks emptiness of the stack, and then 
	performs non-deterministically many push operations with constraint $z=0$.
	The simulation of a transition of \TA with clock constraint $x_1 - x_2 \sim k$ is done by a pop operation of \TPDA
	with the corresponding pop constraint $x - z \sim k$.
\end{proof}

If the stack has bounded height (even of height one), then the problem is \NP-hard.

\begin{lemma}
	The nonemptiness problem for 1-clock discrete-time \TPDA of stack height one is \NP-hard.
\end{lemma}

\begin{proof}
	We follow a similar reduction in \cite[Theorem 4]{ClementeHofmanTotzke:CONCUR:2019} from the \NP-complete \emph{subset sum problem} \cite{GareyJohnson},
	which amounts to establishing, given a set of nonnegative integers $A = \set{n_1, \dots, n_k} \subseteq \N$ and $n \in \N$ (encoded in unary),
	whether there is a subset thereof $B \subseteq A$ \st $\sum_{n_i \in B} = n$.
	The \TPDA has one clock $x$, which is initially $0$ and is immediately pushed on the stack.
	For each $n_i$ (one after the other) the automaton guesses whether $n_i \in B$ or not:
	If so, the automaton uses $x$ to elapse exactly $n_i$ time units, otherwise no time is elapsed.
	The automaton accepts only if it can pop with stack constraint $z = n$,
	which checks that exactly $n$ time units elapsed since the beginning of the run.
\end{proof}

However, the complexity does become polynomial if the stack is untimed.
\begin{lemma}
	The nonemptiness problem for 1-clock \TPDA with untimed stack is \PTIME-complete.
	(Hardness holds already for \PDA.)
\end{lemma}

\begin{proof}
	Since the stack is untimed, we can perform the polynomial time region construction from \cite{LaroussinieMarkeySchnoebelen:CONCUR:2004} on the finite control,
	preserving all reachability property.
	This yields an equi-nonempty \PDA, which can be solved in \PTIME.
\end{proof}
% !TEX root = main.tex

\section{Discussion}
\label{sec:discussion}

After the seminal work on \TA \cite{AD94},
many extensions have been proposed aiming at generalising \TA in different directions.
One such direction is the introduction of additional discrete data structures, usually resulting in increased expressive power.
A prominent example is the study of \TA extended with a stack, such as the \TPDA model that we study in this work.
This permits reasoning about real-time programs with procedure calls,
which exhibit a subtle interaction between recursion and timing constraints.
In this section,
we provide an extensive discussion on the relationship between \TPDA and other related models which have been proposed in the literature.

\paragraph{Pushdown timed automata}

One of the earliest extensions of \TA with a (classical, untimed) stack is \emph{pushdown timed automata} (\PDTA),
which have been proposed in the 1990's by Bouajjani \emph{et al.} \cite{bouajjani:timed:PDA:94}.
Since the stack is untimed, the very same region construction leading to the \PSPACE nonemptiness algorithm for \TA
can be used to show that nonemptiness for \PDTA is in \EXPTIME (cf.~\Cref{thm:TPDA:nonemptiness:EXPTIME}).
The construction of the \PDTA  reachability relation is more difficult,
and it has been solved only several years later both in discrete \cite{DangEtAl:CAV:2000}
and dense time \cite{Dang:CAV:2001,Dang:PTA:2003}.
Our construction shares with \cite{Dang:PTA:2003} an important intuition,
namely, separating clocks into integral and fractional part to simplify the analysis.

Going beyond the nonemptiness problem,
Quaas has recently shown that the problem of deciding language inclusion between a \PDTA and a one-clock \TA is undecidable
\cite{Quaas:LMCS:2015} (correcting the opposite claim in \cite{EmmiMajumdar:RealTime:2006}),
and that the universality problem is undecidable for the class of timed visibly one-counter automata.

\paragraph{Event-clock visibly pushdown automata}

\emph{Event-clock timed automata} \cite{AlurFixHenzinger:TCS:1999} constitute a robust subclass of \TA,
enjoining closure under Boolean operations and decidable nonemptiness and inclusion problems.
In the context of untimed context-free languages, a similar status is shared by \emph{visibly pushdown automata} \cite{AlurMadhusudan:STOC:2004}.
These two worlds have been joined in 2009 by Tang and Ogawa \cite{TangOgawa:SOFSEM:2009},
who showed that the resulting \emph{event-clock visibly pushdown automata} have a decidable inclusion problem.
We note that in such a model, being a subclass of \PDTA, the stack is untimed.

\paragraph{Dense-timed pushdown automata}
More recently, \emph{dense-timed pushdown automata} (\dtPDA)
have been proposed in 2012 by Abdulla \emph{et al.}~\cite{AbdullaAtigStenman:DensePDA:12}
as an extension of \PDTA featuring, for the first time, a \emph{timed stack}.
This was an important innovation, conferring a certain popularity to the model,
as testified by the number of works published in the years following its introduction.
The idea is to equip a stack symbol with a real-valued \emph{age} (i.e., a stack clock),
which initially is $0$ when pushed on the stack
and increases with the elapse of time at the same rate as the control clocks;
when a symbol is popped, its age is tested for membership in an interval.
In the syntax of this paper, \dtPDA correspond to \TPDA with only one stack clock $\z$,
and push/pop stack constraints are Boolean combinations of constraints the form $z \sim k$
(hence, no diagonal constraints, no fractional constraints, no modular constraints).
The nonemptiness problem for \dtPDA is \EXPTIME-complete \cite{AbdullaAtigStenman:DensePDA:12},
which is shown by an elegant region-based transformation untiming the stack \emph{while preserving nonemptiness}.
While it is sufficient to decide nonemptiness,
we note that the transformation of \cite{AbdullaAtigStenman:DensePDA:12} does not preserve the \emph{timed language} recognised by the automaton.

Subsequent works building on \dtPDA include \emph{nested timed automata}
(stack of \TA which can be pushed and popped),
whose nonemptiness problem reduce to \dtPDA \cite{LiCaiOgawaYuen:FORMATS:2013},
input/output \dtPDA \cite{MHemdiJulliandMassonRobbana:ICSTW:2015,MHemdiJulliandMassonRobbana:2016},
visible \dtPDA \cite{Devendra-BhaveKrishnaTrivedi:LATA:2016},
and classes of decidable timed multistack languages closed under Boolean operations (i.e., a so called \emph{perfect class}),
such as round-bounded \cite{BhaveDaveKrishnaPhawadeTrivedi:DLT:2016}, later generalised to scope-bounded \cite{BhaveKrishnaPhawadeTrivedi:DLT:2019}.
Alternative analyses of the elegant \dtPDA construction have been performed,
e.g., via a subclass of pushdown automata with stack manipulation \cite{UezatoMinamide:ATVA:2013}
and well-structured pushdown systems \cite{CaiOgawa:FLOPS:2014}.
A logical characterisation of the class of \dtPDA languages has also been proposed \cite{DrostePerevoshchikov:CSR:2015}.

One important expressiveness question about \dtPDA, apparently not considered in all the previous works,
is whether the class of timed language recognised by \dtPDA
is strictly larger than \PDTA; in other words, whether the timed stack adds any expressiveness.
This is a very relevant question to ask, before furthering the study of \dtPDA
(which otherwise would reduce to the study of \PDTA).
In \cite{ClementeLasota:LICS:2015} we made the surprising observation that this is not the case,
i.e., the class of timed languages recognised by \dtPDA and \PDTA is the same.
In other words, the timed stack of a \dtPDA can be \emph{untimed while preserving the timed language} (not just nonemptiness);
this even yields an optimal \EXPTIME decision procedure for nonemptiness.
This is the consequence of the interplay between the kind of clock constraints allowed in \dtPDA (non-diagonal integral constraints) and the monotonicity of time.

The semantic collapse of \dtPDA to \TPDA has potential consequences on the works mentioned above.
For instance, the logical characterisation of \dtPDA languages in \cite{DrostePerevoshchikov:CSR:2015}
uses certain \emph{distance matching} predicates constraining the time elapse between a push and its matching pop;
since the stack can be untimed, this raises the question whether such distance matching predicates are really necessary,
i.e., whether they enhance the expressive power of the logic.
Another example is visible \dtPDA \cite{Devendra-BhaveKrishnaTrivedi:LATA:2016},
which is supposedly strictly generalising the corresponding untimed stack model of Tang and Ogawa \cite{TangOgawa:SOFSEM:2009};
since our stack untiming construction preserves visibility,
this appears not to be the case.
This motivates the quest for a strictly more expressive generalisation of \PDTA and \dtPDA with a truly timed stack.

%a related logical characterisation for a model of input-determined dense-time integer-reset visibly pushdown automata has been presented in \cite{BhaveGuha:RP:2017}

%From a syntactic point of view,
%\TPDA significantly lifts the restrictions of \dtPDA \cite{AbdullaAtigStenman:DensePDA:12}---%
%which allow only classical non-diagonal constraints, i.e., interval tests,
%and thus has neither diagonal, nor modulo, nor fractional constraints

\paragraph{Synchronised recursive timed automata}

Uezato and Minamide observed in 2015 that adding fractional stack constraints prevents the stack from being untimed \cite{UezatoMinamide:LPAR15},
and thus strictly enriches the expressive power of the model.
(Example~\ref{ex:TPDA} is easily adapted to show this: just drop the modular constraint $\floor \z \eqv 2 0$ from
the constraint $\floor \z \eqv 2 0 \land \fract{\z} \leq \fract{\x}$
in the pop rules~\eqref{eq:popconstr}.
Moreover, if the fractional constraint $\fract{\z} \leq \fract{\x}$ is dropped,
then the example shows that also modular stack constraints strictly increase the expressiveness of the model.)
\noindent
This is in contrast with \TA, where, if epsilon transitions are available,
fractional/modular constraints do not increase expressiveness~\cite{BerardDiekertGastinPetit:1998:Epsilon,ChoffrutGoldwurm:Periodic:2000}.
The resulting model is called \emph{synchronised recursive timed automata} (\SRTA),
and Uezato and Minamide show that, despite their increased expressive power vs.~\dtPDA/\PDTA,
the construction of Abdulla \emph{et al.}~can be adapted to decide nonemptiness in \EXPTIME.

Our \emph{timed pushdown automata} model (\TPDA) is strongly influenced by \SRTA, and in fact is a syntactic extension thereof.
More precisely, we consider the full class of diagonal constraints as potential push constraints,
while in \SRTA the only push constraint is push-copy.
Moreover, we consider \emph{modular constraints}, which are not present in \SRTA.
The difference in the order of words \wrt~``pushdown timed automata'' (\PDTA)
stresses the fact that the stack is timed (and inherently so).
Despite the syntactical generalisation,
since in the presence of fractional constraints integral and modulo constraints can be removed (as we show in Sec.~\ref{sec:simplify:fractional}),
\TPDA are in fact expressively equivalent to \SRTA.
While Uezato and Minamide solve the control state reachability problem,
we study the more general problem of computing the reachability relation.
This means that our reduction techniques need not only preserve nonemptiness, like \cite{UezatoMinamide:LPAR15},
but additionally enable the reconstruction of the reachability relation.

\paragraph{Timed register pushdown automata}

Another approach to the study of timed systems is the use of \emph{timed registers} over $(\R, \leq, +1)$ \cite{BojanczykLasota:ICALP:2012},
which are analogous to clocks under the reset-point/local time semantics \cite{BengtssonJonssonLiliusYi:CONCUR:1998}
(cf.~also \cite{Fribourg:1998,Fribourg:LOPSTR:2000}).
We have pursued this direction with \emph{timed register pushdown automata} (\TRPDA),
showing that nonemptiness is decidable \cite{ClementeLasota:LICS:2015,ClementeLasotaLazicMazowiecki:LICS:2017}
and that the reachability relation is computable \cite{ClementeLasotaLazicMazowiecki:TOCL:2019}.
Using a construction along the lines of \cite{ClementeLasota:LICS:2015},
\TPDA nonemptiness reduces to \TRPDA nonemptiness.
One may wonder whether analogous constructions
can perform the same reduction for the reachability relation.
This is not the case,
since the former reduction forgets the exact value of ``very large'' clocks,
which preserves nonemptiness but not the reachability relation.
For this reason, in the present work we follow another route
by encoding the integral part of clocks in the language and reducing to a model with only fractional clocks.

\paragraph{Other \TA extensions}

Another expressive extension of \TA, called \emph{recursive timed automata} (\RTA),
has been proposed independently in 2010 by Trivedi and Wojtczak \cite{TrivediWojtczak:RTA:2010}
and by Benerecetti \emph{et al.}~\cite{BenerecettiMinopoliPeron:RTA:2010}.
\RTA use a timed stack to store the current clock valuation,
which does not evolve as time elapses
and can be restored at the time of pop.
This facility makes \RTA expressively incomparable to all models previously mentioned.
Nonemptiness for \RTA is undecidable,
even in the timed-bounded case and already for five clocks \cite{KrishnaManasaTrivedi:LATA:2015}.
An expressive \emph{recursive hybrid automata} model generalising \dtPDA and \RTA has been investigated in \cite{KrishnaManasaTrivedi:HSCC:2015}.

Alternative approaches for the analysis of timed system
extended with discrete data structures, such as stacks and queues,
include the tree automata approach of \cite{AkshayGastinKrishna:LMCS:2018},
based on the observation that the timed behaviours of such systems can be represented as graphs of bounded tree-width,
and a method based on propositional dynamic logic \cite{AkshayGastinJugeKrishna:LICS:2019}.

A model of \emph{commutative timed context-free grammars} modelling unbounded networks of timed processes
has recently been studied \cite{ClementeHofmanTotzke:CONCUR:2019},
where a more general problem of synchronised reachability (where all processes are required to have zero clocks at the end of the run)
is shown to be solvable in \EXPTIME for an arbitrary number of clocks,
and in \NP for one clock per process.

%The method of quantifier elimination was recently applied to the analysis of another timed model,
%namely \emph{timed communicating automata} \cite{Clemente:TCA:ArXiv:2018}.
%One direction, inspired by modelling time by means of registers instead of clocks~\cite{BL12icalp}
%\cite{ClementeLasota:CSL:2015}
% !TEX root = main.tex

\section{Conclusions}
\label{sec:conclusions}

We have shown how to compute the reachability relation for \TPDA, an expressive model combining recursion with timing constraints.
Several directions for further research can be identified.

One direction concerns decidable extensions of \TPDA with more general stack manipulation primitives beyond simple push/pop.
For untimed \PDA, general prefix-rewriting rules such as $\trule p {\pop(u); \push(v)} q$ allowing to atomically replace $u \in \Gamma^*$ with $v \in \Gamma^*$ from the top of the stack
do not increase the expressiveness of the model,
in the sense that such generalised \PDA still effectively recognise the class of context-free languages, and thus have a decidable nonemptiness problem.
For \TPDA this is not the case.
Already top-of-stack rewrite rules of the form $$\rewriteop \alpha \beta \psi$$ replacing $\alpha \in \Gamma$ from the top of the stack with $\beta \in \Gamma$,
where $\psi(\bar y, \bar y')$ relates via diagonal constraints old $\bar y$ and new $\bar y'$ top-of stack clocks,
yield a model with undecidable emptiness, and this holds already for a stack of height one.
This follows from the fact that non-destructive operations $\rewriteop \alpha \alpha {y_i' = y_i + 1}$
can be used to simulate counter increments, and similarly for decrements and zero tests,
thus enabling the simulation of 2 counter Minsky machines, which have an undecidable nonemptiness problem \cite{Minsky:1961}.
On the other hand, ``long'' push-only operations $$\pushop {\alpha_1, \dots, \alpha_n} {\psi(\bar x, \bar y_1, \dots, \bar y_n)}$$
can be simulated by a standard \TPDA push $\pushop {\tuple{\alpha_1, \dots, \alpha_n}} \psi$ by adding new stack clocks,
and thus do not enhance the expressive power of \TPDA.
``Long'' pop-only operations $\popop {\alpha_1, \dots, \alpha_n} {\psi(\bar x, \bar y_1, \dots, \bar y_n)}$
can be similarly converted to the short form.
It remains open whether the expressive model can be extended in such a way as to preserve decidability.

Another direction for further work is to identify suitable \TPDA subclasses
for which the nonemptiness problem has lower computational complexity.
We have seen in \Cref{sec:complexity} that 1-clock \TPDA are \PSPACE-hard with unbounded stack and \NP-hard with stack of height one.
It would be very interesting to find an expressive \PTIME subclass.

We represent the reachability relation as a formula of linear arithmetic,
whose integral part is computed using full Presburger arithmetic (due to the use of Parikh's theorem in \Cref{thm:Parikh}).
It is clear that fractional \TPDA from \Cref{sec:simplify:fractional} (to which we apply Parikh's theorem in \Cref{sec:fractional:TPDA})
read the ticks $\tick i$'s according to certain structural restrictions and thus do not generate all semilinear sets when projected to $\set{\tick 1, \dots, \tick n}$.
For instance, the semilinear set recognised by the formula $$\varphi(x, x') \;\equiv\; x' = 2 \cdot x$$
forces the final value $x'$ to be twice its initial value $x$, which clearly is not expressible as a \TPDA reachability relation \wrt a clock $x$.
It would be interesting to identify which fragment of linear arithmetic would capture precisely \TPDA reachability relations;
since full Presburger arithmetic is necessary to represent the Parikh image of transitions,
one would look for a fragment that describes the reachability relation projected to the clock values.

\let\section\oldsection

%\newpage
\bibliographystyle{abbrv}
\bibliography{bib}

\newpage
\appendix
\section{Proofs}
\label{app:proofs}

\subsection{Proofs for \Cref{sec:simplify:pop-integer-free:A}}

We first recall the correctness statement.
\lemCorrectnessA*

In order to formally prove the lemma above we need to find a stronger inductive statement.
This is provided by the two lemmas \Cref{lem:soundness:A,lem:completeness:A} below,
from which \Cref{lem:correctness:A} follows immediately.

%\lemSoundnessA*

\begin{restatable}[Soundness \protect{[A]}]{lemma}{lemSoundnessA}
	\label{lem:soundness:A}

	\begin{enumerate}[a)]
		\item If $(p, d, e), \mu \reach w {} (q, d, f), \nu$, $\mu(x_i) = \mu(x_i^e)$, and
		$\mu(x_i) + \delta_{\mu\nu} - \nu(x_j) \precsim k$, then
		\begin{align}
			\label{eq:soundness:A:a}
			p, \restrict \mu \X \reach {\tilde w} {} q, \restrict \nu \X.
		\end{align}

		\item If $(p, d), \mu \reach w {} (q, e), \nu$ and $\mu(x_i) = \mu(x_i^d)$, then
		\begin{align}
			\label{eq:soundness:A:b}
			p, \restrict \mu \X \reach {\tilde w} {} q, \restrict \nu \X.
		\end{align}
	\end{enumerate}
\end{restatable}

\begin{proof}
	We start with point a).
	%We first establish \eqref{eq:soundness:1:A}.
	%Since clock $x_i^d$ is never reset, $\nu(x_i^d) = \mu(x_i^d) + \delta_{\mu\nu}$ follows.
	%
	%The condition $\nu(x_i) = \nu(x_i^f)$ follows from the assumption $\mu(x_i) = \mu(x_i^e)$
	%since $x_i^f$ is reset when $x_i$ is; cf.~\eqref{eq:A:reset:2}.
	%
	We establish \eqref{eq:soundness:A:a} by structural induction on $\reach {} {}$,
	by following the characterisation of Lemma~\ref{lem:characterisation}.
	We focus on the interesting cases.
	In the transitivity case \eqref{eq:reachrel:E} we have runs
	\begin{align*}
		(p, d, e), \mu \reach u {} (r, d, g), \rho \reach v {} (q, d, f), \nu.
	\end{align*}
	Since $\delta_{\mu\nu} = \delta_{\mu\rho} + \delta_{\rho\nu}$, by assumption we have
	\begin{align}
		\label{eq:mu-nu}
		\mu(x_i) + \delta_{\mu\rho} +\delta_{\rho\nu} - \nu(x_j) \precsim k.
	\end{align}
	By definition of clock reset, we also have
	\begin{align}
		\label{eq:mu-rho}
		\rho(x_i) &\leq \mu(x_i) + \delta_{\mu\rho}, \\
		\label{eq:rho-nu}
		\nu(x_j) &\leq \rho(x_j) + \delta_{\rho\nu}.
	\end{align}
	Consequently, we derive
	\begin{align*}
		&\mu(x_i) + \delta_{\mu\rho} - \rho(x_j)
			\stackrel {\tiny(\textrm{by }\eqref{eq:rho-nu})} {\leq}
				\mu(x_i) + \delta_{\mu\rho} - (\nu(x_j) - \delta_{\rho\nu})
					\stackrel {\tiny(\textrm{by }\eqref{eq:mu-nu})} {\precsim}
						k \\
		&\rho(x_i) + \delta_{\rho\nu} - \nu(x_j)
			\stackrel {\tiny(\textrm{by }\eqref{eq:mu-rho})} {\leq}
				(\mu(x_i) + \delta_{\mu\rho}) + \delta_{\rho\nu} - \nu(x_j)
					\stackrel {\tiny(\textrm{by }\eqref{eq:mu-nu})} {\precsim}
						k.
	\end{align*}
	This is the crucial point where we make use of the intuition that outer type A constraints subsume (imply) the inner ones.
	Thanks to the two inequalities above,
	we can invoke the inductive assumption (twice)
	and we obtain a run
	\begin{align*}
		p, \restrict \mu \X \reach {\tilde u} {} r, \restrict \rho \X \reach {\tilde v} {} q, \restrict \nu \X.
	\end{align*}
	By \eqref{eq:reachrel:E}, we have \eqref{eq:soundness:A:a}, as required.

	In the push-pop case, if the symbol pushed is $\alphapsi$, then there are operations
	$\delta_\push = \pushop {\alphapsi} {\psicopy}$ and
	$\delta_\pop = \popop {\alphapsi} {\true}$
	with $\tilde \delta_\push = \delta_\push$ and $\tilde \delta_\pop = \popop \alpha \psi$ (recall $\psi \equiv y_i - x_j \precsim k$),
	and a run
	\begin{align*}
		(p, d, e), \mu, \varepsilon \goesto {\delta_\push}
			(r, d, e), \mu, (\alphapsi, \mu(x_i)) \reach u {}
				(s, d, f), \nu, (\alphapsi, \mu(x_i) + \delta_{\mu\nu}) \goesto {\delta_\pop}
					(q, d, f), \nu, \varepsilon,
	\end{align*}
	\st $w = \delta_\push \cdot u \cdot \delta_\pop$.
	%
	%\begin{align}
		%\label{eq:cond:assumption}
	%	\mu (x_i) + \delta_{\mu\nu} - \nu(x_j) \precsim k.
	%\end{align}
	%
	Thanks to the assumption $\mu(x_i) + \delta_{\mu\nu} - \nu(x_j) \precsim k$ we can apply the induction hypothesis, obtaining
	\begin{align*}
		p, \mu, \varepsilon \goesto {\tilde \delta_\push}
			r, \mu, (\alphapsi, \mu(x_i)) \reach {\tilde u} {}
				s, \nu, (\alphapsi, \mu(x_i) + \delta_{\mu\nu}) \goesto {\tilde \delta_\pop}
					q, \nu, \varepsilon,
	\end{align*}
	%for $\delta_\pop' = \popop {\alphapsi} {\psi}$ ,
	where the latter operation is legal thanks again to the assumption above.
	By \eqref{eq:reachrel:F}, we have \eqref{eq:soundness:A:a}, as required.

	We now continue to point b).
	We focus on the interesting case,
	which is the push-pop case \eqref{eq:reachrel:F} when we (necessarily) push $\hat\alphapsi$.
	We have operations $\delta_\push = \pushop {\hat\alphapsi} {\psicopy}$ and
	$\delta_\pop = \popop {\hat\alphapsi} {\true};\; \testop{x_i^d - x_j \precsim k}$
	with $\tilde \delta_\push = \pushop {\alphapsi} {\psicopy}$ and $\tilde \delta_\pop = \popop {\alphapsi} {\psi}$ (recall $\psi \equiv y_i - x_j \precsim k$),
	and a run of the form
	\begin{align*}
		(p, d), \mu, \varepsilon \goesto {\delta_\push}
			(r, d, d), \mu, (\alphapsi, \mu(x_i)) \reach u {}
				(s, d, e), \nu, (\alphapsi, \mu(x_i) + \delta_{\mu\nu}) \goesto {\delta_\pop}
					(q, e), \nu, \varepsilon,
	\end{align*}
	with $w = \delta_\push \cdot u \cdot \delta_\pop$.
	In particular, $\nu (x_i^d) - \nu(x_j) \precsim k$.
	By the inductive assumption from the first part of point a) applied to the middle run above,
	$\nu(x_i^d) = \mu(x_i^d) + \delta_{\mu\nu}$.
	By the last two conditions and the assumption $\mu(x_i) = \mu(x_i^d)$, we have
	\begin{align}
		\label{eq:soundness:cond:0}
		\mu(x_i) + \delta_{\mu\nu} - \nu(x_j) \precsim k.
	\end{align}
	We can thus apply the second part of point a),
	implying the existence of the middle run below:
	\begin{align*}
		p, \restrict \mu \X, \varepsilon \goesto {\tilde \delta_\push}
			r, \restrict \mu \X, (\alphapsi, \mu(x_i)) \reach {\tilde u} {}
				s, \restrict \nu \X, (\alphapsi, \mu(x_i) + \delta_{\mu\nu}) \goesto {\tilde\delta_\pop}
					q, \restrict \nu \X, \varepsilon,
	\end{align*}
	where the latter operation is legal thanks to \eqref{eq:soundness:cond:0}.
	By \eqref{eq:reachrel:F}, we have \eqref{eq:soundness:A:b}, as required.
\end{proof}

%\lemCompletenessA*
\begin{restatable}[Completeness \protect{[A]}]{lemma}{lemCompletenessA}
	\label{lem:completeness:A}
	Assume $p, \restrict \mu \X \reach {\tilde w} {} q, \restrict \nu \X$.
	\begin{enumerate}[a)]

		\item For every $d, e$ \st $\mu(x_i) = \mu(x_i^e)$, there is $f$ \st
		\begin{align}
			\label{eq:completeness:1}
			(p, d, e), \mu \reach w {} (q, d, f), \nu,
					\quad \nu(x_i) = \nu(x_i^f),
						\quad %\textrm { and }
							\nu(x_i^d) = \mu(x_i^d) + \delta_{\mu\nu}.
		\end{align}

		\item For every $d$ \st $\mu(x_i) = \mu(x_i^d)$, there is $e$ \st
		\begin{align}
			\label{eq:completeness:0}
			(p, d), \mu \reach w {} (q, e), \nu
				\quad \textrm{ and } \quad
					\nu(x_i) = \nu(x_i^e).
		\end{align}
	\end{enumerate}
\end{restatable}

\begin{proof}
	First of all, $\nu(x_i^d) = \mu(x_i^d) + \delta_{\mu\nu}$ holds by construction,
	since no run of the form \eqref{eq:completeness:1} can reset clock $x_i^d$,
	and $\nu(x_i) = \nu(x_i^e)$, $\nu(x_i) = \nu(x_i^f)$ follow from the reset policy of $x_i^e$ ($x_i^f$, resp.).
	%the fact
	%that every time $x_i$ is reset from a control location $(p, 1, d, d)$ or $(p, 1, d, 1-d)$,
	%then also $x_i^{1-d}$ is reset and we go to a control location $(q, 1, d, 1-d)$.

	We begin from case a).
	We proceed by structural induction on $\reach {\tilde w} {}$
	according to the characterisation of Lemma~\ref{lem:characterisation}.
	%
	%In the base case, $w = \delta$ is a single transition.
	%If it is of type input, read, test, time elapse, or reset not of $x_i$,
	%then we immediately have $(p, 1, d, e), \mu \reach {w'} {} (q, 1, d, f), \nu$
	%for $w' = \delta$ and $f = e$.
	%
	%If $\delta = \resetop{Y \cup \set{x_i}}$ resets $x_i$,
	%then we take $\delta' = \resetop{Y \cup \set{x_i, x_i^{1-d}}}$ and $f = 1 - d$.
	%
	%Independently of whether $e = d$ or $1-d$,
	%we have $(p, 1, d, e), \mu \reach {\delta'} {} (q, 1, d, f), \nu$.
	We focus on the two inductive cases, which are the most interesting.
	In the transitivity case,
	\begin{align*}
		p, \restrict \mu \X \reach {\tilde u} {} r, \restrict \rho \X \reach {\tilde v} {} s, \restrict \nu \X, \quad \textrm{ for some } r, \rho,
	\end{align*}
	with $\tilde w = \tilde u \cdot \tilde v$.
	By the inductive hypothesis applied twice,
	\begin{align*}
		(p, d, e), \mu \reach u {} (q, d, g), \rho \reach v {} (q, d, f), \nu
	\end{align*}
	and thus $(p, d, e), \mu \reach w {} (q, d, f), \nu$ by transitivity.

	In the push-pop case there are operations
	${\tilde\delta_\push = \pushop{\alphapsi} \psicopy}$ and
	$\tilde\delta_\pop = \popop{\alphapsi}{\psi}$ \st
	$\delta_\push = \tilde\delta_\push$,
	$\delta_\pop = \popop{\alphapsi}{\true}$, and
	\begin{align*}
		p, \restrict \mu \X \goesto {\tilde\delta_\push} r, \restrict \mu \X, (\alphapsi, \mu(x_i))
			\reach {\tilde u} {} s, \restrict \nu \X, (\alphapsi, \mu(x_i)+\delta_{\mu\nu})
				\goesto {\tilde\delta_\pop} q, \restrict \nu \X.
	\end{align*}
	By inductive assumption, we can find the middle run in
	\begin{align*}
		(p, d, e), \mu \goesto {\delta_\push} (r, d, e), \mu, (\alphapsi, \mu(x_i))
			\reach u {} (s, d, f), \nu, (\alphapsi, \mu(x_i)+\delta_{\mu\nu})
				\goesto {\delta_\pop} (q, d, f), \nu,
	\end{align*}
	because the pop condition is trivial.

	The push-pop case when the stack symbol is not $\alphapsi$ follows straightforwardly from the inductive assumption.

	We now proceed to point b).
	%All cases are either immediate or follow by a direct application of the inductive hypothesis.
	The only non-trivial case is the push-pop case.
	There are operations
	$\tilde\delta_\push = \pushop{\alphapsi} \psicopy$ and
	$\tilde\delta_\pop = \popop{\alphapsi}{\psi}$ \st
	$\delta_\push = \pushop{\widehat\alphapsi}{\psicopy}$,
	$\delta_\pop = \popop{\widehat\alphapsi}{\true}; \testop {x_i^d - x_j \precsim k}$, and
	\begin{align*}
		p, \restrict \mu \X \reach {\tilde\delta_\push} {} r, \restrict \mu \X, (\alphapsi, \mu(x_i))
			\reach {\tilde u} {} s, \restrict \nu \X, (\alphapsi, \mu(x_i)+\delta_{\mu\nu})
				\reach {\tilde \delta_\pop} {} q, \restrict \nu \X.
	\end{align*}
	In particular, it holds that
	\begin{align}
		\label{eq:completeness:cond:0}
		\mu(x_i)+\delta_{\mu\nu} - \nu(x_j) \precsim k.
	\end{align}
	By point a) we can find the middle run in
	\begin{align*}
		(p, d), \mu \reach {\delta_\push} {} (r, d, d), \mu, (\widehat\alphapsi, \mu(x_i))
			\reach u {} (s, d, e), \nu, (\widehat\alphapsi, \mu(x_i)+\delta_{\mu\nu})
				\reach {\delta_\pop} {} (q, e), \nu
	\end{align*}
	\st $\nu(x_i) = \nu(x_i^e)$ and $\nu(x_i^d) = \mu(x_i^d) + \delta_{\mu\nu}$.
	By the last equation, the assumption $\mu(x_i) = \mu(x_i^d)$,
	and \eqref{eq:completeness:cond:0}, it follows that $\nu(x_i^d) - \nu(x_j) \precsim k$ holds,
	and thus $\delta_\pop'$ can be executed.
	By \eqref{eq:reachrel:F}, we obtain \eqref{eq:completeness:0}, as required.
\end{proof}

\subsection{Proofs for \Cref{sec:simplify:pop-integer-free:B}}

We first recall the lemma stating the correctness of the construction.

\lemCorrectnessB*

The lemma above follows immediately from
the stronger inductive statements \Cref{lem:soundness:B,lem:completeness:B} below.

%\lemSoundnessB*
\begin{restatable}[Soundness \protect{[B]}]{lemma}{lemSoundnessB}
	\label{lem:soundness:B}
	If
	\begin{align}
		\label{eq:00}
		(p, 0, d, d), \mu &\reach w {} (q, 0, e, e), \nu, \textrm{ or } \\
		\label{eq:11}
		(p, 1, d, e), \mu &\reach w {} (q, 1, d, f), \nu, \textrm{ or } \\
		\label{eq:22}
		(p, 2, d, d), \mu &\reach w {} (q, 2, e, e), \nu, \textrm{ or } \\
		\label{eq:02}
		(p, 0, d, d), \mu &\reach w {} (q, 2, e, e), \nu,
	\end{align}
	and $\mu(x_i) = \mu(x_i^d)$, then
	\begin{align}
		\label{eq:conclusion}
		p, \mu \reach {\tilde w} {} q, \nu
			\quad \textrm{ and } \quad
				\nu(x_i) = \nu(x_i^e).
	\end{align}
	Moreover,
	\begin{align}
		\label{eq:soundeness:1:1}
		&\nu(x_i^d) = \mu(x_i^d) + \delta_{\mu\nu}
			&&\textrm{ in case \eqref{eq:11}}, \textrm{ and } \\
		\label{eq:soundeness:0:2}
			&\mu(x_i) + \delta_{\mu\nu} - \nu(x_j) \succsim k
			&&\textrm{ in case \eqref{eq:02}}.
	\end{align}
\end{restatable}
\begin{proof}
	By direct inspection, the lemma considers all possible cases of runs in $\P_B$.
	Time elapse and resetting $x_i^d$ clearly preserve the invariant.
	In case \eqref{eq:11}, $x_i^d$ is never reset,
	and thus \eqref{eq:soundeness:1:1} holds.
	
	We consider the transitivity cases.
	The first three cases are of the form
	\begin{align*}
		(p, 0, d, d), \mu &\reach u {} (r, 0, f, f), \rho \reach v {} (q, 0, e, e), \nu, \textrm{ or } \\
		(p, 1, d, e), \mu &\reach u {} (r, 0, d, g), \rho \reach v {} (q, 1, d, f), \nu, \textrm{ or } \\
		(p, 2, d, d), \mu &\reach u {} (r, 2, f, f), \rho \reach v {} (q, 2, e, e), \nu,
	\end{align*}
	with $w = u \cdot v$.
	For each of them, by a double application of the induction hypothesis and \eqref{eq:reachrel:E}
	we obtain \eqref{eq:conclusion}.
	We have another case of the form
	\begin{align*}
		(p, 0, d, d), \mu \reach u {} (r, 0, f, f), \rho \reach v {} (q, 2, e, e), \nu.
	\end{align*}
	By the inductive hypothesis applied twice % (implying $\rho(x_i) = \rho(x_i^f)$)
	we have $\nu(x_i) = \nu(x_i^e)$ and
	\begin{align*}
		p, \mu \reach {\tilde u} {} r, \rho \reach {\tilde v} {} q, \nu
			\quad \textrm{ and } \quad 
				\rho(x_i) + \delta_{\rho\nu} - \nu(x_j) \succsim k.
	\end{align*}
	By definition, $\rho(x_i) \leq \mu(x_i) + \delta_{\mu\rho}$, and thus
	\begin{align*}
		\mu(x_i) + \delta_{\mu\nu} - \nu(x_j) =
			\mu(x_i) + \delta_{\mu\rho} + \delta_{\rho\nu} - \nu(x_j) \geq
				\rho(x_i) + \delta_{\rho\nu} - \nu(x_j) \succsim k,
	\end{align*}
	establishing \eqref{eq:soundeness:0:2} as required.
	In the last transitivity case, we have a run
	\begin{align*}
		(p, 0, d, d), \mu \reach u {} (r, 2, f, f), \rho \reach v {} (q, 2, e, e), \nu.
	\end{align*}
	By the inductive assumption applied twice
	%(implying $\rho(x_i) = \rho(x_i^f)$)
	we have $\nu(x_i) = \nu(x_i^e)$ and
	\begin{align*}
		p, \mu \reach {\tilde u} {} r, \rho \reach {\tilde v} {} q, \nu
			\quad \textrm{ and } \quad
				\mu(x_i) + \delta_{\mu\rho} - \rho(x_j) \succsim k.
	\end{align*}
	By definition, $\nu(x_j) \leq \rho(x_j) + \delta_{\rho\nu}$, and thus
	\begin{align*}
		\mu(x_i) + \delta_{\mu\nu} - \nu(x_j) =
			\mu(x_i) + \delta_{\mu\rho} - (\nu(x_j) - \delta_{\rho\nu}) \geq
				\mu(x_i) + \delta_{\mu\rho} - \rho(x_j) \succsim k,
	\end{align*}
	establishing \eqref{eq:soundeness:0:2} as required.

	In the push-pop case, if the stack symbol is $\alpha \not \in \set{\alphapsi, \hat\alphapsi}$,
	then we conclude immediately by an application of the induction hypothesis.
	This covers entirely the cases \eqref{eq:00} and \eqref{eq:11},
	and one subcase for each of \eqref{eq:22} and \eqref{eq:02}.
	%We conclude by the induction hypothesis and \eqref{eq:reachrel:F}.

	We now consider the push-pop cases where the stack symbol is
	$\alpha \in \set{\alphapsi, \hat\alphapsi}$.
	In the first one, for $b \in \set{0, 2}$ we have a run
	\begin{align*}
		(p, b, d, d), \mu \goesto {\pushop{\alphapsi}{\psicopy}}
			(r, 0, d, d), \mu, (\alphapsi, \mu(x_i)) \reach u {}
				(s, 2, e, e), \nu, (\alphapsi, \mu(x_i) + \delta_{\mu\nu}) \\
					\goesto {\popop{\alphapsi}{\true}}
						(q, 2, e, e), \nu
	\end{align*}
	By the inductive assumption $\nu(x_i) = \nu(x_i^e)$
	and there exists a valid run
	\begin{align*}
		p, \mu \goesto {\pushop{\alphapsi}{\psicopy}}
			r, \mu, (\alphapsi, \mu(x_i)) \reach {\tilde u} {}
				s, \nu, (\alphapsi, \mu(x_i) + \delta_{\mu\nu}) \goesto {\popop{\alphapsi}{\psi}}
					q, \nu.
	\end{align*}
	The pop transition above is legal since
	$\mu(x_i) + \delta_{\mu\nu} - \nu(x_j) \succsim k$ by inductive assumption \eqref{eq:soundeness:0:2}.

	In the second one, for $b \in \set{0, 2}$ we have a run
	\begin{align*}
		&(p, b, d, d), \mu \goesto {\pushop{\hat\alphapsi}{\psicopy}}
			(r, 1, d, d), \mu, (\hat\alphapsi, \mu(x_i)) \reach u {} \\
				\reach {} {}
					&(s, 1, d, e), \nu, (\hat\alphapsi, \mu(x_i) + \delta_{\mu\nu}) \goesto {\popop{\alphapsi}{\true};\; \testop{x_i^d - x_j \succsim k}}
						(q, 2, e, e), \nu.
	\end{align*}
	In particular, $\nu(x_i^d) - \nu(x_j) \succsim k$ holds.
	By the inductive assumption we have a run
	\begin{align*}
		p, \mu \goesto {\pushop{\alphapsi}{\psicopy}}
			r, \mu, (\hat\alphapsi, \mu(x_i)) \reach {\tilde u} {}
				s, \nu, (\hat\alphapsi, \mu(x_i) + \delta_{\mu\nu}) \goesto {\popop{\alphapsi}{\psi}}
					q, \nu.
	\end{align*}
	The pop transition above is legal since
	$\nu(x_i^d) = \mu(x_i^d) + \delta_{\mu\nu}$ by \eqref{eq:soundeness:1:1}
	and $\mu(x_i) = \mu(x_i^d)$ by assumption,
	thus implying \eqref{eq:soundeness:0:2}, as required.
\end{proof}

%\lemCompletenessB*
\begin{restatable}[Completeness \protect{[B]}]{lemma}{lemCompletenessB}
	\label{lem:completeness:B}

	Assume we have a run
	\begin{align}
		\label{eq:completeness:assumption}
		p, \restrict \mu \X \reach {\tilde w} {} q, \restrict \nu \X.
	\end{align}
	%
	%For every $d$, if there is $e$ \st $\nu(x_i) = \nu(x_i^e)$
	%
	For every $d$ \st $\mu(x_i) = \mu(x_i^d)$,
	\begin{align}
		\label{eq:completeness:22}
		\exists e \cdot (p, 2, d, d), \mu \reach w {} (q, 2, e, e), \nu
			\quad \textrm{ and } \quad
				\nu(x_i) = \nu(x_i^e).
	\end{align}
	Additionally:
	\begin{enumerate}[a)]

		\item If in \eqref{eq:completeness:assumption} no $\alphapsi$ is pushed on the stack,
		then for every $d$ \st $\mu(x_i) = \mu(x_i^d)$,
		\begin{align}
			\label{eq:completeness:00}
			\exists e \cdot (p, 0, d, d), \mu &\reach w {} (q, 0, e, e), \nu
				\quad \textrm{ and } \quad
					\nu(x_i) = \nu(x_i^e), \\
		\intertext{and for every $d, e$ \st $\mu(x_i) = \mu(x_i^e)$,}
			\label{eq:completeness:11}
			\exists f \cdot (p, 1, d, e), \mu &\reach w {} (q, 1, d, f), \nu,
				\  \nu(x_i) = \nu(x_i^f),
					\  %\textrm{ and }
						\nu(x_i^d) = \mu(x_i^d) + \delta_{\mu\nu}.
		\end{align}

		\item If in \eqref{eq:completeness:assumption} some $\alphapsi$ is pushed on the stack,
		then for every $d$ \st $\mu(x_i) = \mu(x_i^d)$,
		\begin{align}
			\label{eq:completeness:02}
			\exists e \cdot (p, 0, d, d), \mu &\reach w {} (q, 2, e, e), \nu
				\quad \textrm{ and } \quad
					\nu(x_i) = \nu(x_i^e).
		\end{align}
	\end{enumerate}
\end{restatable}
\begin{proof}
	First of all, the conditions $\nu(x_i) = \nu(x_i^e)$ and $\nu(x_i) = \nu(x_i^f)$ appearing in the statement of the lemma
	follow directly from the reset policy of $x_i$ \eqref{eq:B:reset:02}, \eqref{eq:B:reset:1}.
	Moreover, in case a), the condition $\nu(x_i^d) = \mu(x_i^d) + \delta_{\mu\nu}$ from \eqref{eq:completeness:11}
	follows immediately from the fact that $x_i^d$ is never reset in a run of the form \eqref{eq:completeness:11}.

	We proceed by induction on the characterisation of $\reach {\tilde w} {}$ from Lemma~\ref{lem:characterisation}.
	We focus on the interesting cases, which are transitivity and push-pop.
	In the transitivity case, we have a run of the form
	\begin{align}
		\label{eq:completeness:transitivity}
		p, \restrict \mu \X \reach {\tilde u} {} r, \restrict \rho \X \reach {\tilde v} {} q, \restrict \nu \X,
	\end{align}
	with $\tilde w = \tilde u \cdot \tilde v$.
	The condition \eqref{eq:completeness:22} holds by a double application of the induction hypothesis to the above.
	Moreover, if in the run \eqref{eq:completeness:transitivity} no $\alphapsi$ is pushed on the stack,
	then the same holds true in the two component runs,
	and each of \cref{eq:completeness:00,eq:completeness:11} follows from a double application of the induction hypothesis.
	On the other hand, if in \eqref{eq:completeness:transitivity} some $\alphapsi$ is pushed on the stack,
	then we have two subcases.
	If $\alphapsi$ is pushed in the first component run $p, \restrict \mu \X \reach {\tilde u} {} r$,
	then by the inductive hypothesis \eqref{eq:completeness:02} and \eqref{eq:completeness:22} we have
	\begin{align*}
		(p, 0, d, d), \mu \reach u {} (r, 2, f, f), \rho \reach v {} (q, 2, e, e), \nu.
	\end{align*}
	If $\alphapsi$ is not pushed in the first component run, then it must be pushed in the second one;
	by the inductive hypothesis \eqref{eq:completeness:00} and \eqref{eq:completeness:02} we have
	\begin{align*}
		(p, 0, d, d), \mu \reach u {} (r, 0, f, f), \rho \reach v {} (q, 2, e, e), \nu.
	\end{align*}
	In either case, by \eqref{eq:reachrel:E} we have \eqref{eq:completeness:02}, as required.

	%We have to establish the four cases \cref{eq:completeness:00,eq:completeness:22,eq:completeness:02,eq:completeness:11}.
	%Cases \cref{eq:completeness:00,eq:completeness:22,eq:completeness:11}
	%follow immediately by applying the respective inductive hypothesis twice.
	%Case \cref{eq:completeness:02} can be established in two ways, with the inductive hypothesis being applied
	%either to \cref{eq:completeness:00,eq:completeness:02},
	%or to \cref{eq:completeness:02} and \cref{eq:completeness:22}.

	In the push-pop case we have a run of the form
	\begin{align}
		\label{eq:completeness:push-pop}
		p, \restrict \mu \X \goesto {\pushop{\alphapsi}{\psicopy}}
			r, \restrict \mu \X, (\alphapsi, \mu(x_i)) \reach {\tilde u} {}
				s, \restrict \nu \X, (\alphapsi, \mu(x_i) + \delta_{\mu\nu}) \goesto {\popop{\alphapsi}{\psi}}
					q, \restrict \nu \X.
	\end{align}
	In particular, 
	\begin{align}
		\label{eq:completeness:push-pop:psi}
		\mu(x_i) + \delta_{\mu\nu} - \nu(x_j) \succsim k.
	\end{align}
	Since $\alphapsi$ is pushed on the stack in this run, we just need to establish \eqref{eq:completeness:02}.
	There are two cases to consider, depending on whether some more $\alphapsi$ is pushed in
	\begin{align}
		\label{eq:completeness:push-pop:component}
		r, \restrict \mu \X \reach {\tilde u} {} s, \restrict \nu \X.
	\end{align}
	If no other $\alphapsi$ is pushed in \eqref{eq:completeness:push-pop:component},
	then by the inductive hypothesis \eqref{eq:completeness:11}
	\begin{align}
		%\label{eq:completeness:push-pop:cond1}
		\nonumber
		&(r, 1, d, d), \mu \reach u {} (s, 1, d, e), \nu, \textrm{ and } \\
		\label{eq:completeness:push-pop:cond2}
		&\nu(x_i^d) = \mu(x_i^d) + \delta_{\mu\nu}.
	\end{align}
	By assumption, $\mu(x_i) = \mu(x_i^d)$, and thus from \eqref{eq:completeness:push-pop:psi}
	and \eqref{eq:completeness:push-pop:cond2}, we have
	\begin{align*}
		\nu(x_i^d) - \nu(x_j) \succsim k.
	\end{align*}
	Thus, $\nu \models x_i^d - x_j \succsim k$, and
	\begin{align*}
		\forall b \in \set{0, 2} \cdot (p, b, d, d), \mu \goesto {\pushop{\hat\alphapsi}{\psicopy}}
			(r, 1, d, d), \mu, (\hat\alphapsi, \mu(x_i)) \\
				\reach u {} (s, 1, d, e), \nu, (\hat\alphapsi, \mu(x_i) + \delta_{\mu\nu})
					\goesto {\popop{\hat\alphapsi}{\true}; \testop{x_i^d - x_j \succsim k}}
						(q, 2, e, e), \nu,
	\end{align*}
	yielding for $b = 0$ the sought run \eqref{eq:completeness:02} by \eqref{eq:reachrel:F}.

	In the other case, some more $\alphapsi$ is pushed in \eqref{eq:completeness:push-pop:component}.
	By the inductive hypothesis \eqref{eq:completeness:02} we obtain the middle run in
	\begin{align*}
		&\forall b \in \set{0, 2} \cdot (p, b, d, d), \mu \goesto {\pushop{\alphapsi}{\psicopy}}
			(r, 0, d, d), \mu, (\alphapsi, \mu(x_i)) \\
		\reach u {}
			&(s, 2, d, e), \nu, (\alphapsi, \mu(x_i) + \delta_{\mu\nu}) \goesto {\popop{\alphapsi}{\true}}
				(q, 2, e, e), \nu,
	\end{align*}
	yielding for $b = 0$ the sought run \eqref{eq:completeness:02} by \eqref{eq:reachrel:F}.

	Finally, notice that in either of the two cases above,
	for $b = 2$ we obtain \eqref{eq:completeness:22}, as required.
\end{proof}

\subsection{Proofs for \Cref{sec:simplify:fractional}}

In this section we prove correctness of the construction from \Cref{sec:simplify:fractional}.
We start by recalling the two soundness and completeness statements.
\lemFractionalSoundness*
\lemFractionalCompleteness*

In order two prove the two lemmas above we need to find suitable stronger inductive statements.
Those are found below in \Cref{lem:fractional:soundness:inductive},
resp., \Cref{lem:fractional:completeness:inductive},
from which \Cref{lem:fractional:soundness,lem:fractional:completeness} follow immediately.

The following lemma states some structural properties of the automaton $\QQ$.
\begin{lemma}
	\label{lem:fractional:basics}
	If $\tuple {p, \lambda, \U}, \mu \reach w {} \tuple {q, \xi, \V}, \nu$ then
	\begin{enumerate}[a)]
		\item $\U \subseteq \V$, and
		\item for every clock $\x_i \in \X \setminus \V$, $\card w_{\tick i} = 0$.
	\end{enumerate}
\end{lemma}

\begin{proof}
	\begin{enumerate}[a)]
		\item The third component $\U$ in a control location $\tuple {p, \lambda, \U}$
		is changed only in transitions \eqref{eq:fractional:reset}, where it increases.

		\item The symbol $\tick i$ is read only from a control location of the form $\tuple {p, \lambda, \U, q}$ with $\x_i \in \U$;
		the claim follows from the previous point since $\U$ can only increase. \hfill \qedhere
	\end{enumerate}
\end{proof}

For a unary abstraction $\lambda \in \Lambda_M$ and a clock valuation $\mu \in \Rgeq^\X$, let
\begin{align*}
	R_\lambda(\mu) &= \setof { \nu \in \Rgeq^\X } {\fract \mu = \fract \nu \textrm{ and } \lambda(\nu) = \lambda}
\end{align*}
be the set of clock valuations $\nu$ having the same fractional values as $\mu$ and unary abstraction $\lambda$.
\begin{restatable}{lemma}{lemFractionalSoundnessInductive}
	\label{lem:fractional:soundness:inductive}
	For control locations $p, q \in \L$,
	clock valuations ${\mu, \nu \in \Rgeq^\X}$,
	unary abstractions $\lambda, \xi \in \Lambda_M$,
	and a sequence of operations $w \in (\Delta')^*$, if
	\begin{align*}
		\tuple {p, \lambda, \U}, \mu \reach w {} \tuple {q, \xi, \V}, \nu
			%\textrm{ and }
				%\lambda(\mu) = \lambda;
	\end{align*}
	then for every $\mu' \in R_\lambda(\mu)$,
	there is $\nu' \in R_\xi(\nu)$ \st
	%$\forall \x_i \in \V \setminus \U \st \floor {\nu'(\x_i)} = \card {w}_{\checkmark_i}$,
	%$\forall \x_i \in \U \st \floor {\nu'(\x_i)} = \floor {\mu'(\x_i)} + \card {w}_{\checkmark_i}$,
	%
	\begin{align*}
		&p, \mu' \reach {} {} q, \nu'
			\quad \textrm{ and } \quad 
				\textrm{(C)}\ \forall \x_i \in \V \st \floor {\nu'(\x_i)} = \card {w}_{\checkmark_i} + \floor {\mu'(\x_i)} \cdot \condone {\x_i \in \U}.
	\end{align*}
\end{restatable}
%\lemFractionalSoundnessInductive*
\begin{proof}%[Proof (of \Cref{lem:fractional:soundness:inductive})]
	We proceed by induction on the characterisation of the reachability relation from Lemma~\ref{lem:characterisation}.
	Let $\tuple{p, \lambda, \U}, \mu \reach w {} \tuple{q, \xi, \V}, \nu$ and consider an arbitrary $\mu' \in R_\lambda(\mu)$.

	If $w = \readop a$, then in fact $\tuple{p, \lambda, \U}, \mu \reach {\readop a} {} \tuple{q, \lambda, \U}, \nu$
	and by construction we have $p, \mu' \reach {\readop a} {} q, \nu'$ for $\nu' = \mu'$ (the other conditions are trivially satisfied).

	If $w = \testop {\restrict \varphi \lambda}$,
	then in fact $\tuple{p, \lambda, \U}, \mu \reach {\testop {\restrict \varphi \lambda}} {} \tuple{q, \lambda, \U}, \nu$.
	Since $\mu \models \restrict \varphi \lambda$ holds,
	$\lambda$ satisfies the integral and modular constraints of $\varphi$,
	and $\fract {\mu}$ its fractional constraints.
	By assumption $\lambda(\mu') = \lambda, \fract{\mu'} = \fract{\mu}$,
	and thus $\mu' \models \varphi$.
	Consequently, $p, \mu' \reach {\testop \varphi} {} q, \nu'$ for $\nu' = \mu'$.

	If $w = \resetop \Y$,
	then in fact $\tuple{p, \lambda, \U}, \mu \reach {\resetop \Y} {} \tuple{q, \xi, \V}, \nu$,
	where $\xi = \lambda[\Y \mapsto 0]$, $\nu = \mu[\Y \mapsto 0]$, $\Y \subseteq \X \setminus \U$, and $\U \subseteq \V \subseteq \U \cup \Y$.
	Let $\nu' = \mu'[\Y \mapsto 0]$; by construction, there is a run $p, \mu' \reach {\resetop \Y} {} q, \nu'$.
	Since $\fract {\mu'} = \fract {\mu}$, $\fract{\nu'} = \fract{\nu}$.
	Since $\lambda(\mu') = \lambda$, $\lambda(\nu') = \xi$.
	Thus, $\nu' \in R_\xi(\nu)$, as required.
	Notice that $\card w {\tick i} = 0$.
	If $\x_i \in \U$, then it is not reset, and thus $\nu'(\x_i) = \mu'_(\x_i)$;
	if $\x_i \in \V \setminus \U \subseteq \Y$, then $\x_i$ is reset, and thus $\nu'(\x_i) = 0$, as required.
	
	The simulation of time elapse is more involved.
	The sequence of operations is $w = \readop{\varepsilon}; \ops_{\Y_1, \U}; \cdots; \ops_{\Y_m, \U}; \readop{\varepsilon}$,
	yielding a run of the form
	\begin{align*}
		\tuple{p, \lambda_0, \U}, \mu_0 \reach {\readop{\varepsilon}} {}
			\tuple{p, \lambda_0, \U, q}, \mu_0 \reach {\ops_{\Y_1, \U}} {}
				%\tuple{p, \lambda_1, \U, q}, \mu_1 \reach {\ops_{\Y_2, \U}} {}
					\cdots \reach {\ops_{\Y_m, \U}} {}
						\tuple{p, \lambda_m, \U, q}, \mu_m \\ \qquad \reach {\readop{\varepsilon}} {}
							\tuple{q, \lambda_m, \U}, \mu_m.
	\end{align*}
	where $\lambda_0 = \lambda$, $\lambda_m = \xi$, and $\V = \U$.
	Since the first and last steps are trivial,
	it suffices to focus on a single time elapse cycle $w = \ops_{\Y, \U}$, i.e.,
	\begin{align*}
		\tuple{p, \lambda, \U, q}, \mu \reach {\ops_{\Y, \U}} {} \tuple{p, \xi, \U, q}, \nu,
	\end{align*}
	where $\xi = \lambda[\Y \mapsto \Y + 1]$ and $\nu = \mu + \delta$ for some $\delta \in \Rgeq$.
	Let $\nu' = \mu' + \fract{\delta}$ and we have a run $p, \mu' \reach {} {} q, \nu'$.
	Since $\fract{\mu'} = \fract{\mu}$, $\fract{\nu'} = \fract{\nu}$.
	By the definition of $\ops_{\Y, \U}$, clocks in $\Y$ have maximal fractional value in $\mu$ and they have zero fractional value in $\nu$;
	thus for $\x_i \in \Y$, $\floor{\nu'} = \floor{\mu'} + 1$
	and for $\x_i \in \X\setminus\Y$, $\floor{\nu'} = \floor{\mu'}$.
	Since $\lambda(\mu') = \lambda$, $\lambda(\nu') = \xi$.
	We obtain $\nu' \in R_\xi(\nu)$, as required.
	Since by construction $\card w {\tick i} = 1$ for $\x_i \in \Y \cap \U$ and $0$ otherwise,
	this also entails $\floor{\nu'} = \floor{\mu'} + \card w {\tick i}$, as required.
	
	In the transitivity case, $w = uv$ and we have two runs
	\begin{align*}
		\tuple{p, \lambda, \U}, \mu \reach u {} \tuple{r, \theta, \T}, \rho \reach v {} \tuple{q, \xi, \V}, \nu,
	\end{align*}
	where $\U \subseteq \T \subseteq \V$ by Lemma~\ref{lem:fractional:basics}.
	By inductive assumption applied to the first run we obtain $\rho' \in R_\theta(\rho)$
	and a run $p, \mu' \reach {} {} r, \rho'$,
	and by inductive assumption applied to the second run we obtain $\nu' \in R_\xi(\nu)$
	and a run $r, \rho' \reach {} {} q, \nu'$,
	yielding $p, \mu' \reach {} {} q, \nu'$ as required.
	Moreover, the two inductive assumptions also give
	\begin{align}
		\label{eq:fractional:trans1}
		&\forall \x_i \in \T \st \floor {\rho'(\x_i)} = \card {u}_{\checkmark_i} + \floor {\mu'(\x_i)} \cdot \condone {\x_i \in \U}, \textrm{ and } \\
		\label{eq:fractional:trans2}
		&\forall \x_i \in \V \st \floor {\nu'(\x_i)} = \card {v}_{\checkmark_i} + \floor {\rho'(\x_i)} \cdot \condone {\x_i \in \T},
	\end{align}
	and we need to establish $\forall \x_i \in \V \st \floor {\nu'(\x_i)} = \card {w}_{\checkmark_i} + \floor {\mu'(\x_i)} \cdot \condone {\x_i \in \U}$.
	If $\x_i \in \U$, then $\x_i \in \T$ and by \Cref{eq:fractional:trans1,eq:fractional:trans1} we have
	$\floor {\nu'(\x_i)} = \card {v}_{\checkmark_i} + \card {u}_{\checkmark_i} + \floor {\mu'(\x_i)} = \card w_{\checkmark_i} + \floor {\mu'(\x_i)}$.
	If $\x_i \in \T \setminus \U$, by \Cref{eq:fractional:trans1,eq:fractional:trans1} we have
	$\floor {\nu'(\x_i)} = \card {v}_{\checkmark_i} + \card {u}_{\checkmark_i} = \card w_{\checkmark_i}$.
	Finally, if $\x_i \in \V \setminus \T$, then by \Cref{eq:fractional:trans2} we have
	$\floor {\nu'(\x_i)} = \card {v}_{\checkmark_i} = \card w_{\checkmark_i}$,
	where the last equality follows from $\card {u}_{\checkmark_i} = 0$ by Lemma~\ref{lem:fractional:basics}.
	This concludes the transitivity case.

	The last case is push-pop.
	We have $w = \op_\push \cdot u \cdot \op_\pop$
	where $\op_\push = \pushop{\tuple {\gamma, \lambda_\push}} {\psi_\push}$ ($\psi_\push$ was defined in \eqref{eq:fractional:psipush}),
	$\op_\pop = \popop{\tuple {\gamma, \lambda_\push}} {\restrict \psi {\lambda_\push, \lambda_\pop}}$, giving rise to a run of the form
	\begin{align*}
		&\tuple{p, \lambda_\push, \U}, \mu \reach {\op_\push} {} 
			\tuple{r, \lambda_\push, \U}, \mu, (\tuple {\gamma, \lambda_\push}, \rho) \\ &\qquad\reach u {}
				\tuple{s, \lambda_\pop, \V}, \nu, (\tuple {\gamma, \lambda_\push}, \rho + \delta_{\mu\nu}) \reach {\op_\pop} {}
					\tuple{q, \lambda_\pop, \V}, \nu,
	\end{align*}
	where we assume $\forall \x_i \in X \st \rho(\y_i) = \mu(\x_i)$.
	By the inductive assumption applied to the middle run labelled with $u$
	we obtain $\nu' \in R_{\lambda_\pop}(\nu)$ satisfying condition (C) and a run
	$r, \mu' \reach {} {} s, \nu'$.
	It remains to justify the existence of the run
	\begin{align*}
		p, \mu' \reach {\op_\push'} {} 
			r, \mu', (\gamma, \rho') \reach {} {}
				s, \nu', (\gamma, \rho' + \delta_{\mu'\nu'}) \reach {\op_\pop'} {}
					q, \nu',
	\end{align*}
	where $\op_\push' = \pushop \gamma {\psicopy}$,
	$\op_\pop' = \popop \gamma \psi$,
	and $\forall \x_i \in \X \st \rho'(\y_i) = \mu'(\x_i)$.
	Since $\fract{\mu'} = \fract \mu$, $\fract{\nu'} = \fract \nu$,
	we also have $\fract {\delta_{\mu'\nu'}} = \fract {\delta_{\mu\nu}}$.
	Since $(\nu, \rho + \delta_{\mu\nu}) \models \restrict \psi {\lambda_\push, \lambda_\pop}$
	and the latter formula has the same fractional constraints as $\psi$,
	$(\nu', \rho + \delta_{\mu'\nu'})$ satisfies the fractional constraints of $\psi$,
	and so does $(\nu', \rho' + \delta_{\mu'\nu'})$,
	because $\fract {\rho} = \fract {\rho'}$.
	Since $\psi$ does not have integral constraints (it is pop-integer-free),
	in order to have $(\nu', \rho' + \delta_{\mu'\nu'}) \models \psi$ as required,
	we need to show that the latter pair of valuations also satisfies the modular constraints in $\psi$,
	i.e., those of the form
	$$\floor {y_i} - \floor {x_j} \eqv M k.$$
	Since the modular constraint above was resolved to be $\true$ in $\QQ$,
	by definition we have
	$$\underbrace {\lambda_\push (\x_i) + (\lambda_\pop (\x_0) - \lambda_\push (\x_0) + \condone {\fract {\y_i} < \fract {\y_1}} - \condone {\fract {\x_0} < \fract {\y_1}})}_A - \underbrace{\lambda_\pop(\x_j)}_B \eqv M k.$$
	Since $\lambda(\nu') = \lambda_\pop$,
	the expression $B$ above is in the same residue class modulo $M$ as $\floor{\nu'(x_j)}$.
	%
	%The quantity in parenthesis is in the same residue class modulo $M$ as the integral time elapsed $\floor{\delta_{\mu'\nu'}}$ between push and pop.
	%
	Since $\lambda(\mu') = \lambda_\push$,
	by \Cref{fact:fractional} the expression $A$ above is in the same residue class modulo $M$ as $\floor{\rho'(x_i) + \delta_{\mu'\nu'}}$,
	%i.e., $\floor{\mu'_\cp(\y_i) + \delta_{\mu'\nu'}}$.
	%
	Consequently, $(\nu', \rho' + \delta_{\mu'\nu'}) \models \floor {y_i} - \floor {x_j} \eqv M k$, as required.
	This concludes the push-pop case, and the soundness proof.
\end{proof}

For a sequence of actions $\pi \in \Delta^*$,
let $\Resets \pi \subseteq \X$ be the set of clocks which are reset at least once in $\pi$.
\begin{restatable}{lemma}{lemFractionalCompletenessInductive}
	\label{lem:fractional:completeness:inductive}
	For every run $p, \mu \reach \pi {} q, \nu$,
	sets of clocks $\U, \V \subseteq \X$ \st $\Resets \pi \subseteq \X \setminus \U$ and $\U \subseteq \V \subseteq \U \cup \Resets \pi$,
	there exists a run
	\begin{align*}
		\tuple {p, \lambda(\mu), \U}, \mu \reach w {} \tuple {q, \lambda(\nu), \V}, \nu
			\textrm{ \st }
				\forall \x_i \in \V \st \floor{\nu(\x_i)} = \card w_{\checkmark_i} + \floor {\mu(\x_i)} \cdot \condone {\x_i \in \U}.
	\end{align*}
\end{restatable}
%\lemFractionalCompleteness*
\begin{proof}%[Proof of \Cref{lem:fractional:completeness:inductive}]
	We proceed by induction on the characterisation of the reachability relation from Lemma~\ref{lem:characterisation}.
	Let $p, \mu \reach \pi {} q, \nu$ be a run in $\P$
	and let $\Resets \pi \subseteq \X \setminus \U$ and $\U \subseteq \V \subseteq \U \cup \Resets \pi$.

	If $\pi = \readop a$ is a read transition, then $\mu = \nu$, $\Resets \pi = \emptyset$, and thus $\V = \U$.
	By definition we have a run $\tuple {p, \lambda \mu, \U}, \mu \reach w {} \tuple{q, \lambda(\mu), \U}, \nu$ with $w = \readop a$,
	and $\forall \x_i \in \V \st \floor {\nu(\x_i)} = \floor {\mu(\x_i)}$ because $\card w_{\tick i} = 0$.

	If $\pi = \testop \varphi$ is a test transition, then $\mu = \nu$, $\Resets \pi = \emptyset$, $\V = \U$, and ${\mu \models \varphi}$.
	By \Cref{fact:unary}, $\fract{\mu} \models \restrict \varphi {\lambda(\mu)}$,
	and thus we have a run $\tuple {p, \lambda (\mu), \U}, \mu \reach w {} \tuple{q, \lambda(\mu), \U}, \nu$ with $w = \testop {\restrict \varphi {\lambda(\mu)}}$,
	and $\forall \x_i \in \V \st \floor {\nu(\x_i)} = \floor {\mu(\x_i)}$. % because $\card w_{\tick i} = 0$.

	If $\pi = \resetop \Y$ is a reset transition, then $\nu = \mu[\Y \mapsto 0]$ and $\Resets \pi = \Y$.
	Let $\V$ be \st $\U \subseteq \V \subseteq \U \cup \Y$, and take $\Y' = \V \setminus \U \subseteq \Y$; thus $\V = \U \cup \Y'$.
	Since $\lambda(\nu) = \lambda(\mu[\Y \mapsto 0]) = \lambda(\mu)[\Y \mapsto 0]$
	and $\Y \subseteq \X \setminus \U$ by assumption,
	we have a run $\tuple {p, \lambda (\mu), \U}, \mu \reach w {} \tuple{q, \lambda(\mu)[\Y \mapsto 0], \U}, \nu$ with $w = \resetop \Y$.
	Let $\x_i \in \V$. Notice $\card w_{\tick i} = 0$.
	If $\x_i \in \U$, then $\x_i \in \X \setminus \Y$ is not reset and $\nu(\x_i) = \mu(\x_i)$.
	If $\x_i \in \V \setminus \U$, then $\x_i \in \Y' \subseteq \Y$ is reset and $\nu(\x_i) = 0$.
	Thus, $\forall \x_i \in \V \st \floor {\nu(\x_i)} = \floor {\mu(\x_i)} \cdot \condone {\x_i \in \U}$, as required.

	In the time elapse case $\pi = \elapse$, $\nu = \mu + \delta$ for some $\delta \in \Rgeq$, $\Resets \pi = \emptyset$, and $\U = \V$.
	The simulation of a single time elapse of $\P$ requires in general many time elapses of $\QQ$.
	The idea is that each time elapse of $\QQ$ advances the clocks with maximal fractional value to the next integral value.
	Formally, we decompose $\delta$ as
	\begin{align*}
		\delta = \delta_1 + \cdots + \delta_m, \quad \delta_1, \dots, \delta_m \in \Rgeq \cap (0, 1),
	\end{align*}
	where, for every $0 \leq j \leq m$,
	$\mu_j = \mu + \delta_1 + \cdots + \delta_j$
	and, for $0 \leq j < m$, $\delta_{j+1} = 1 - \max_{\x_i \in \X}(\fract{\mu_j})$;
	thus, $\mu_0 = \mu$ and $\mu_m = \mu + \delta = \nu$.
	Consequently, the set of clocks that become integral when going from phase $j$ to phase $j+1$ is
	$\Y_j = \setof{\x_i \in \X} {\fract{\mu_j(\x_i)} = \max_{\x_i \in \X}(\fract{\mu_j})}$.
	Thus $\forall \x_i \in \Y_j \st \fract{\mu_{j+1}(\x_i)} = 0$ and $\floor{\mu_{j+1}(\x_i)} = \floor{\mu_j(\x_i)} + 1$.
	Let $\ops_{j+1} = \ops_{\Y_j, \U}$, where the latter is defined in \eqref{eq:fractional:ops},
	and $\lambda_{j+1} = \lambda_j[\Y_j \mapsto \Y_j + 1]$;
	thus, $\lambda_j = \lambda(\mu_j)$.
	With these definitions in place, take $w = \readop\varepsilon \cdot \ops_1 \cdots \ops_m \cdot \readop\varepsilon$, yielding a run
	\begin{align*}
		&\tuple {p, \lambda(\mu), \U}, \mu \reach {\readop\varepsilon} {} \\
		&	\quad \tuple {p, \lambda_0, \U, q}, \mu_0 \reach {\ops_1} {}
				\tuple {p, \lambda_1, \U, q}, \mu_1 \reach {\ops_2} {} \cdots \reach {\ops_m} {}
					\tuple {p, \lambda_m, \U, q}, \mu_m \\ & \quad \quad \reach {\readop\varepsilon} {}
						\tuple {q, \lambda(\nu), \U}, \nu.
	\end{align*}
	For $\x_i \in \V$,
	since one symbol $\tick i$ is read by $\ops_j$ whenever $\floor {x_i}$ increases by one, % (and this happens precisely when $\x_i \in \Y_j$),
	we have $\floor{\mu_m(\x_i)} = \card w_{\tick i} + \floor{\mu_0(\x_i)}$, as required.
	This concludes the time elapse case.

	In the transitivity case $\pi = \sigma \cdot \tau$ and we have two runs
	\begin{align*}
		p, \mu \reach \sigma {} r, \rho \reach \tau {} q, \nu.
	\end{align*}
	Let $\T = \U \cup (\Resets \sigma \setminus \Resets \tau)$ 
	contain all clocks which are reset in the first part of the run $\sigma$,
	but not in the second part $\tau$;
	thus, $\U \subseteq \T \subseteq \U \cup \Resets \sigma$
	and $\U \subseteq \X \setminus \Resets \tau$.
	By the inductive assumption applied twice we obtain two runs
	\begin{align*}
		\tuple {p, \lambda(\mu), \U}, \mu \reach u {} \tuple {r, \lambda(\rho), \T}, \rho \reach v {} \tuple {q, \lambda(\nu), \V}, \nu
	\end{align*}
	\st $\forall \x_i \in \T \st \floor{\rho(\x_i)} = \card u_{\checkmark_i} + \floor {\mu(\x_i)} \cdot \condone {\x_i \in \U}$
	and $\forall \x_i \in \V \st \floor{\nu(\x_i)} = \card v_{\checkmark_i} + \floor {\rho(\x_i)} \cdot \condone {\x_i \in \T}$.
	Therefore, by taking $w = uv$ we obtain $\tuple {p, \lambda(\mu), \U}, \mu \reach w {} \tuple {q, \lambda(\nu), \V}, \nu$.
	Let $\x_i \in \V$, and we need to show $\floor {\nu(\x_i)} = \card w_{\tick i} + \floor {\mu(\x_i)} \cdot \condone {\x_i \in \U}$.
	There are three cases to consider.
	Recall that $\U \subseteq \T \subseteq \V$.
	If $\x_i \in \U$, then $\floor{\nu(\x_i)} = \card v_{\checkmark_i} + \card u_{\checkmark_i} + \floor {\mu(\x_i)} = \card w_{\checkmark_i} + \floor {\mu(\x_i)}$;
	if $\x_i \in \T \setminus \U$, then $\floor{\nu(\x_i)} = \card v_{\checkmark_i} + \card u_{\checkmark_i} = \card w_{\checkmark_i}$;
	if $\x_i \in \V \setminus \T$, then $\floor{\nu(\x_i)} = \card v_{\checkmark_i} = \card w_{\checkmark_i}$,
	where the last equality follows from the fact that since $\x_i \in \X \setminus \T$,
	then $\card u_{\checkmark_i} = 0$ by \Cref{lem:fractional:basics}.
	This concludes the transitivity case.

	In the push-pop case, $\pi = \pushop \gamma \psicopy \cdot \sigma \cdot \popop \gamma \psi$
	and we have a run
	\begin{align*}
		p, \mu \reach {\pushop \gamma \psicopy} {}
			r, \mu, (\gamma, \rho) \reach \sigma {}
				s, \nu, (\gamma, \rho + \delta_{\mu\nu}) \reach {\popop \gamma \psi} {}
					q, \nu,
	\end{align*}
	where for all $\x_i \in \X$, $\rho(\y_i) = \mu(\x_i)$.
	By induction assumption, we have a run
	$\tuple{r, \lambda(\mu), \U}, \mu \reach u {} \tuple{s, \lambda(\nu), \V}$.
	Let the corresponding push/pop operations in $\QQ$ be
	$\op_\push = \pushop {\tuple {\gamma,\lambda(\mu)}} {\psi_\push}$
	and $\op_\pop = \popop {\tuple {\gamma,\lambda(\mu)}} {\restrict \psi {\lambda(\mu), \lambda(\nu)}}$,
	where $\psi_\push$ is defined in \eqref{eq:fractional:psipush}.
	We show that the following run exists in $\QQ$, as required:
	\begin{align*}
		&\tuple{p, \lambda_\push, \U}, \mu \reach {\op_\push} {}
			\tuple{r, \lambda_\push, \U}, \mu, (\gamma, \rho) \\ &\qquad \reach u {}
				\tuple{s, \lambda_\pop, \V}, \nu, (\gamma, \rho + \delta_{\mu\nu}) \reach {\op_\pop} {}
					\tuple{q, \lambda_\pop, \V}, \nu,
	\end{align*}
	where $\lambda_\push = \lambda(\mu)$, $\lambda_\pop = \lambda_\pop$, and we assume $\fract{\rho(\y_1)} = 0$.
	This amounts to showing that the pop constraint $\restrict \psi {\lambda_\push, \lambda_\pop}$ is satisfied in $\QQ$,
	i.e., $(\nu, \rho + \delta_{\mu\nu}) \models \restrict \psi {\lambda_\push, \lambda_\pop}$.
	Since $(\nu, \rho + \delta_{\mu\nu}) \models \psi$,
	the fractional constraints in $\restrict \psi {\lambda_\push, \lambda_\pop}$ are satisfied
	because they are the same as $\psi$'s fractional constraints.
	It remains to show that $\psi$'s modular constraints of the form $\floor {\y_i} - \floor {\x_j} \eqv M k$
	evaluate to $\true$ in $\psi_{\lambda_\push, \lambda_\pop}$
	when $\floor {\x_j}$ is replaced by $\lambda_\pop(\x_j)$ and $\floor {\y_i}$ by
	$$\lambda_\push (\x_i) + (\lambda_\pop (\x_0) - \lambda_\push (\x_0) + \condone {\fract {\y_i} < \fract {\y_1}} - \condone {\fract {\x_0} < \fract {\y_1}}).$$
	By assumption, $\floor {\rho(\y_i)} - \floor {\nu(\x_j)} \eqv M k$,
	and thus by \Cref{fact:fractional},
	$$\left(\floor {\mu(\x_i)} + \floor {\nu(\x_0)} -  \floor {\mu(\x_0)} +	\condone {\fract{\rho(\y_i)} < \fract {\rho(\y_1)}} - \condone {\fract {\nu(\x_0)} < \fract {\rho(\y_1)}}\right) -  \floor {\nu(\x_j)} \eqv M k.$$
	The result follows by replacing the integral values above with their residual modulo $M$ as given by $\lambda_\push$ and $\lambda_\pop$.
	This concludes the push-pop case, and the proof of the lemma.
\end{proof}

\subsection{Proofs for Sec.~\ref{sec:fractional:TPDA}}

\lemCFGsoundness
\begin{proof}
  We proceed by induction on the size of derivation trees showing $w \in L(p, \varphi, q)$.
  For the base case, $w = \delta \in \Delta$ is derived by rule \eqref{eq:CFG:base}.
  Since $(\mu, \nu) \models \varphi$, by \Cref{fact:onestep:CDR},
  $\mu \freach \delta {pq} \nu$ holds, as required.
  
  For the inductive step, there are two cases to consider. %let $w$ be derived with a derivation of length $l > 1$.
  In the first case,
  $w \in L(p, \varphi, q)$ is derived using 
  a production of the form $\tuple{p, \varphi, q} \from \tuple {p, \psi, r} \cdot \tuple {r, \xi, r}$ \eqref{eq:CFG:transitivity},
  where $\varphi \equiv \psi \circ \xi$,
  and thus there are words $u, v \in \Sigma^*$
  \st $w = uv$, $u \in L(p, \psi, r)$, and $v \in L(r, \xi, q)$.
  Since $(\mu, \nu) \models \varphi$, by definition \eqref{eq:CDR:composition}
  there exists a clock valuation $\rho \in \Rgeq^\X$ \st $(\mu, \rho) \models \psi$ and $(\rho, \nu) \models \xi$.
  By using the induction assumption twice,
  $\mu \freach u {pr} \rho$ and $\rho \freach v {rq} \nu$,
  and thus $\mu \freach w {pq} \nu$ by Fact~\ref{fact:freach:transitive}.
  
  In the second case, 
  $w = \delta_\push \cdot u \cdot \delta_\pop \in L(p, \varphi, q)$ is derived using 
  a production of the form $\tuple{p, \varphi, q} \from \tuple {r, \psi, s}$ \eqref{eq:CFG:push-pop}
  for transitions $\delta_\push = \trule p {\pushop \alpha {\psi_\push}} r$
  and $\delta_\pop = \trule s {\popop \alpha {\psi_\pop}} q$,
  where $\varphi$ is defined in \eqref{eq:CFG:push-pop}.
  Since $(\mu, \nu) \models \varphi$, by the definition of $\varphi$
  there exist stack clock valuations $\mu_\ZZ, \nu_\ZZ \in \Rgeq^\ZZ$ \st
  \begin{enumerate}[a)]
    \item $(\mu, \nu) \models \psi$,
    \item $(\mu, \mu_\ZZ) \models \psi_\push$,
    \item $(\nu, \nu_\ZZ) \models \psi_\pop$, and
    \item for every $i$, $\fract{\nu_\ZZ(\z_i)} = \fract{\mu_\ZZ(\z_i) + \nu(\x_0) - \mu(\x_0)}$.
  \end{enumerate}
  Point d) implies that we can think of $\nu_\ZZ$ to be of the form
  $\nu_\ZZ = \mu_\ZZ + \delta_{\mu\nu}$, where $\delta_{\mu\nu} = \nu(\x_0) - \mu(\x_0)$
  is the time elapsed between push and pop.
  The inductive assumption applied to point a) and $u \in L(r, \psi, s)$ yields a run $\mu \freach w {rs} \nu$.
  By the definition of fractional reachability,
  there are clock valuations $\tilde\mu, \tilde\nu \in \Rgeq^\X$ \st $\tilde\mu \reach w {rs} \tilde\nu$.
  From points b) and c) we obtain $(\tilde\mu, \mu_\ZZ) \models \psi_\push$ and, resp., $(\tilde\nu, \nu_\ZZ) \models \psi_\pop$,
  since the push and pop constraints are fractional.
  By equation \eqref{eq:reachrel:F}, $\tilde\mu \reach w {pq} \tilde\nu$,
  yielding $\mu \freach w {pq} \nu$ as required.
\end{proof}

\lemCFGcompleteness
\begin{proof}%[Proof (of \Cref{lem:CFG:completeness})]
  We proceed by induction on derivations establishing $\mu \reach w {pq} \nu$.
  In the base case, $w = \delta = \tuple {p, \op, q} \in \Delta$ is a single transition
  and $\mu \reach w {pq} \nu$ is obtained by one of the rules \Cref{eq:reachrel:A,eq:reachrel:B,eq:reachrel:C,eq:reachrel:D}.
  By the definition of fractional reachability we have $\mu \freach \delta {pq} \nu$.
  By \Cref{fact:onestep:CDR}, $(\mu, \nu) \models \varphi_\op$,
  and $w \in L(p, \varphi, q)$ holds by rule \eqref{eq:CFG:base}.
  
  For the inductive step, there are two cases to consider.
  In the first case, $\mu \reach w {pq} \nu$ is obtained by applying a transitivity step according to \eqref{eq:reachrel:E}.
  There exist words $u, v \in \Delta^*$,
  an intermediate clock valuation $\rho \in \Rgeq^\X$,
  and a control location $r \in \L$ \st
  $w = uv$ and $\mu \reach u {pr} \rho \reach v {rq} \nu$.
  %Let $\varphi_{\mu\rho}$ and $\varphi_{\rho\nu}$ be the corresponding characteristic \CDR's.
  %Since $(\mu, \rho) \models \varphi_{\mu\rho}$ and $(\rho, \nu) \models \varphi_{\rho\nu}$ by definition,
  By the inductive hypothesis there are \CDR's $\psi$ and $\xi$
  \st $u \in L(p, \psi, r)$, $v \in L(r, \xi, q)$, $(\mu, \rho) \models \psi$, and $(\rho, \nu) \models \xi$.
  Take $\varphi \equiv \psi \circ \xi$, and thus $(\mu, \nu) \models \psi \circ \xi$ as witnessed by $\rho$.
  By production \eqref{eq:CFG:transitivity},
  $w \in L(p, \varphi, q)$, as required.
  
  In the second case, $\mu \reach w {pq} \nu$ is obtained by applying a push-pop step according to \eqref{eq:reachrel:F}:
  There exist control locations $r, s \in \L$,
  push $\delta_\push = \trule p {\pushop \alpha {\psi_\push}} r$
  and pop $\delta_\pop = \trule s {\popop \alpha {\psi_\pop}} q \in \Delta$ transitions,
  and an initial stack clock valuation $\mu_\ZZ \in \Rgeq^\ZZ$
  \st $w = \delta_\push \cdot u \cdot \delta_\pop$ for some $u \in \Delta^*$,
  $(\mu, \mu_\ZZ) \models \psi_\push$,
  $(\nu, \mu_\ZZ + \nu(\x_0) - \mu(\x_0)) \models \psi_\pop$,
  and $\mu \reach u {rs} \nu$.
  By the inductive assumption there exists a \CDR $\psi$ \st $w \in L(r, \psi, s)$ and $(\mu, \nu) \models \psi$.
  Thus $(\mu, \nu) \models \varphi$ holds for $\varphi(\bar x, \bar x')$ from \eqref{eq:CFG:push-pop}
  (and thus for the actual unique \CDR equivalent to it formally used in the grammar),
  as witnessed by $\mu_\ZZ$ for variables $\bar z$ and $\mu_\ZZ + \nu(\x_0) - \mu(\x_0)$ for $\bar z'$.
  %
%  By taking $\nu_\ZZ \equiv \mu_\ZZ + \nu(\x_0) - \mu(\x_0)$,
 % clearly $(\mu, \mu_\ZZ) \models \psi_\push$ and $(\nu, \nu_\ZZ) \models \psi_\pop$.
  %
  By applying production \eqref{eq:CFG:push-pop} to nonterminal $\tuple{r, \psi, s}$,
  we obtain $w \in L(p, \varphi, q)$, as required.
\end{proof}

\end{document}